%% file: main.tex
\documentclass[letterpaper,twocolumn,10pt]{article}
\usepackage{usenix2019_v3}

\input{packages}
\input{symbols_maths}

\input{symbols_env_usenix}

\input{symbols_ckt}

\input{symbols_mpc}

\input{symbols_blaze}

\input{symbols_3pc_SWIFT}

\input{symbols_4pc_SWIFT}


\input{comments}

\begin{document}
\title{\Large \bf SWIFT: Super-fast and Robust Privacy-Preserving Machine Learning\thanks{This article is the full and extended version of an article to appear in USENIX Security’21.}}
\author{
	{\rm Nishat Koti$^{\star}$, Mahak Pancholi$^{\star}$, Arpita Patra$^{\star}$, Ajith Suresh$^{\star}$}\\
$^{\star}$Department of Computer Science and Automation, Indian Institute of Science, Bangalore, India\\
\{kotis, mahakp, arpita, ajith\}@iisc.ac.in
} 

\maketitle

\begin{abstract}
Performing machine learning (ML) computation on private data while maintaining data privacy, aka Privacy-preserving Machine Learning~(PPML), is an emergent field of research. Recently, PPML has seen a visible shift towards the adoption of the Secure Outsourced Computation~(SOC) paradigm due to the heavy computation that it entails. In the SOC paradigm, computation is outsourced to a set of powerful and specially equipped servers that provide service on a pay-per-use basis. In this work, we propose SWIFT, a {\em robust} PPML framework for a range of ML algorithms in SOC setting, that guarantees output delivery to the users irrespective of any adversarial behaviour. Robustness, a highly desirable feature,   evokes user participation without the fear of denial of service.   

At the heart of our framework lies a highly-efficient, maliciously-secure, three-party computation (3PC) over rings that provides guaranteed output delivery (GOD) in the honest-majority setting.  To the best of our knowledge, SWIFT is the first robust and efficient  PPML framework in the 3PC setting. SWIFT is as fast as (and is strictly better in some cases than) the best-known 3PC framework BLAZE~(Patra et al. NDSS'20), which only achieves fairness. We extend our 3PC framework for four parties (4PC). In this regime,  SWIFT is as fast as the best known {\em fair} 4PC framework Trident~(Chaudhari et al. NDSS'20) and twice faster than the best-known {\em robust} 4PC framework FLASH~(Byali et al. PETS'20). 

We demonstrate our framework's practical relevance by benchmarking popular ML algorithms such as Logistic Regression and deep Neural Networks such as VGG16 and LeNet, both over a 64-bit ring in a WAN setting. For deep NN, our results testify to our claims that we provide improved security guarantee while incurring no additional overhead for 3PC and obtaining $2 \times$ improvement for 4PC.

\end{abstract}


\section{Introduction}
\label{sec:Intro}
\input{Main_1_Introduction}


\section{Preliminaries}
\label{sec:Prelims}
\input{Main_2_Preliminaries}

\section{Robust 3PC and PPML}
\label{sec:3PC}
\input{Main_3_3PC}
\section{Robust 4PC and PPML}
\label{sec:4PC}
\input{Main_4PC}

\section{Applications and Benchmarking}
\label{sec:Implementation}
\input{Main_5_Bench}


\section{Conclusion}
\label{sec:Conclusions}
\input{Main_Conclusion.tex}

\section*{Acknowledgements}
We thank our shepherd Guevara Noubir, and anonymous reviewers for their valuable feedback.

Nishat Koti would like to acknowledge financial support from Cisco PhD Fellowship 2020. Mahak Pancholi would like to acknowledge financial support from Cisco MTech Fellowship 2020. Arpita Patra would like to acknowledge financial support from SERB MATRICS (Theoretical Sciences) Grant 2020 and Google India AI/ML Research Award 2020. Ajith Suresh would like to acknowledge financial support from Google PhD Fellowship 2019. The authors would also like to acknowledge the financial support from Google Cloud to perform the benchmarking.

\bibliographystyle{abbrv}
\bibliography{main}

\appendix
\section{Preliminaries}
\label{app:Prelims}

\input{Appendix_Prelims}
\section{Instantiating $\FDotPPre$}
\label{appsec:piDotP} 
\input{Appendix_DotPPre}
\section{Security Analysis of Our Protocols}
\label{app:proofs}
\input{Appendix_Security}

\end{document}

%% file: packages.tex
\usepackage{amssymb}
\usepackage{amsthm}
\usepackage{booktabs}
\usepackage{caption}
\usepackage{cite}
\usepackage{color}
\usepackage{colortbl}
\usepackage{cuted}
\usepackage{dblfloatfix} 
\usepackage{enumerate}
\let\labelindent\relax 
\usepackage{enumitem}
\usepackage{floatrow}
\usepackage{graphicx}
\usepackage{hyperref}
\hypersetup{
	colorlinks   = true,
	citecolor    = blue
}
\usepackage{makecell}
\usepackage{multirow}
\usepackage{multicol}
\usepackage{pgf}
\usepackage{pgfplots}
\usepackage{pgfplotstable}
\usepackage{ragged2e}
\usepackage{stmaryrd}
\usepackage{subcaption}
\usepackage{tabu}
\usepackage{tabularx}
\usepackage{threeparttable}
\usepackage{tikz}
\usepackage{todonotes}
\usepackage{url}
\usepackage{verbatim}
\usepackage{wrapfig}
\usepackage{xcolor}
\usepackage[framemethod=tikz]{mdframed} 
\usepackage{algpseudocode}
\usepackage{nccmath}

%% file: symbols_maths.tex
\newtheorem{theorem}{Theorem}[section]

\newtheorem{lemma}[theorem]{Lemma}

\theoremstyle{definition}

\newtheorem{notation}[theorem]{Notation}

\newcommand{\xor}{\oplus}

\newcommand{\band}{\odot}

\newcommand{\Order}{\mathcal{O}}
\newcommand{\BigO}[1]{\ensuremath{\operatorname{O}\left(#1\right)}}

\newcommand{\Z}[1]{\ensuremath{\mathbb{Z}}_{2^{#1}}}

\newcommand{\bitb}{\ensuremath{\mathsf{b}}} 
\newcommand{\arval}[1]{\ensuremath{#1^{\sf R}}} 

\newcommand{\nf}{\ensuremath{\mathsf{n}}} 

\newcommand{\vecd}{\ensuremath{\vec{\mathbf{d}}}}
\newcommand{\vece}{\ensuremath{\vec{\mathbf{e}}}}

%% file: symbols_env_usenix.tex
\setlist[description]{style=unboxed,leftmargin=0cm}

\newenvironment{myitemize}{
	\begin{list}{{$\bullet$}}{
			\setlength\partopsep{0pt}
			\setlength\parskip{0pt}
			\setlength\parsep{0pt}
			\setlength\topsep{2pt}
			\setlength\itemsep{3pt}
			\setlength{\itemindent}{4pt}
			\setlength{\leftmargin}{2pt}
		}
	}{
	\end{list}
}

\newenvironment{newitemize}{
	\begin{list}{{$\bullet$}}{
			\setlength\partopsep{0pt}
			\setlength\parskip{0pt}
			\setlength\parsep{0pt}
			\setlength\topsep{2pt}
			\setlength\itemsep{4pt}
			\setlength{\itemindent}{1pt}
			\setlength{\leftmargin}{9pt}
		}
	}{
	\end{list}
}

\newenvironment{inneritemize}{
	\begin{list}{{$\bullet$}}{
			\setlength\partopsep{0pt}
			\setlength\parskip{0pt}
			\setlength\parsep{0pt}
			\setlength\topsep{2pt}
			\setlength\itemsep{2pt}
			\setlength{\itemindent}{2pt}
			\setlength{\leftmargin}{9pt}
		}
	}{
	\end{list}
}

\newenvironment{mylist}{
	\begin{list}{{$\bullet$}}{
			\setlength\partopsep{0pt}
			\setlength\parskip{0pt}
			\setlength\parsep{0pt}
			\setlength\topsep{2pt}
			\setlength\itemsep{1pt}
			\setlength{\itemindent}{0pt}
			\setlength{\leftmargin}{9pt}
		}
	}{
	\end{list}
}

\newlist{myenumlist}{enumerate}{2}
\setlist[myenumlist]{leftmargin=*,label={\arabic*)}}

\newcounter{itemcount}

\newcommand{\tabref}[1]{Table~\protect\ref{tab:#1}}

\newcommand{\figref}[1]{Fig.~\ref{fig:#1}}
\newcommand{\figlab}[1]{\label{fig:#1}}

\newenvironment{boxfig*}[2]{
	\begin{figure*}[h!]		
		\fontsize{5}{5}\selectfont
		\newcommand{\FigCaption}{#1}
		\newcommand{\FigLabel}{#2}
		\vspace{-.05cm}
		\begin{center}
			\begin{small}			 
				\begin{adjustbox}{max width=\textwidth}
					\begin{tabular}{@{}|@{~~}l@{~~}|@{}}
						\hline
						\rule[-1ex]{0pt}{1ex}\begin{minipage}[b]{.95\linewidth}
							\vspace{1ex}	
						}{%
						\end{minipage}\\
						\hline
					\end{tabular}	
				\end{adjustbox}		
			\end{small}
			\vspace{-0.1cm}
			\caption{\FigCaption}
			\figlab{\FigLabel}
		\end{center}
		\vspace{-.38cm}
	\end{figure*}
}

\newenvironment{myboxfig*}[2]{
	\begin{figure*}[!htb]		
		\fontsize{5}{5}\selectfont
		\newcommand{\FigCaption}{#1}
		\newcommand{\FigLabel}{#2}
		\vspace{-.10cm}
		\begin{center}
			\caption{\FigCaption}
			\begin{small}			 
				\begin{adjustbox}{max width=\textwidth}
					\begin{tabular}{@{}|@{~~}l@{~~}|@{}}
						\hline
						\rule[-1ex]{0pt}{1ex}\begin{minipage}[b]{.95\linewidth}
							\vspace{1ex}	
						}{%
						\end{minipage}\\
						\hline
					\end{tabular}	
				\end{adjustbox}		
			\end{small}
			\vspace{-0.25cm}
			\figlab{\FigLabel}
		\end{center}
		\vspace{-.38cm}
	\end{figure*}
}

\newcommand{\boxref}[1]{Fig. \ref{#1}}

\RequirePackage{mdframed}

\newenvironment{titlebox}[5]
{\mdfsetup{
		style=#2,
		innertopmargin=1.1\baselineskip,
		skipabove={\dimexpr0.2\baselineskip+\topskip\relax},
		skipbelow={1em},needspace=3\baselineskip,
		singleextra={\node[#3,right=10pt,overlay] at (P-|O){~{\sffamily\bfseries #1 }};},%
		firstextra={\node[#3,right=10pt,overlay] at (P-|O) {~{\sffamily\bfseries #1 }};},
		frametitleaboveskip=9em,
		innerrightmargin=5pt
	}
	\newcommand{\TitleCaption}{#4}
	\newcommand{\TitleLabel}{#5}
	\begin{mdframed}[font=\small]
		\setlist[itemize]{leftmargin=13pt}\setlist[enumerate]{leftmargin=13pt}\raggedright%
	}
	{\end{mdframed}
	\vspace{-1.8em}
	{\captionof{figure}{\small \TitleCaption}\label{\TitleLabel}}
	\medskip
}

\tikzstyle{normal} = [thick, fill=white, text=black, draw, rounded corners, rectangle, minimum height=.7cm, inner sep=3pt]
\tikzstyle{gray} = [thick, fill=gray!90, text=white, rounded corners, rectangle, minimum height=.7cm, inner sep=3pt]

\mdfdefinestyle{commonbox}{%
	align=center, middlelinewidth=1.1pt,userdefinedwidth=\linewidth
	innerrightmargin=0pt,innerleftmargin=2pt,innertopmargin=5pt,
	splittopskip=7pt,splitbottomskip=3pt
}

\mdfdefinestyle{roundbox}{style=commonbox,roundcorner=5pt,userdefinedwidth=\linewidth}

\newenvironment{systembox}[3]
{\begin{titlebox}{Functionality \normalfont #1}{roundbox}{normal}{#2}{#3}}
	{\end{titlebox}}

\newenvironment{protocolbox}[3]
{\begin{titlebox}{Protocol \normalfont #1}{commonbox}{normal}{#2}{#3}}
	{\end{titlebox}}

\newenvironment{simulatorbox}[3]
{\begin{titlebox}{Simulator \normalfont #1}{commonbox}{normal}{#2}{#3}}
	{\end{titlebox}}

\newenvironment{splittitlebox}[5]
{\mdfsetup{
		style=#2,
		innertopmargin=1.1\baselineskip,
		skipabove={\dimexpr0.2\baselineskip+\topskip\relax},
		skipbelow={1em},needspace=3\baselineskip,
		singleextra={\node[#3,right=10pt,overlay] at (P-|O){~{\sffamily\bfseries #1 }};},%
		firstextra={\node[#3,right=10pt,overlay] at (P-|O) {~{\sffamily\bfseries #1 }};},
		frametitleaboveskip=9em,
		innerrightmargin=5pt
	}
	\newcommand{\TitleCaption}{#4}
	\newcommand{\TitleLabel}{#5}
	\begin{mdframed}[font=\small]
		\setlist[itemize]{leftmargin=13pt}\setlist[enumerate]{leftmargin=13pt}\raggedright%
	}
	{\end{mdframed}
	\vspace{-0.8em}
	{\captionof{figure}{\small \TitleCaption}\label{\TitleLabel}}
	\medskip
}

\mdfdefinestyle{commonsplitbox}{%
	align=center, middlelinewidth=1.1pt,userdefinedwidth=\linewidth
	innerrightmargin=0pt,innerleftmargin=2pt,innertopmargin=5pt,
	splittopskip=7pt,splitbottomskip=3pt
}

\newenvironment{simulatorsplitbox}[3]
{\begin{splittitlebox}{Simulator \normalfont #1}{commonsplitbox}{normal}{#2}{#3}}
	{\end{splittitlebox}}

\newenvironment{systembox*}[3]
{\begin{strip}
\vspace{\baselineskip}\begin{titlebox}{Functionality \normalfont #1}{roundbox}{normal}{#2}{#3}}
	{\end{titlebox}
\end{strip}}

\newenvironment{gsystembox*}[3]
{\begin{strip}
\vspace{\baselineskip}\begin{titlebox}{Global Functionality \normalfont #1}{roundbox}{normal}{#2}{#3}}
	{\end{titlebox}
\end{strip}}

\newenvironment{protocolbox*}[3]
{\begin{strip}
\begin{titlebox}{Protocol \normalfont #1}{commonbox}{normal}{#2}{#3}}
	{\end{titlebox}
\end{strip}}

\newenvironment{algobox*}[3]
{\begin{strip}
\begin{titlebox}{Algorithm \normalfont #1}{commonbox}{normal}{#2}{#3}}
	{\end{titlebox}
\end{strip}}

\newenvironment{reductionbox*}[3]
{\begin{strip}
\begin{titlebox}{Reduction \normalfont #1}{commonbox}{normal}{#2}{#3}}
	{\end{titlebox}
\end{strip}}

\newenvironment{gamebox*}[3]
{\begin{strip}
\begin{titlebox}{Game \normalfont #1}{commonbox}{gray}{#2}{#3}}
	{\end{titlebox}
\end{strip}}

\newenvironment{simulatorbox*}[3]
{\begin{strip}
\begin{titlebox}{Simulator \normalfont #1}{commonbox}{normal}{#2}{#3}}
	{\end{titlebox}
\end{strip}}

\mdfdefinestyle{mycommonbox}{%
	align=center, middlelinewidth=1.1pt,userdefinedwidth=\linewidth
	innerrightmargin=0pt,innerleftmargin=2pt,innertopmargin=3pt,
	splittopskip=7pt,splitbottomskip=3pt
}
\mdfdefinestyle{myroundbox}{style=mycommonbox,roundcorner=5pt,userdefinedwidth=\linewidth}

\newenvironment{titlebox*}[5]
{\mdfsetup{
		style=#2,
		innertopmargin=0.3\baselineskip,
		skipabove={0.4em},
		skipbelow={1em},needspace=3\baselineskip,
		frametitleaboveskip=5em,
		innerrightmargin=5pt
	}
	\newcommand{\TitleCaption}{#4}
	\newcommand{\TitleLabel}{#5}
	\begin{mdframed}[font=\small]
		\setlist[itemize]{leftmargin=13pt}\setlist[enumerate]{leftmargin=13pt}\raggedright%
	}
	{\end{mdframed}
	\vspace{-2em}
	{\captionof{figure}{\normalfont \TitleCaption}\label{\TitleLabel}}
	\vspace{-2mm}
}

\newenvironment{mysystembox*}[3]
{\begin{strip}
		\vspace{\baselineskip}\begin{titlebox*}{Functionality \normalfont #1}{myroundbox}{normal}{#2}{#3}}
		{\end{titlebox*}
\end{strip}}

\newenvironment{mygsystembox*}[3]
{\begin{strip}
		\vspace{\baselineskip}\begin{titlebox*}{Global Functionality \normalfont #1}{myroundbox}{normal}{#2}{#3}}
		{\end{titlebox*}
\end{strip}}

\newenvironment{myprotocolbox*}[3]
{\begin{strip}
		\begin{titlebox*}{Protocol \normalfont #1}{mycommonbox}{normal}{#2}{#3}}
		{\end{titlebox*}
\end{strip}}

\newenvironment{myalgobox*}[3]
{\begin{strip}
		\begin{titlebox*}{Algorithm \normalfont #1}{mycommonbox}{normal}{#2}{#3}}
		{\end{titlebox*}
\end{strip}}

\newenvironment{myreductionbox*}[3]
{\begin{strip}
		\begin{titlebox*}{Reduction \normalfont #1}{mycommonbox}{normal}{#2}{#3}}
		{\end{titlebox*}
\end{strip}}

\newenvironment{mygamebox*}[3]
{\begin{strip}
		\begin{titlebox*}{Game \normalfont #1}{mycommonbox}{gray}{#2}{#3}}
		{\end{titlebox*}
\end{strip}}

\newenvironment{mysimulatorbox*}[3]
{\begin{strip}
		\begin{titlebox*}{Simulator \normalfont #1}{mycommonbox}{normal}{#2}{#3}}
		{\end{titlebox*}
\end{strip}}

\newenvironment{mytbox}[5]
{\mdfsetup{
		style=#2,
		innertopmargin=1.0\baselineskip,
		skipabove={\dimexpr0.2\baselineskip+\topskip\relax},
		skipbelow={1em},needspace=3\baselineskip,
		singleextra={\node[#3,right=10pt,overlay] at (P-|O){~{\sffamily\bfseries #1 }};},%
		firstextra={\node[#3,right=10pt,overlay] at (P-|O) {~{\sffamily\bfseries #1 }};},
		frametitleaboveskip=9em,
		innerrightmargin=5pt
	}
	\newcommand{\TitleCaption}{#4}
	\newcommand{\TitleLabel}{#5}
	\begin{mdframed}[font=\small]
		\setlist[itemize]{leftmargin=13pt}\setlist[enumerate]{leftmargin=13pt}\raggedright%
	}
	{\end{mdframed}
	\vspace{-1.8em}
	{\captionof{figure}{\small \TitleCaption}\label{\TitleLabel}}
	\vspace{2mm}
}

\newenvironment{mytsplitbox}[5]
{\mdfsetup{
		style=#2,
		innertopmargin=1.8\baselineskip,
		skipabove={\dimexpr0.2\baselineskip+\topskip\relax},
		skipbelow={1em},needspace=3\baselineskip,
		singleextra={\node[#3,right=10pt,overlay] at (P-|O){~{\sffamily\bfseries #1 }};},%
		firstextra={\node[#3,right=10pt,overlay] at (P-|O) {~{\sffamily\bfseries #1 }};},
		frametitleaboveskip=9em,
		innerrightmargin=5pt
	}
	\newcommand{\TitleCaption}{#4}
	\newcommand{\TitleLabel}{#5}
	\begin{mdframed}[font=\small]
		\setlist[itemize]{leftmargin=13pt}\setlist[enumerate]{leftmargin=13pt}\raggedright%
	}
	{\end{mdframed}
	\vspace{-0.8em}
	{\captionof{figure}{\small \TitleCaption}\label{\TitleLabel}}
	\vspace{2mm}
}

\newenvironment{mypbox}[3]
{\begin{mytbox}{Protocol \normalfont #1}{commonbox}{normal}{#2}{#3}}
	{\end{mytbox}}

\newcommand{\algoHead}[1]{\vspace{0.2em} \underline{\textbf{#1}} \vspace{0.3em}}

\makeatletter
\algnewcommand{\ExtendedState}[1]{\State
	\parbox[t]{\dimexpr\linewidth-\ALG@thistlm}{\hangindent=\algorithmicindent\strut\hangafter=3#1\strut}}
\makeatother

\algnewcommand\algorithmicinput{\textbf{Input:}}
\algnewcommand\Input{\item[\algorithmicinput]}

\algrenewcommand{\algorithmiccomment}[1]{{\color{gray}// #1}}



%% file: symbols_ckt.tex
\newcommand{\ckt}{\ensuremath{\mathsf{ckt}}}

\newcommand{\wx}{\mathsf{x}}
\newcommand{\wy}{\mathsf{y}}
\newcommand{\wz}{\mathsf{z}}


%% file: symbols_mpc.tex
\newcommand{\negl}{\ensuremath{\mathsf{negl}}}
\newcommand{\csec}{\kappa}

\newcommand{\abort}{\ensuremath{\mathtt{abort}}}

\newcommand{\flag}{\ensuremath{\mathsf{flag}}}

\newcommand{\Partyset}{\ensuremath{\mathcal{P}}}
\newcommand{\User}{\ensuremath{\mathsf{U}}}
\newcommand{\TTP}{\ensuremath{\mathsf{TTP}}}
\newcommand{\SOC}{\ensuremath{\mathsf{SOC}}}
\newcommand{\Adv}{\ensuremath{\mathcal{A}}}
\newcommand{\Sim}{\ensuremath{\mathcal{S}}}

\newcommand{\Hash}{\ensuremath{\mathsf{H}}}
\newcommand{\commit}{\ensuremath{\mathsf{Com}}}
\newcommand{\Commit}[1]{\ensuremath{\commit(#1)}}

\newcommand{\ttp}{\mathsf{ttp}}
\newcommand{\msg}{\mathsf{msg}}
\newcommand{\valid}{\mathsf{valid}}
\newcommand{\maxv}{\ensuremath{\mathsf{max}}}

\newcommand{\Key}[1]{\ensuremath{k_{#1}}}

\newcommand{\rtt}{\ensuremath{\mathsf{rtt}}}
\newcommand{\TP}{\ensuremath{\mathsf{TP}}}

\newcommand{\SELECT}{\ensuremath{\mathsf{Select}}}
\newcommand{\INPUT}{\ensuremath{\mathsf{Input}}}
\newcommand{\OUTPUT}{\ensuremath{\mathsf{Output}}}

\newcommand{\FSETUP}{\ensuremath{\mathcal{F}_{\mathsf{setup}}}} 

\newcommand{\vc}{\ensuremath{\mathsf{c}}}

\newcommand{\va}{\ensuremath{\mathsf{a}}}
\newcommand{\vk}{\ensuremath{\mathsf{k}}}
\newcommand{\vp}{\ensuremath{\mathsf{p}}}
\newcommand{\vq}{\ensuremath{\mathsf{q}}}
\newcommand{\vr}{\ensuremath{\mathsf{r}}}
\newcommand{\vu}{\ensuremath{\mathsf{u}}}
\newcommand{\vv}{\ensuremath{\mathsf{v}}}
\newcommand{\vx}{\ensuremath{\mathsf{x}}}
\newcommand{\vy}{\ensuremath{\mathsf{y}}}
\newcommand{\vz}{\ensuremath{\mathsf{z}}}

\newcommand{\vecP}{\ensuremath{\vec{\mathbf{p}}}}
\newcommand{\vecQ}{\ensuremath{\vec{\mathbf{q}}}}

\newcommand{\vecW}{\ensuremath{\vec{\mathbf{w}}}}
\newcommand{\vecX}{\ensuremath{\vec{\mathbf{x}}}}
\newcommand{\vecY}{\ensuremath{\vec{\mathbf{y}}}}

\newcommand{\Mat}[1]{\ensuremath{\mathbf{#1}}}

\newcommand{\trunc}[1]{\ensuremath{{#1}^{d}}}
\newcommand{\vrt}{\ensuremath{\mathsf{r}^d}}

\newcommand{\sA}{\ensuremath{\mathsf{A}}}
\newcommand{\sB}{\ensuremath{\mathsf{B}}}



%% file: symbols_blaze.tex
\newcommand{\sqr}[1]{\ensuremath{\left[#1\right]}}  
\newcommand{\sqrA}[1]{\ensuremath{\left[#1\right]_{1}}}
\newcommand{\sqrB}[1]{\ensuremath{\left[#1\right]_{2}}}

\newcommand{\sqrV}[2]{\ensuremath{\left[#1\right]_{#2}}}

\newcommand{\sgr}[1]{\ensuremath{\langle #1 \rangle}}

\newcommand{\shr}[1]{\ensuremath{\llbracket #1 \rrbracket}}

\newcommand{\shrB}[1]{\ensuremath{\llbracket #1 \rrbracket}^{\bf B}} 
\newcommand{\shrd}{\ensuremath{\llbracket \cdot \rrbracket}}

\newcommand{\av}[1]{\ensuremath{\alpha_{#1}}}   
\newcommand{\bv}[1]{\ensuremath{\beta_{#1}}}    
\newcommand{\gv}[1]{\ensuremath{\gamma_{#1}}}   

\newcommand{\val}{\ensuremath{\mathsf{v}}}

\newcommand{\Sh}{\ensuremath{\mathsf{sh}}}

\newcommand{\JSh}{\ensuremath{\mathsf{jsh}}} 
\newcommand{\Rec}{\ensuremath{\mathsf{rec}}}

\newcommand{\Jmp}{\ensuremath{\mathsf{jmp}}}

\newcommand{\Mult}{\ensuremath{\mathsf{mult}}}

\newcommand{\MulPre}{\ensuremath{\mathsf{mulPre}}}

\newcommand{\BitExt}{\ensuremath{\mathsf{bitext}}}
\newcommand{\BitA}{\mathsf{bit2A}}
\newcommand{\BitInj}{\mathsf{BitInj}}
\newcommand{\DotP}{\ensuremath{\mathsf{dotp}}}
\newcommand{\Trunc}{\ensuremath{\mathsf{trgen}}}
\newcommand{\DotPTr}{\ensuremath{\mathsf{dotpt}}}
\newcommand{\DotPPre}{\ensuremath{\mathsf{dotpPre}}}

\newcommand{\ReLU}{\ensuremath{\mathsf{relu}}}
\newcommand{\Sig}{\ensuremath{\mathsf{sig}}}
\newcommand{\MSB}{\ensuremath{\mathsf{msb}}}

\newcommand{\piSh}{\ensuremath{\Pi_{\Sh}}}
\newcommand{\piShS}{\ensuremath{\Pi_{\Sh}^{\SOC}}}
\newcommand{\piJSh}{\ensuremath{\Pi_{\JSh}}} 
\newcommand{\piRec}{\ensuremath{\Pi_{\Rec}}}
\newcommand{\piRecS}{\ensuremath{\Pi_{\Rec}^{\SOC}}}
\newcommand{\piJmp}{\ensuremath{\Pi_{\Jmp}}} 
\newcommand{\piMult}{\ensuremath{\Pi_{\Mult}}}

\newcommand{\piMulPre}{\ensuremath{\Pi_{\MulPre}}}

\newcommand{\piBitExt}{\ensuremath{\Pi_{\BitExt}}}
\newcommand{\PiBitA}{\ensuremath{\Pi_{\BitA}}}
\newcommand{\PiBitInj}{\ensuremath{\Pi_{\BitInj}}}

\newcommand{\piDotP}{\ensuremath{\Pi_{\DotP}}}
\newcommand{\piTrunc}{\ensuremath{\Pi_{\Trunc}}}
\newcommand{\piDotPTr}{\ensuremath{\Pi_{\DotPTr}}}

\newcommand{\piDotPPre}{\ensuremath{\Pi_{\DotPPre}}}



%% file: symbols_3pc_SWIFT.tex
\newcommand{\ESet}{\ensuremath{P_1, P_2}}		
\newcommand{\EInSet}{\ensuremath{ \{1, 2\}}}	
\newcommand{\PSet}{\ensuremath{P_0, P_1, P_2}}	

\newcommand{\md}{\ensuremath{\mathsf{d}}} 
\newcommand{\me}{\ensuremath{\mathsf{e}}} 
\newcommand{\mf}{\ensuremath{\mathsf{f}}} 

\newcommand{\jsend}{\ensuremath{\mathsf{jmp\mbox{-}send}}}

\newcommand{\Chi}{\ensuremath{\chi}}
\newcommand{\starbeta}[1]{\ensuremath{\mathsf{\beta}_{#1}^{\star}}}

\newcommand{\GammaV}[1]{\ensuremath{\Gamma_{#1}}}
\newcommand{\GammaA}[1]{\ensuremath{\sqrV{\Gamma_{#1}}{1}}}
\newcommand{\GammaB}[1]{\ensuremath{\sqrV{\Gamma_{#1}}{2}}}
\newcommand{\Gammaxy}{\ensuremath{\Gamma_{\wx \wy}}}

\newcommand{\sz}{\ensuremath{\zeta}} 
\newcommand{\nm}{\ensuremath{m}} 
\newcommand{\cc}{\ensuremath{c}} 
\newcommand{\cg}{\ensuremath{g}} 
\newcommand{\cG}{\ensuremath{G}} 
\newcommand{\nL}{\ensuremath{L}} 
\newcommand{\nM}{\ensuremath{M}} 
\newcommand{\ic}{\ensuremath{u}}

\newcommand{\rt}{\ensuremath{\theta}}
\newcommand{\re}{\ensuremath{\eta}}
\newcommand{\cv}{\ensuremath{\mathsf{CV}}} 

\newcommand{\Funct}{\ensuremath{\mathcal{F}_{\mathsf{3PC}}}}
\newcommand{\FJmp}{\ensuremath{\mathcal{F}_{\mathsf{jmp}}}}

\newcommand{\FuncBitA}{\ensuremath{\mathcal{F}_{\mathsf{bit2A}}}}

\newcommand{\FMulPre}{\ensuremath{\mathcal{F}_{\mathsf{MulPre}}}}
\newcommand{\FDotPPre}{\ensuremath{\mathcal{F}_{\mathsf{DotPPre}}}}

%% file: symbols_4pc_SWIFT.tex
\newcommand{\A}{\mathsf{A}}
\newcommand{\B}{\mathsf{B}}
\newcommand{\R}{\mathsf{R}}
\newcommand{\ve}{\mathsf{e}}

\newcommand{\JmpF}{\ensuremath{\mathsf{jmp4}}}

\newcommand{\Gammaxdy}{\ensuremath{\Gamma_{\vecX \band \vecY}}}

\newcommand{\jsendf}{\ensuremath{\mathsf{jmp4\mbox{-}send}}}

\newcommand{\piJmpF}{\ensuremath{\Pi_{\mathsf{jmp4}}}}

\newcommand{\pifinput}{\ensuremath{\Pi_{\mathsf{sh4}}}}

\newcommand{\pifmul}{\ensuremath{\Pi_{\mathsf{mult4}}}}
\newcommand{\pifrec}{\ensuremath{\Pi_{\mathsf{rec4}}}}

\newcommand{\pifjsh}{\ensuremath{\Pi_{\mathsf{jsh4}}}}
\newcommand{\pifbitext}{\ensuremath{\Pi_{\mathsf{bitext4}}}}
\newcommand{\pifbita}{\ensuremath{\Pi_{\mathsf{bit2A4}}}}
\newcommand{\pifdotp}{\ensuremath{\Pi_{\mathsf{dotp4}}}}
\newcommand{\piftrgen}{\ensuremath{\Pi_{\mathsf{trgen4}}}}
\newcommand{\pifdotpt}{\ensuremath{\Pi_{\mathsf{dotpt4}}}}
\newcommand{\piBitInjF}{\ensuremath{\Pi_{\mathsf{bitinj4}}}}
\newcommand{\pifsgrsh}{\ensuremath{\Pi_{\mathsf{ash4}}}}

\newcommand{\Funcf}{\ensuremath{\mathcal{F}_{\mathsf{4PC}}}}
\newcommand{\FSETUPf}{\ensuremath{\mathcal{F}_{\mathsf{setup4}}}}
\newcommand{\FJmpF}{\ensuremath{\mathcal{F}_{\mathsf{jmp4}}}}

\newcommand{\FBitextF}{\ensuremath{\mathcal{F}_{\mathsf{bitext4}}}}
\newcommand{\FBitaF}{\ensuremath{\mathcal{F}_{\mathsf{bit2A4}}}}

\newcommand{\FTrgenF}{\ensuremath{\mathcal{F}_{\mathsf{trgen4}}}}
\newcommand{\FBitInjF}{\ensuremath{\mathcal{F}_{\mathsf{bitinj4}}}}

%% file: comments.tex
\definecolor{UniBlau}{cmyk}{1,0.7,0,0}
\definecolor{UniGruen}{cmyk}{0.6,0,1,0}
\definecolor{UniOrange}{cmyk}{0,0.3,1,0}
\definecolor{UniRot}{cmyk}{0.4,1,0,0}
\definecolor{darkred}{rgb}{.6,0,0}
\definecolor{darkgreen}{rgb}{0,.4,0}
\definecolor{darkblue}{rgb}{0,0,.6}

\newif\ifsubmission
\submissionfalse

\ifsubmission
	\newcommand{\commentA}[1] {}
	\newcommand{\commentAJ}[1]{}
	\newcommand{\commentM}[1] {} 
	\newcommand{\commentN}[1] {}
	\newcommand{\EXTRALINES}[1] {}

\else
	\newcommand{\commentA}[1] {\textcolor{blue}  {{\sf (}{\sl{#1}} {\sf - Arpita)}}}
	\newcommand{\commentAJ}[1]{\textcolor{violet}{{\sf (}{\sl{#1}} {\sf - Ajith)}}}
	\newcommand{\commentM}[1] {\textcolor{olive} {{\sf (}{\sl{#1}} {\sf - Mahak)}}}
	\newcommand{\commentN}[1] {\textcolor{cyan} {{\sf (}{\sl{#1}} {\sf - Nishat)}}}
	\newcommand{\EXTRALINES}[1] {\textcolor{darkblue} {#1}}

\fi

%% file: Main_1_Introduction.tex
Privacy Preserving Machine Learning (PPML), a booming field of research, allows Machine Learning (ML) computations over private data of users while ensuring the privacy of the data. PPML finds applications in sectors that deal with sensitive/confidential data, e.g.  healthcare, finance, and in cases where organisations are prohibited from sharing client information due to privacy laws such as CCPA and GDPR. However, PPML solutions make the already computationally heavy ML algorithms more compute-intensive. An average end-user who lacks the infrastructure required to run these tasks prefers to outsource the computation to a powerful set of specialized cloud servers and leverage their services on a pay-per-use basis. This is addressed by the Secure Outsourced Computation (SOC) paradigm, and thus is an apt fit for the need of the moment. 
Many recent works~\cite{MohasselZ17, MakriRSV19,RiaziWTS0K18, MR18, WaghGC19,ASTRA,FLASH,Trident,BLAZE} exploit Secure Multiparty Computation (MPC) techniques to realize PPML in the SOC setting where the servers enact the role of the parties. Informally, MPC enables $n$ mutually distrusting parties to compute a function over their private inputs, while ensuring the privacy of the same against an adversary controlling up to $t$ parties.  Both the training and prediction phases of PPML can be realized in the SOC setting. The common approach of outsourcing followed in the PPML literature, as well as by our work,  requires the users to secret-share\footnote{The threshold of the secret-sharing is decided based on the number of corrupt servers so that privacy is preserved.} their inputs between the set of hired (untrusted) servers, who jointly interact and compute the secret-shared output, and reconstruct it towards the users. 

In a bid to improve practical efficiency, many recent works~\cite{DPSZ12,SPDZ2,SPDZ3,KOS16,BaumDTZ16,DamgardOS18,CramerDESX18,KellerPR18,ASTRA,FLASH,Trident,BLAZE} cast their protocols into the preprocessing model wherein the input-independent (yet function-dependent) phase computationally heavy tasks are computed in advance, resulting in a fast online phase. This paradigm suits scenario analogous to PPML setting, where functions (ML algorithms) typically need to be evaluated a large number of times, and the function description is known beforehand. To further enhance practical efficiency by leveraging  CPU optimizations, recent works~\cite{BogdanovLW08,DamgardOS18,CramerFIK03,DSZ15,Damgard0FKSV19} propose MPC protocols that work over $32$ or $64$ bit rings. Lastly, solutions for a small number of parties have received a huge momentum due to the many cost-effective customizations that they permit, for instance, a cheaper realisation of multiplication through custom-made secret sharing schemes~\cite{AFLNO16,ABFLLNOWW17,ASTRA,BLAZE,Trident,FLASH}.

We now motivate the need for robustness aka guaranteed output delivery (GOD) over fairness\footnote{This ensures either all parties or none learn the output.}, or even abort security\footnote{This may allow the corrupt parties {\em alone} to learn the output.}, in the domain of PPML. Robustness provides the guarantee of output delivery to all protocol participants, no matter how the adversary misbehaves. Robustness is  crucial for real-world deployment and usage of PPML techniques. Consider the following scenario wherein an ML model owner wishes to provide inference service. The model owner shares the model parameters between the servers, while the end-users share their queries. A protocol that provides security with abort or fairness will not suffice as in both the cases a malicious adversary can lead to the protocol aborting, resulting in the user not obtaining the desired output. This leads to denial of service and heavy economic losses for the service provider. For data providers, as more training data leads to more accurate models, collaboratively building a model enables them to provide better ML services, and consequently, attract more clients. A robust framework encourages active involvement from multiple data providers. Hence, for the seamless adoption of PPML solutions in the real world, the robustness of the protocol is of utmost importance. 
Several works~\cite{ASTRA,MR18,WaghGC19,BLAZE,Trident} realizing PPML via MPC  settle for weaker guarantees such as abort and fairness. Achieving  the strongest notion of GOD without degrading performance is an interesting goal which is the focus of this work. The hall-mark result of~\cite{Cleve86} suggests that an honest-majority amongst the servers is necessary to achieve robustness. Consequent to the discussion above, we focus on the honest-majority setting with a small set of parties, especially 3 and 4 parties, both of which have drawn enormous attention recently~\cite{MRZ15,AFLNO16,FLNW17,ABFLLNOWW17,ByaliJPR18,NV18,ASTRA,BonehBCGI19,BLAZE,Trident,FLASH,BGIN19}. 

The $3/4$-party setting enables simpler, more efficient, and customized secure protocols compared to the $n$-party setting. Real-world MPC applications and frameworks such as the Danish sugar beet auction~\cite{auction} and Sharemind~\cite{BogdanovLW08}, have demonstrated the practicality of $3$-party protocols. Additionally, in an outsourced setting, $3/4$PC is useful and relevant even when there are more parties.   Specifically, here the entire computation is offloaded to $3/4$ hired servers, after initial sharing of inputs by the parties amongst the servers.  This is precisely what we (and some existing papers~\cite{BLAZE, FLASH, MazloomLRG20}) contemplate as the setting for providing ML-as-a-service.
Our protocols work over rings, are cast in the preprocessing paradigm, and achieve GOD. 

\paragraph{Related Work}
We restrict the relevant work to a small number of parties and honest-majority, focusing first on MPC, followed by PPML. MPC protocols for a small population can be cast into orthogonal domains of low latency protocols~\cite{PatraR18, ByaliJPR18, ByaliHPS19}, and high throughput protocols~\cite{BogdanovLW08, AFLNO16, ABFLLNOWW17, FLNW17, CGHIKLN18, ASTRA, NV18, BLAZE, BGIN19, AbspoelDEN19, EOP19}. In the 3PC setting,~\cite{AFLNO16, ASTRA} provide efficient semi-honest protocols wherein ASTRA~\cite{ASTRA} improved upon~\cite{AFLNO16} by casting the protocols in the preprocessing model and provided a fast online phase. ASTRA further provided security with fairness in the malicious setting with an improved online phase compared to~\cite{ABFLLNOWW17}. 
Later, a maliciously-secure 3PC protocol based on distributed zero-knowledge techniques was proposed by Boneh et al.~\cite{BonehBCGI19} providing abort security. Further, building on~\cite{BonehBCGI19} and enhancing the security to GOD, Boyle et al.~\cite{BGIN19} proposed a concretely efficient 3PC protocol with an amortized communication cost of $3$ field elements (can be extended to work over rings) per multiplication gate. Concurrently, BLAZE~\cite{BLAZE} provided a fair protocol in the preprocessing model, which required communicating $3$ ring elements in each phase. However, BLAZE eliminated the reliance on the computationally intensive distributed zero-knowledge system (whose efficiency kicks in for large circuit or many multiplication gates) from the online phase and pushed it to the preprocessing phase. This resulted in a faster online phase compared to~\cite{BGIN19}. 

In the regime of  4PC, Gordon et al.~\cite{GordonR018} presented protocols achieving abort security and GOD. However,~\cite{GordonR018} relied on expensive public-key primitives and broadcast channels to achieve GOD. Trident~\cite{Trident} improved over the abort protocol of~\cite{GordonR018}, providing a fast online phase achieving security with fairness, and presented a framework for mixed world computations~\cite{DSZ15}. A robust 4PC protocol was provided in FLASH~\cite{FLASH}, which requires communicating $6$ ring elements, each, in the preprocessing and online phases.

In the PPML domain, MPC has been used for various ML algorithms such as Decision Trees~\cite{LindellP02}, Linear Regression~\cite{DuA01,SanilKLR04}, k-means clustering~\cite{JagannathanW05, BunnO07}, SVM Classification~\cite{YuVJ06, VaidyaYJ08}, Logistic Regression~\cite{SlavkovicNT07}. In the 3PC SOC setting, the works of ABY3~\cite{MR18} and SecureNN~\cite{WaghGC19}, provide security with abort.  This was followed by ASTRA~\cite{ASTRA}, which improves upon ABY3 and achieves security with fairness. ASTRA presents primitives to build protocols for Linear Regression and Logistic Regression inference.  
Recently, BLAZE improves over the efficiency of ASTRA and additionally tackles training for the above ML tasks, which requires building additional  PPML building blocks, such as truncation and bit to arithmetic conversions. In the 4PC setting, the first robust framework for PPML was provided by FLASH~\cite{FLASH} which proposed efficient building blocks for ML such as dot product, truncation, MSB extraction, and bit conversion. The works of~\cite{MohasselZ17,MR18,WaghGC19,ASTRA,BLAZE,Trident,FLASH} work over rings to garner practical efficiency. In terms of efficiency, BLAZE and respectively FLASH and Trident are the closest competitors of this work in 3PC and 4PC settings. 

\subsection{Our Contributions}
We propose, {\bf SWIFT}, a robust maliciously-secure framework for PPML in the SOC setting, with a set of $3$ and $4$ servers having an honest-majority. At the heart of our framework lies highly-efficient, maliciously-secure, 3PC and 4PC over rings (both $\Z{\ell}$ and $\Z{1}$) that provide GOD in the honest-majority setting.  We cast our protocols in the preprocessing model, which helps obtain a fast online phase. As mentioned earlier, the input-independent (yet function-dependent) computations will be performed in the preprocessing phase. 

	\begin{table*}[htb!]
		\centering
		\resizebox{.94\textwidth}{!}{
			\begin{tabular}{r || r | r | r | r | r || r | r | r | r | r}
				\toprule
				\multirow{3}[3]{*}{\makecell{Building\\Blocks}}\tnote{1}
				& \multicolumn{5}{c||}{3PC} & \multicolumn{5}{c}{4PC} \\
				\cmidrule{2-11}
				& \multicolumn{1}{c|}{\multirow{2}[2]{*}{Ref.}} 
				& \multicolumn{1}{c|}{Pre.} & \multicolumn{2}{c|}{Online} 
				& \multicolumn{1}{c||}{\multirow{2}[2]{*}{Security}}  
				& \multicolumn{1}{c|}{\multirow{2}[2]{*}{Ref.}} 
				& \multicolumn{1}{c|}{Pre.} & \multicolumn{2}{c}{Online} 
				& \multicolumn{1}{|c}{\multirow{2}[2]{*}{Security}} \\
				\cmidrule{3-5}\cmidrule{8-10}
				& & Comm.~($\ell$) & Rounds & Comm.~($\ell$) &
				& & Comm.~($\ell$) & Rounds & Comm.~($\ell$)   \\
				\midrule 
				\multirow{4}{*}{Multiplication}     
				& \cite{BonehBCGI19} & $1$          & 1 & $2$    & Abort      
				&                    &              &            &         \\ 
				& \cite{BGIN19}      & -            & 3 & $3$    & GOD        
				& Trident            & $3$          & 1 & $3$    & Fair   \\ 
				&  BLAZE             & $3$          & 1 & $3$    & Fair      
				&  FLASH             & $6$          & 1 & $6$    & GOD     \\ 
				&  {\bf SWIFT}       & $\mathbf{3}$ & $\mathbf{1}$ & $\mathbf{3}$ & GOD     
				&  {\bf SWIFT}       & $\mathbf{3}$ & $\mathbf{1}$ & $\mathbf{3}$ & GOD  \\ 
				\midrule	
				\multirow{3}{*}{Dot Product} 
				& & & & &   
				& Trident     & $3$              & 1            & $3$          & Fair \\ 
				& BLAZE       & $3 \nf$          & 1            & $3 $         & Fair    
				& FLASH       & $6$              & 1            & $6$          & GOD   \\ 
				& {\bf SWIFT} & $\mathbf{3}$   & $\mathbf{1}$ & $\mathbf{3}$ & GOD      
				& {\bf SWIFT} & $\mathbf{3}$     & $\mathbf{1}$ & $\mathbf{3}$ & GOD   \\ 
				\midrule
				\multirow{3}{*}{\makecell[r]{Dot Product\\ with Truncation}}   
				& & & & &   
				& Trident     &  $6$                  & 1            & $3$          & Fair \\ 
				& BLAZE       & $ 3 \nf + 2$          & 1            & $ 3$         & Fair    
				& FLASH       & $8$                   & 1            & $6$          & GOD   \\ 
				& {\bf SWIFT} & $\mathbf{15}$  & $\mathbf{1}$ & $\mathbf{3}$ & GOD      
				& {\bf SWIFT} & $\mathbf{4}$          & $\mathbf{1}$ & $\mathbf{3}$ & GOD   \\ 
				\midrule
				\multirow{3}{*}{\makecell[r]{Bit\\Extraction}}  
				& & & & &    
				& Trident     & $\approx 8$                     & $\log \ell + 1$            & $\approx 7$               & Fair \\ 
				& BLAZE       & 9          & $1+\log \ell$           & 9         & Fair    
				& FLASH       & $ 14 $                   & $\log \ell$  & $ 14 $         & GOD   \\ 
				& {\bf SWIFT} & $\mathbf{9}$  & $\mathbf{1+\log \ell}$ & $\mathbf{9}$ & GOD      
				& {\bf SWIFT} & $\approx \mathbf{7}$ & $\mathbf{\log \ell}$ & $\approx \mathbf{7}$ & GOD   \\ 
				\midrule
				\multirow{3}{*}{\makecell[r]{Bit to\\Arithmetic}}    
				& & & & &    
				& Trident     & $\approx 3$  & 1            & $3$              & Fair \\ 
				& BLAZE       & $9$          & 1            & $4$              & Fair    
				& FLASH       & $6$          & 1            & $8$              & GOD   \\ 
				& {\bf SWIFT} & $\mathbf{9}$ & $\mathbf{1}$ & $\mathbf{4}$     & GOD      
				& {\bf SWIFT} & $\approx \mathbf{3}$ & $\mathbf{1}$ & $\mathbf{3}$     & GOD   \\ 
				\midrule
				\multirow{3}{*}{\makecell[r]{Bit\\Injection}}      
				& & & & &    
				& Trident     & $\approx 6$  & 1            & $3$              & Fair \\ 
				& BLAZE       & $12$         & 2            & $7$              & Fair    
				& FLASH       & $8$          & 2            & $10$              & GOD   \\ 
				& {\bf SWIFT} & $\mathbf{12}$& $\mathbf{2}$ & $\mathbf{7}$     & GOD      
				& {\bf SWIFT} & $\approx \mathbf{6}$ & $\mathbf{1}$ & $\mathbf{3}$     & GOD   \\ 
				\bottomrule
			\end{tabular}
		}
		{\footnotesize
			\begin{tablenotes}
				\item[1] -- Notations: $\ell$ - size of ring in bits, $\nf$ - size of vectors for dot product. 
			\end{tablenotes}
		}
		\vspace{-2mm}
		\caption{\small 
			3PC and 4PC: Comparison of SWIFT with its closest competitors in terms of Communication and Round Complexity\label{tab:comparison}}
	\end{table*}

To the best of our knowledge, SWIFT is the first robust and efficient PPML framework in the 3PC setting and is as fast as (and is strictly better in some cases than) the best known {\em fair} 3PC framework BLAZE~\cite{BLAZE}. We extend our 3PC framework for 4 servers.  In this regime, SWIFT is as fast as the best known {\em fair} 4PC framework Trident~\cite{Trident} and twice faster than best known {\em robust} 4PC framework FLASH~\cite{FLASH}. We detail our contributions next.

\vspace{-2mm}
\paragraph{Robust 3/4PC frameworks} 
The framework consists of a range of primitives realized in a privacy-preserving way which is ensured via running computation in a secret-shared fashion.  
We use secret-sharing  over both $\Z{\ell}$ and its special instantiation $\Z{1}$ and refer them  as {\em arithmetic} and  respectively {\em boolean} sharing. Our framework consists of realizations for all primitives needed for general MPC and PPML such as multiplication,  dot-product,  truncation, bit extraction (given arithmetic sharing of a value $\val$, this is used to generate boolean sharing of the most significant bit ($\MSB$) of the value), bit to arithmetic sharing conversion (converts the boolean sharing of a single bit value to its arithmetic sharing),  bit injection (computes the arithmetic sharing of $\bitb \cdot \val$, given the boolean sharing of a bit $\bitb$ and the arithmetic sharing of a ring element $\val$) and above all, input sharing and output reconstruction in the SOC setting. A highlight of our 3PC framework, which, to the best of our knowledge is achieved for the first time, is a robust dot-product protocol whose (amortized) communication cost is independent of the vector size, which we obtain by extending the techniques of~\cite{BonehBCGI19, BGIN19}.
The performance comparison in terms of concrete cost for communication and rounds, for PPML primitives in both 3PC and 4PC setting, appear in Table~\ref{tab:comparison}. As claimed, SWIFT is on par  with   BLAZE for most of the  primitives (while improving security from fair to GOD) and is strictly better than BLAZE in case of dot product and dot product with truncation.  For 4PC, SWIFT is on par with Trident in most cases (and is slightly better for dot product with truncation and bit injection), while it is doubly faster than FLASH. Since BLAZE outperforms the 3PC abort framework of ABY3~\cite{MR18} while Trident outperforms the known 4PC with abort~\cite{GordonR018}, SWIFT attains robustness with better cost than the know protocols with weaker guarantees. No performance loss coupled with the strongest security guarantee makes our robust framework an opt choice for practical applications including PPML.

\paragraph{Applications and Benchmarking} 
We demonstrate the practicality of our protocols by  benchmarking PPML, particularly, Logistic Regression (training and inference) and popular Neural Networks (inference) such as~\cite{MohasselZ17}, LeNet~\cite{lenet} and VGG16~\cite{vgg16} having millions of parameters. The NN training requires mixed-world conversions~\cite{DSZ15,MR18,Trident}, which we leave as future work. Our PPML blocks can be used to perform training and inference of Linear Regression, Support Vector Machines, and Binarized Neural Networks (as demonstrated in~\cite{ASTRA,BLAZE,Trident,FLASH}).

\paragraph{New techniques and Comparisons with Prior Works}
To begin with, we introduce a new primitive called Joint Message Passing~($\Jmp$) that allows two servers to relay a common message to the third server such that either the relay is successful or an honest server is identified. The identified honest party enacts the role of a trusted third party~($\TTP$) to take the computation to completion. $\Jmp$ is extremely efficient as for a message of $\ell$ elements it only incurs the minimal communication cost of $\ell$ elements  (in an amortized sense). Without any extra cost, it allows us to replace several pivotal private communications, that may lead to abort, either because the malicious sender does not send anything or  sends a wrong message.  All our primitives, either for a general 3PC or a PPML task, achieve GOD relying on $\Jmp$. 

Second, instead of using the multiplication of~\cite{BGIN19} (which has the same overall communication cost as that of our online phase), we build a new  protocol. This is because the former involves  distributed zero-knowledge protocols. The cost of this  heavy machinery  gets amortized only for large circuits having millions of gates, which is very unlikely for inference and moderately heavy training tasks in PPML. As in BLAZE~\cite{BLAZE}, we follow a similar structure for our multiplication protocol but differ considerably in techniques as our goal is to obtain GOD. Our approach is to manipulate and transform some of the protocol steps so that two other servers can locally compute the information required by a server in a round. However, this transformation is not straight forward since BLAZE was constructed with a focus towards providing only fairness (details appear in \S\ref{sec:3PC}). 
The multiplication protocol forms a technical basis for our dot product protocol and other PPML building blocks. We emphasise again that the (amortized) cost of our dot product protocol is independent of the vector size.

Third, extending to 4PC brings several performance improvements over 3PC. Most prominent of all is a conceptually simple $\Jmp$ instantiation, which forgoes the broadcast channel while retaining the same communication cost; and a dot product with cost independent of vector size sans the 3PC amortization technique.

Fourth, we  provide robust protocols for input sharing and output reconstruction phase in  the SOC setting, wherein a user shares its input with the servers, and the output is reconstructed towards a user. The need for robustness and communication efficiency together makes these tasks slightly non-trivial.
As a highlight, we introduce a super-fast  online phase for the reconstruction protocol, which gives $4\times$ improvement in terms of rounds (apart from improvement in the communication) compared to BLAZE. Although we aim for GOD, we ensure that an end-user is never part of broadcast which is relatively expensive than {\em atomic} point-to-point communication.

As a final remark, we note that the recent work of \cite{DEK20} proposes a variant of GOD in the 4PC setting, which is termed as private robustness. The authors of \cite{DEK20} state that private robustness is a variant of GOD which guarantees that the correct output is produced in the end, but without relying on an honest party learning the user’s private inputs. Thus, departing from the approach of employing a $\TTP$ to complete the computation when malicious behaviour is detected, \cite{DEK20} attains GOD by eliminating a potentially corrupt party, and {\em repeating} the secure computation with fewer number of parties which are deemed to be honest.  We point out a few concerns on this work.
 Firstly, as mentioned earlier, the goal of private robustness is to prevent an honest party learning the user's input, thereby preventing it from misusing this private information (user's input) in the future if it goes rogue. We note, however, that in the private robustness setting, although an honest party does not learn a user's input as a part of the protocol, nothing prevents an adversary from revealing its view to an honest party. In such a scenario, if the honest party goes rogue in the future, it can use its view together with the view received from the adversary to obtain a user's input. Hence, we believe that the attacks that can be launched in our variant for achieving GOD, can also be launched in the private robustness variant, making the two equivalent. Secondly and importantly, a formal treatment of private robustness is missing in \cite{DEK20} which makes it unclear what additional security is achieved on top of traditional GOD security. 
 Here, we additionally note that the notion of private robustness does not comply with the recently introduced notion of FaF security \cite{alon20}\footnote{Althought \cite{DEK20} states that the issue of private robustness was identified and treated formally in \cite{alon20}, it is not clear whether \cite{DEK20} achieves the FaF security of \cite{alon20}. 
 For a corruption threshold $t$ and an honest threshold $h^*$,  FaF security demands that the view of any $t$ corrupt parties and {\em separately} the view of any $h^*$ honest parties must be simulatable.  The latter part which is a new addition compared to traditional security definition requires the presence of a (semi-honest) simulator that can simulate the view of any subset of $h^*$  honest parties, given the input and output of those honest parties. This notion is shown to be achievable if and only if $2t + h^* < n$, where $n$ is the total number of parties. With $n=4$, $t = 1$ and $h^* = 1$,  the results of \cite{DEK20} does not satisfy FaF security since an honest party's view may include the inputs of the other honest parties  when a corrupt party, deviating from the protocol steps,  sends its view to it. 
 A formal analysis of protocols in \cite{DEK20} satisfying the FaF security notion of \cite{alon20} is missing.}.  
 Lastly, we emphasize that the approach of eliminating a potentially corrupt party, and re-running the computation results in doubling or tripling the communication cost, thereby undermining the efficiency gains.

\subsection{Organisation of the paper}
The rest of the paper is organized as follows. In \S\ref{sec:Prelims} we describe the system model, preliminaries and notations used. \S\ref{sec:3PC} and \S\ref{sec:4PC} detail our constructs in the 3PC and 4PC setting respectively. These are followed by the applications and benchmarking in \S\ref{sec:Implementation}. \S\ref{app:Prelims} elaborates on additional preliminaries while the security proofs for our constructions are provided in \S\ref{app:proofs}.

%% file: Main_2_Preliminaries.tex
We consider a set of three servers $\Partyset = \{ P_0, P_1, P_2 \}$ that are connected by pair-wise private and authentic channels in a synchronous network, and a static, malicious adversary that can corrupt at most one server. We use a broadcast channel for 3PC alone, which is inevitable~\cite{CHOR18}.
For ML training, several data-owners who wish to jointly train a model, secret share (using the sharing semantics that will appear later) their data among the servers. For ML inference, a model-owner and client secret share the model and the query, respectively, among the servers. Once the inputs are available in the shared format, the servers perform computations and obtain the output in the shared form. In the case of training, the output model is reconstructed towards the data-owners, whereas for inference, the prediction result is reconstructed towards the client. We assume that an arbitrary number of data-owners may collude with a corrupt server for training, whereas for the case of prediction, we assume that either the model-owner or the client can collude with a corrupt server.  We prove the security of our protocols using a standard real-world / ideal-world paradigm.  We also explore the above model for the four server setting with  $\Partyset = \{ P_0, P_1, P_2, P_3\}$. The aforementioned setting has been explored extensively \cite{MohasselZ17,MR18,ASTRA,Trident,FLASH,BLAZE}. 

Our constructions achieve the strongest security guarantee of GOD. A protocol is said to be {\em robust} or achieve GOD if all parties obtain the output of the protocol regardless of how the adversary behaves. In our model, this translates to all the data owners obtaining the trained model for the case of ML training, while the client obtaining the query output for ML inference.
All our protocols are cast into: {\em input-independent} preprocessing phase and {\em input-dependent} online phase.

For 3/4PC, the function to be computed is expressed as a circuit $\ckt$, whose topology is public, and is evaluated over an arithmetic ring $\Z{\ell}$ or boolean ring $\Z{1}$.  For PPML, we consider computation over the same algebraic structure.  To deal with floating-point values, we use Fixed-Point Arithmetic (FPA)~\cite{MohasselZ17,MR18,ASTRA,FLASH,Trident,BLAZE} representation in which a decimal value is represented as an $\ell$-bit integer in signed 2's complement representation. The most significant bit (MSB) represents the sign bit, and $x$ least significant bits are reserved for the fractional part. The $\ell$-bit integer is then treated as an element of $\Z{\ell}$, and operations are performed modulo $2^{\ell}$. We set $\ell = 64, x = 13$, leaving $\ell - x - 1$ bits for the integer part.

The servers use a one-time key setup, modelled as a functionality $\FSETUP$ (\boxref{fig:FSETUP}), to establish pre-shared random keys for pseudo-random functions (PRF) between them. A similar setup is used in \cite{RiaziWTS0K18, FLNW17, ABFLLNOWW17, MR18, ASTRA, BLAZE, BGIN19} for three server case and in \cite{FLASH, Trident} for four server setting. The key-setup can be instantiated using any standard MPC protocol in the respective setting. Further, our protocols make use of a {\em collision-resistant} hash function, denoted by $\Hash()$, and a commitment scheme, denoted by $\commit()$. The formal details of key setup, hash function, and commitment scheme are deferred to \S\ref{app:Prelims}.   

\begin{notation}
	The $i^{th}$ element of a vector $\vecX$ is denoted as $\vx_i$. The dot product of two $\nf$ length vectors, $\vecX$ and $\vecY$, is computed as $\vecX \band \vecY = \sum_{i = 1}^{\nf} \vx_i \vy_i$. For two matrices  $\Mat{X}, \Mat{Y}$, the operation $\Mat{X} \circ \Mat{Y}$ denotes the matrix multiplication. The bit in the $i^{th}$ position of an $\ell$-bit value $\val$ is denoted by $\val[i]$.
\end{notation}

\begin{notation}\label{boolequivring}
	For a bit $\bitb \in \{0,1\}$, we use $\arval{\bitb}$ to denote the equivalent value of $\bitb$ over the ring $\Z{\ell}$. $\arval{\bitb}$ will have its least significant bit  set to $\bitb$, while all other bits will be set to zero. 
\end{notation}

%% file: Main_3_3PC.tex
In this section, we first introduce the sharing semantics for three servers. Then, we introduce our new  Joint Message Passing (\Jmp) primitive, which plays a crucial role in obtaining the strongest security guarantee of GOD, followed by our protocols in the three server setting. 

\paragraph{Secret Sharing Semantics}
We use the following secret-sharing semantics.
\begin{newitemize}
	%
	\item[$\circ$] {\em $\sqr{\cdot}$-sharing:}
	A value $\val \in \Z{\ell}$ is $\sqr{\cdot}$-shared among $P_1,P_2$, if 
	$P_s$ for $s \in \{1,2\}$ holds $\sqr{\val}_s \in \Z{\ell}$ such that $\val = \sqrA{\val} + \sqrB{\val} $.
	\item[$\circ$] {\em $\sgr{\cdot}$-sharing:}
	A value $\val \in \Z{\ell}$ is $\sgr{\cdot}$-shared among $\Partyset$, if
	\begin{mylist}
		\item[--] there exists $\val_0,\val_1,\val_2 \in \Z{\ell}$ such that $\val = \val_0+\val_1+\val_2$.
		\item[--] $P_s$ holds $(\val_s, \val_{(s+1)\%3})$ for $s \in \{0,1,2\}$.
	\end{mylist}
	\item[$\circ$] {\em $\shr{\cdot}$-sharing:}
	A value $\val \in \Z{\ell}$ is $\shrd$-shared among $\Partyset$, if 
	\begin{mylist}
		\item[--] there exists $\av{\val} \in \Z{\ell}$ that is $\sqr{\cdot}$-shared among $\ESet$.
		\item[--] there exists $\bv{\val}, \gv{\val} \in \Z{\ell}$ such that $\bv{\val} = \val + \av{\val}$ and $P_0$ holds $(\sqrA{\av{\val}}, \sqrB{\av{\val}},\bv{\val}+\gv{\val})$ while $P_s$ for $s \in \EInSet$ holds $( \sqr{\av{\val}}_s, \bv{\val}, \gv{\val})$.
	\end{mylist}
\end{newitemize}

\paragraph{Arithmetic and Boolean Sharing} 
{\em Arithmetic} sharing refers to sharing over $\Z{\ell}$ while {\em boolean} sharing, denoted as $\shrB{\cdot}$, refers to sharing over $\Z{1}$.

\paragraph{Linearity of the Secret Sharing Scheme}
Given $\sqr{\cdot}$-shares of $\val_1,\val_2$, and public constants $c_1, c_2$, servers can locally compute $\sqr{\cdot}$-share of $c_1\val_1+c_2\val_2$ as $c_1\sqr{\val_1}+c_2\sqr{\val_2}$. It is trivial to see that linearity property is satisfied by $\sgr{\cdot}$ and $\shr{\cdot}$ sharings. 

\subsection{Joint Message Passing  primitive}
\label{subsec:jmp}
The $\Jmp$ primitive allows two servers to relay a common message to the third server such that either the relay is successful or an honest server (or a conflicting pair) is identified. The striking feature of $\Jmp$  is that it offers a rate-$1$ communication i.e. for a message of $\ell$ elements, it only incurs a communication of $\ell$ elements (in an amortized sense). 
The task of $\Jmp$ is  captured in an ideal functionality (\boxref{fig:JmpFunc}) and the protocol for the same appears in \boxref{fig:piJmp}. Next, we give an overview.

Given two servers $P_i, P_j$ possessing a common value $\val \in \Z{\ell}$, protocol $\piJmp$ proceeds as follows. First, $P_i$ sends $\val$ to $P_k$ while $P_j$ sends a hash of $\val$ to $P_k$. The communication of the hash is done once and for all from $P_j$ to $P_k$. In the simplest case,  $P_k$ receives a consistent (value, hash) pair, and the protocol terminates. In all other cases,  a  $\TTP$ is identified as follows without having to communicate $\val$ again.  Importantly, the following part can be run once and for all instances of $\piJmp$  with $P_i,P_j,P_k$ in the same roles, invoked in the final 3PC protocol. Consequently, the  cost relevant to this part vanishes in an amortized sense, making the construction rate-1.

\begin{systembox}{$\FJmp$}{3PC: Ideal functionality for $\Jmp$ primitive}{fig:JmpFunc}
	\justify
	$\FJmp$ interacts with the servers in $\Partyset$ and the adversary $\Sim$. 
	\begin{myitemize}
		\item[\bf Step 1:] $\FJmp$ receives $(\INPUT,\val_s)$ from $P_s$ for $s \in \{i,j\}$, while it receives $(\SELECT,\ttp)$ from $\Sim$. Here $\ttp$ denotes the server that $\Sim$ wants to choose as the $\TTP$.  Let $P^{\star} \in \Partyset$ denote the server corrupted by $\Sim$.
		\item[\bf Step 2:] If $\val_i = \val_j$ and $\ttp =\bot$, then set $\msg_i = \msg_j = \bot, \msg_k = \val_i$ and go to {\bf Step 5}.
		\item[\bf Step 3:] If $\ttp \in \Partyset\setminus\{P^{\star}\}$, then set $\msg_i = \msg_j = \msg_k = \ttp$.
		\item[\bf Step 4:] Else, $\TTP$ is set to be the honest server with smallest index. Set $\msg_i = \msg_j = \msg_k = \TTP$
		\item[\bf Step 5:] Send $(\OUTPUT, \msg_s)$ to $P_s$ for $s \in \{0,1,2\}$.
	\end{myitemize}
\end{systembox}

Each $P_s$ for $s \in \{i,j,k\}$ maintains a bit $\bitb_s$ initialized to $0$, as an indicator for inconsistency. When $P_k$  receives an inconsistent (value, hash) pair, it sets $\bitb_k = 1$ and sends the bit to both $P_i,P_j$, who cross-check with each other by exchanging the bit and turn on their inconsistency bit  if the bit received from either $P_k$ or its fellow sender is turned on.  A server  broadcasts a hash of its value when its  inconsistency bit is on;\footnote{hash can be computed on a combined message across many calls of $\Jmp$.} $P_k$'s value is the one it receives from $P_i$. At this stage, there are a bunch of possible cases and a detailed analysis determines an eligible $\TTP$ in each case.  

When $P_k$ is silent, the protocol is understood to be complete. This is fine irrespective of the status of $P_k$-- an honest $P_k$ never skips this broadcast with inconsistency bit on, and a corrupt $P_k$ implies honest senders. If either $P_i$ or  $P_j$ is silent, then $P_k$ is picked as $\TTP$ which is surely honest. A corrupt $P_k$ could not make one of $\{P_i, P_j\}$ speak, as the senders (honest in this case)  are in agreement on their inconsistency bit (due to their mutual exchange of inconsistency bit). When all of them speak and (i) the senders' hashes do not match, $P_k$  is picked as $\TTP$;  (ii) one of the senders conflicts with $P_k$, the other sender is picked as $\TTP$;  and lastly  (iii) if there is no conflict, $P_i$ is picked as $\TTP$. The first two cases are self-explanatory.  In the last case, either $P_j$ or $P_k$ is corrupt.  If not, a corrupt $P_i$ can have honest $P_k$ speak (and hence turn on its inconsistency bit), by sending a $\val'$ whose hash is not  same as that of $\val$ and so inevitably, the hashes of honest $P_j$ and $P_k$ will conflict, contradicting (iii).  
As a final touch, we ensure that,  in each step,  a server raises a public alarm (via broadcast) accusing a server which is silent when it is not supposed to be,  and the protocol terminates immediately by labelling the server as $\TTP$ who is neither the complainer nor the accused. 
\begin{notation} \label{jmp-send}
	We say that $P_i,P_j$ $\jsend$ $\val$ to $P_k$ when they invoke $\piJmp(P_i, P_j, P_k, \val)$. 
\end{notation}

\begin{mypbox}{$\piJmp(P_i, P_j, P_k, \val)$}{3PC: Joint Message Passing Protocol}{fig:piJmp}
	Each server $P_s$ for $s \in \{i,j,k\}$ initializes bit $\bitb_s = 0$.
	\justify
	{\em Send Phase:} $P_i$ sends $\val$ to $P_k$.
	
	\noindent {\em Verify Phase:} $P_j$ sends $\Hash(\val)$ to $P_k$. 
	\begin{myitemize}
		\item[--]   $P_k$ broadcasts \texttt{"(accuse,$P_i$)"}, if $P_i$ is silent and $\TTP$ = $P_j$. Analogously for $P_j$. If $P_k$ accuses both $P_i,P_j$, then $\TTP$ = $P_i$. Otherwise,  $P_k$ receives some $\tilde \val$ and  either sets $\bitb_k = 0$ when the value and the hash are consistent or  sets $\bitb_k = 1$. $P_k$ then sends $\bitb_k$ to $P_i,P_j$ and terminates if $\bitb_k = 0$.
		
		\item[--] If $P_i$ does not receive a bit from $P_k$, it broadcasts \texttt{"(accuse,$P_k$)"} and $\TTP$ = $P_j$. Analogously for $P_j$. If both $P_i,P_j$ accuse $P_k$, then $\TTP$ = $P_i$. Otherwise,  $P_s$ for $s \in \{i,j\}$  sets $\bitb_s = \bitb_k$.
		\item[--] $P_i,P_j$ exchange their bits to each other.  If $P_i$ does  not receive $\bitb_j$ from $P_j$, it broadcasts \texttt{"(accuse,$P_j$)"} and  $\TTP$ = $P_k$. Analogously for $P_j$. Otherwise, $P_i$ resets its bit to $\bitb_i \vee \bitb_j$ and likewise $P_j$ resets its bit to $\bitb_j \vee \bitb_i$.
		\item[--] $P_s$ for $s \in \{i,j,k\}$ broadcasts $\Hash_s = \Hash(\val^*)$ if $b_s = 1$, where $\val^* = \val$ for $s \in \{i,j\}$ and $\val^* = \tilde \val$ otherwise.  If $P_k$ does not broadcast, terminate.  If either  $P_i$ or $P_j$ does not broadcast, then  $\TTP$ = $P_k$. Otherwise,
		\begin{inneritemize}
			\item If $\Hash_i \neq \Hash_j$: $\TTP$ = $P_k$.
			\item Else if $\Hash_i \neq \Hash_k$: $\TTP$ = $P_j$.
			\item Else if $\Hash_i = \Hash_j = \Hash_k$: $\TTP$ = $P_i$.
		\end{inneritemize}	
	\end{myitemize}
\end{mypbox}

{\em Using $\Jmp$ in protocols.} As mentioned in the introduction, the $\Jmp$ protocol needs to be viewed as consisting of two phases ({\em send, verify}), where {\em send} phase consists of $P_i$ sending $\val$ to $P_k$ and the rest  goes to  {\em verify} phase. Looking ahead,  most of our protocols use $\Jmp$, and consequently, our final construction, either of general MPC or any PPML task, will have several calls to $\Jmp$. To leverage amortization,  the {\em send} phase will be executed in all protocols invoking $\Jmp$ on the flow, while the {\em verify} for a fixed ordered pair of senders will be executed once and for all in the end. The {\em verify} phase will determine if all the sends were correct. If not, a $\TTP$ is identified, as explained, and the computation completes with the help of $\TTP$, just as in the ideal-world.  

\begin{lemma}[Communication]
	\label{app:piJmp}
	Protocol $\piJmp$ (\boxref{fig:piJmp}) requires $1$ round and an amortized communication of $\ell$ bits overall.
\end{lemma}
\begin{proof}
	Server $P_i$ sends value $\val$ to $P_k$ while $P_j$ sends hash of the same to $P_k$. This accounts for one round and communication of $\ell$ bits. $P_k$ then sends back its inconsistency bit to $P_i,P_j$, who then exchange it; this takes another two rounds. This is followed by servers broadcasting hashes on their values and selecting a $\TTP$ based on it, which takes one more round. All except the first round can be combined for several instances of $\piJmp$ protocol and hence the cost gets amortized.
\end{proof}

\subsection{3PC Protocols}
\label{subsec:3pc_protocols}
We now describe the protocols for 3 parties/servers and defer the security proofs to \S\ref{app:sec3pc}. 

\paragraph{Sharing Protocol} 
Protocol $\piSh$ (\boxref{fig:piIn}) allows a server $P_i$ to generate $\shr{\cdot}$-shares of a value $\val \in \Z{\ell}$. 
In the preprocessing phase, $P_0, P_j$ for $j \in \{1,2\}$ along with $P_i$ sample a random $\sqr{\av{\val}}_j \in \Z{\ell}$, while $P_1, P_2, P_i$ sample random $\gv{\val} \in \Z{\ell}$. This allows $P_i$ to know both $\av{\val}$ and $\gv{\val}$ in clear. 
During the online phase, if $P_i = P_0$, then $P_0$ sends $\bv{\val} = \val + \av{\val}$ to $P_1$. $P_0, P_1$ then $\jsend$ $\bv{\val}$ to $P_2$ to complete the secret sharing. If $P_i = P_1$, $P_1$ sends $\bv{\val} = \val + \av{\val}$ to $P_2$. Then $P_1, P_2$ $\jsend$ $\bv{\val} + \gv{\val}$ to $P_0$. The case for $P_i = P_2$ proceeds similar to that of $P_1$.
The correctness of the shares held by each server is assured by the guarantees of $\piJmp$.  

\begin{protocolbox}{$\piSh(P_i,\val)$}{3PC: Generating $\shr{\val}$-shares by server $P_i$}{fig:piIn}
	\justify
	\algoHead{Preprocessing:}
	\begin{myitemize}
		\item[--] If $P_i = P_0$ : $P_0,P_j$, for $j \in \EInSet$, together sample random $\sqr{\av{\val}}_j \in \Z{\ell}$, while $\Partyset$ together sample random $\gv{\val} \in \Z{\ell}$.
		\item[--] If $P_i = P_1$ : $P_0,P_1$ together sample random $\sqr{\av{\val}}_1 \in \Z{\ell}$, while $\Partyset$ together sample a random $\sqr{\av{\val}}_2 \in \Z{\ell}$. Also, $\ESet$ together sample random $\gv{\val} \in \Z{\ell}$.
		\item[--] If $P_i = P_2$: Symmetric to the case when $P_i = P_1$.
	\end{myitemize}
	\justify\vspace{-3mm}
	\algoHead{Online:}
	\begin{myitemize}
		\item[--] If $P_i = P_0$ : $P_0$ computes $\bv{\val} = \val+\av{\val}$ and sends $\bv{\val}$ to $P_1$. $P_1,P_0$ $\jsend$ $\bv{\val}$ to $P_2$.
		\item[--] If $P_i = P_j$, for $j \in\EInSet$ : $P_j$ computes $\bv{\val} = \val+\av{\val}$, sends $\bv{\val}$ to $P_{3-j}$. $P_1,P_2$ $\jsend$  $\bv{\val}+\gv{\val}$ to $P_0$.
	\end{myitemize}
\end{protocolbox}

\begin{lemma}[Communication]
	\label{app:piSh}
	Protocol $\piSh$ (\boxref{fig:piIn}) is non-interactive in the preprocessing phase and requires $2$ rounds and an amortized communication of $2 \ell$ bits in the online phase.
\end{lemma}
\begin{proof}
	During the preprocessing phase, servers non-interactively sample the $\sqr{\cdot}$-shares of $\av{\val}$ and $\gv{\val}$ values using the shared key setup. In the online phase, when $P_i = P_0$, it computes $\bv{\val}$ and sends it to $P_1$, resulting in one round and $\ell$ bits communicated. They then $\jsend$ $\bv{\val}$ to $P_2$, which requires additional one round in an amortized sense, and $\ell$ bits to be communicated. For the case when $P_i = P_1$, it sends $\bv{\val}$ to $P_2$, resulting in one round and a communication of $\ell$ bits. Then, $P_1, P_2$ $\jsend$ $\bv{\val}+\gv{\val}$ to $P_0$. This again requires an additional one round and $\ell$ bits. The analysis is similar in the case of $P_i = P_2$.
\end{proof}

\paragraph{Joint Sharing Protocol}
Protocol $\piJSh$ (\boxref{fig:piJSh})  allows two servers $P_i, P_j$ to jointly generate a $\shrd$-sharing of a value $\val \in \Z{\ell}$ that is known to both. Towards this, servers execute the preprocessing of $\piSh$ (\boxref{fig:piIn}) to generate $\sqr{\av{\val}}$ and $\gv{\val}$. If $(P_i, P_j )= (P_1, P_0)$, then $P_1,P_0$ $\jsend$ $\bv{\val} = \val + \av{\val}$ to $P_2$. The case when $(P_i, P_j )= (P_2, P_0)$ proceeds similarly. The case for $(P_i, P_j )= (P_1, P_2)$ is optimized further as follows: servers locally set $\sqrA{\av{\val}} = \sqrB{\av{\val}} = 0$. $P_1, P_2$ together sample random $\gv{\val} \in \Z{\ell}$, set $\bv{\val} = \val$ and $\jsend$ $\bv{\val} + \gv{\val}$ to $P_0$. 

\begin{protocolbox}{$\piJSh(P_i, P_j, \val)$}{3PC: $\shrd$-sharing of a value $\val \in \Z{\ell}$ jointly by $P_i, P_j$}{fig:piJSh}
	\justify
	\algoHead{Preprocessing:}
	\begin{myitemize}
		\item[--] If $(P_i, P_j) = (P_1, P_0)$: Servers execute the preprocessing of $\piSh(P_1,\val)$ and then locally set $\gv{\val} = 0$. 
		\item[--] If $(P_i, P_j) = (P_2, P_0)$: Similar to the case above.
		\item[--] If $(P_i, P_j) = (P_1, P_2)$: $P_1, P_2$ together sample random $\gv{\val} \in \Z{\ell}$. Servers locally set $\sqrA{\av{\val}} = \sqrB{\av{\val}} = 0$. 
	\end{myitemize}
	\justify\vspace{-3mm}    
	\algoHead{Online:}
	\begin{myitemize}
		\item[--] If $(P_i, P_j) = (P_1, P_0)$: $P_0, P_1$ compute $\bv{\val} = \val + \sqrA{\av{\val}} + \sqrB{\av{\val}}$. $P_0,P_1$ $\jsend$ $\bv{\val}$ to $P_2$.	
		\item[--] If $(P_i, P_j) = (P_2, P_0)$: Similar to the case above.
		\item[--] If $(P_i, P_j) = (P_1, P_2)$: $P_1,P_2$ locally set $\bv{\val} = \val$. $P_1,P_2$ $\jsend$ $\bv{\val}+\gv{\val}$ to $P_0$.	
	\end{myitemize}
\end{protocolbox}

When the value $\val$ is available to both $P_i, P_j$ in the preprocessing phase, protocol $\piJSh$ can be made non-interactive in the following way: $\Partyset$ sample a random $\vr \in \Z{\ell}$ and locally set their share according to  \tabref{SharesAssign}.

\begin{table}[htb!]
	\centering
	\resizebox{.98\textwidth}{!}{
		\begin{tabular}{c | c | c | c }
			\toprule
			& $(P_1, P_2)$ & $(P_1, P_0)$   & $(P_2, P_0)$\\
			\midrule
			& $\begin{aligned} \sqrA{\av{\val}} = 0, &~\sqrB{\av{\val}} = 0\\ \bv{\val} = \val, &~\gv{\val} = \vr - \val \end{aligned}$ 
			& $\begin{aligned} \sqrA{\av{\val}} = -\val, &~\sqrB{\av{\val}} = 0\\ \bv{\val} = 0, &~\gv{\val} = \vr \end{aligned}$ 
			& $\begin{aligned} \sqrA{\av{\val}} = 0, &~\sqrB{\av{\val}} = -\val\\ \bv{\val} = 0, &~\gv{\val} = \vr  \end{aligned}$ \\
			\midrule
			
			$\begin{aligned} P_0 \\ P_1 \\ P_2 \end{aligned}$ 
			& $\begin{aligned} (0, ~0, ~\vr~~~~~) \\ (0, ~\val, ~\vr - \val) \\ (0, ~\val, ~\vr - \val) \end{aligned}$ 
			& $\begin{aligned} (-\val, ~0, ~\vr)  \\ (-\val, ~0, ~\vr) \\ (~~0, ~0, ~\vr)           \end{aligned}$ 
			& $\begin{aligned} (0, ~-\val, ~\vr)  \\ (0,~~~~0, ~\vr)  \\ (0, ~-\val, ~\vr)       \end{aligned}$  \\
			\bottomrule
		\end{tabular}
	}
	\caption{\small The columns depict the three distinct possibility of input contributing pairs. The first row shows the assignment to various components of the sharing. The last row, along with three sub-rows, specify the  shares held by the three servers.\label{tab:SharesAssign}}
\end{table}

\begin{lemma}[Communication]
	\label{app:piJSh}
	Protocol $\piJSh$ (\boxref{fig:piJSh})  is non-interactive in the preprocessing phase and requires $1$ round and an amortized communication of $\ell$ bits in the online phase.
\end{lemma}
\begin{proof}
	In this protocol, servers execute $\piJmp$ protocol once. Hence the overall cost follows from that of an instance of the $\piJmp$ protocol (Lemma~\ref{app:piJmp}).
\end{proof}

\paragraph{Addition Protocol} 
Given $\shr{\cdot}$-shares on input wires $\wx, \wy$, servers can use linearity property of the sharing scheme  to locally compute $\shr{\cdot}$-shares of the output of addition gate, $\wz = \wx + \wy$ as $\shr{\wz} = \shr{\wx}+\shr{\wy}$.

\paragraph{Multiplication Protocol}
Protocol $\piMult(\Partyset, \shr{\wx}, \shr{\wy})$  (\boxref{fig:piMult}) enables the servers in $\Partyset$ to compute $\shr{\cdot}$-sharing of $\wz = \wx  \wy$, given the $\shrd$-sharing of $\wx$ and $\wy$. We build on the protocol of  BLAZE~\cite{BLAZE}  and discuss along the way the differences and resemblances.
We begin with a protocol for the semi-honest setting, which is also the starting point of BLAZE. During the preprocessing phase, $P_0, P_j$ for $j \in \EInSet$ sample random $\sqr{\av{\wz}}_{j} \in \Z{\ell}$, while $\ESet$ sample random $\gv{\wz} \in \Z{\ell}$. In addition, $P_0$ locally computes $\Gammaxy = \av{\wx}\av{\wy}$ and generates $\sqr{\cdot}$-sharing of the same between $\ESet$. Since, 

\medskip
\resizebox{.94\linewidth}{!}{
	\begin{minipage}{\linewidth}
		\begin{align}
			\bv{\wz} &= \wz + \av{\wz} = \wx \wy + \av{\wz} = (\bv{\wx} - \av{\wx})(\bv{\wy} - \av{\wy}) + \av{\wz} \nonumber\\
			&= \bv{\wx}\bv{\wy} - \bv{x}\av{\wy} - \bv{y}\av{\wx} + \Gammaxy + \av{\wz} \label{eq:betaz}
		\end{align}
\end{minipage}}
 
\medskip
servers $\ESet$ locally compute $\sqr{\bv{\wz}}_j = (j-1)\bv{\wx}\bv{\wy} - \bv{x}\sqr{\av{\wy}}_{j} - \bv{y}\sqr{\av{\wx}}_{j} + \sqr{\Gammaxy}_j + \sqr{\av{\wz}}_{j}$ during the online phase and  mutually exchange their shares to reconstruct $\bv{\wz}$. $P_1$ then sends $\bv{\wz} + \gv{\wz}$ to $P_0$, completing the semi-honest protocol. The  correctness that asserts $\wz = \wx \wy$ or in other words $\bv{\wz} - \av{\wz}= \wx \wy$ holds due to Eq.~\ref{eq:betaz}. 

The following issues arise in the above protocol when a malicious adversary is considered: \vspace{-2mm}
\begin{myenumlist}
	\item When $P_0$ is corrupt, the $\sqr{\cdot}$-sharing of $\Gammaxy$ performed by $P_0$ 
	might not be correct, i.e.\ $\Gammaxy \neq \av{\wx}\av{\wy}$.
	\item When $P_1$ (or $P_2$) is corrupt, $\sqr{\cdot}$-share of $\bv{\wz}$ handed over to the fellow honest evaluator during the online phase might not be correct, causing reconstruction of an incorrect $\bv{\wz}$. 
	\item When $P_1$ is corrupt, the value $\bv{\wz} + \gv{\wz}$ that is sent to $P_0$ during the online phase may not be correct.
\end{myenumlist}

\smallskip
All the three issues are common with BLAZE (copied verbatim), but we differ from BLAZE in handling them. We begin with solving the last issue first. We simply make $P_1, P_2$ $\jsend$ $\bv{\wz}+\gv{\wz}$ to $P_0$ (after $\bv{\wz}$ is computed). This either leads to success or a $\TTP$ selection.  Due to $\Jmp$'s rate-1 communication, $P_1$ alone sending the value to $P_0$ remains as costly as  using $\Jmp$ in amortized sense. Whereas in BLAZE, the malicious version simply  makes $P_2$ to send a hash of $\bv{\wz} + \gv{\wz}$ to $P_0$ (in addition to $P_1$'s communication  of  $\bv{\wz} + \gv{\wz}$ to $P_0$), who $\abort$s if the received values are inconsistent. 

For the remaining two issues, similar to BLAZE, we reduce both to  a multiplication (on values unrelated to inputs) in the preprocessing phase. However, our method leads to either success or  $\TTP$ selection, with no additional cost. 

We start with the second issue. To solve it, where a corrupt $P_1$ (or $P_2$) sends an incorrect $\sqr{\cdot}$-share of $\bv{\wz}$, BLAZE makes use of server $P_0$ to compute a version of $\bv{\wz}$ for verification, based on $\bv{\wx}$ and $\bv{\wy}$, as follows. Using $ \bv{\wx} + \gv{\wx}$,  $ \bv{\wy} + \gv{\wy}$, $\av{\wx}$, $\av{\wy}$, $\av{\wz}$ and $\Gammaxy$, $P_0$ computes: 

\medskip 
\resizebox{.94\linewidth}{!}{
	\begin{minipage}{\linewidth}
		\begin{align*}
			\starbeta{\wz} &= -(\bv{\wx} + \gv{\wx})\av{\wy} - (\bv{\wy} + \gv{\wy})\av{\wx} + 2\Gammaxy + \av{\wz}\\
			&= (-\bv{\wx}\av{\wy} - \bv{\wy}\av{\wx} + \Gammaxy + \av{\wz}) - (\gv{\wx}\av{\wy} + \gv{\wy}\av{\wx} - \Gammaxy)\\
			&= (\bv{\wz} - \bv{\wx}\bv{\wy}) - (\gv{\wx}\av{\wy} + \gv{\wy}\av{\wx} - \Gammaxy )  \hspace*{23pt}[\text{by Eq. \ref{eq:betaz}}]\\
			&= (\bv{\wz} - \bv{\wx}\bv{\wy} ) - \Chi   \hspace*{13pt}[\text{where } \Chi = \gv{\wx}\av{\wy} + \gv{\wy}\av{\wx} - \Gammaxy ]
		\end{align*}
\end{minipage}}
\medskip 

Now if $\Chi$ can be made available to $P_0$, it can send $\starbeta{\wz} + \Chi$ to $P_1$ and $P_2$ who using  the knowledge of $\bv{\wx}, \bv{\wy}$, can verify the correctness of $\bv{\wz}$ by computing $\bv{\wz} - \bv{\wx}\bv{\wy}$ and checking against the value $\starbeta{\wz} + \Chi$ received from $P_0$.  However, disclosing $\Chi$ on clear to $P_0$ will cause a privacy issue when $P_0$ is corrupt, because one degree of freedom on the pair $(\gv{\wx},\gv{\wy})$ is lost and the same impact percolates down to   $(\bv{\wx},\bv{\wy})$ and further to the actual values $(\val_\wx,\val_\wy)$ on the wires $\wx,\wy$. This is resolved through a random value $\psi \in \Z{\ell}$, sampled together by $P_1$ and $P_2$. Now, $\Chi$ and  $\starbeta{\wz}$ are set to $ \gv{\wx}\av{\wy} + \gv{\wy}\av{\wx} - \Gammaxy + \psi$,  $  (\bv{\wz} - \bv{\wx}\bv{\wy} + \psi) - \Chi$, respectively and the check by $P_1,P_2$ involves computing $\bv{\wz} - \bv{\wx}\bv{\wy} + \psi$. The rest of the logic in BLAZE goes on to discuss how to enforce  $P_0$-- (a) to compute a correct $\Chi$ (when honest), and (b) to share correct $\Gammaxy$ (when corrupt). Tying the ends together, they identify the precise shared multiplication triple and map its components to $\Chi$ and $\Gammaxy$ so that these values are correct by virtue of the correctness of the product relation. This reduces ensuring the correctness of these values to doing a single multiplication of two values in the preprocessing phase.

\smallskip
\begin{protocolbox}{$\piMult(\Partyset, \shr{\wx}, \shr{\wy})$}{3PC: Multiplication Protocol ($\wz = \wx \cdot \wy$)}{fig:piMult}
	\justify
	\algoHead{Preprocessing:}
	\begin{myitemize}
		\item[--] $P_0, P_j$ for $j \in \EInSet$ together sample random $\sqr{\av{\wz}}_{j} \in \Z{\ell}$, while $\ESet$ sample random $\gv{\wz} \in \Z{\ell}$. 
		\item[--] Servers in $\Partyset$ locally compute $\sgr{\cdot}$-sharing of $\md = \gv{\wx}+\av{\wx}$ and $\me = \gv{\wy}+\av{\wy}$ by setting the shares as follows~(ref. \tabref{shareZK}):
		\begin{equation*}
			(\md_0 \text{=} \sqr{\av{\wx}}_{2},
			\md_1 \text{=} \sqr{\av{\wx}}_{1}, 
			\md_2 \text{=} \gv{\wx}),~~
			(\me_0 \text{=} \sqr{\av{\wy}}_{2}, 
			\me_1 \text{=} \sqr{\av{\wy}}_{1}, 
			\me_2 \text{=} \gv{\wy})
		\end{equation*}
		\item[--] Servers in $\Partyset$ execute $\piMulPre(\Partyset, \md, \me)$ to generate $\sgr{\mf} = \sgr{\md \me}$.
		\item[--] $P_0, P_1$ locally set $\sqr{\Chi}_1 = \mf_1$, while $P_0, P_2$ locally set $\sqr{\Chi}_2 = \mf_0$. $P_1, P_2$ locally compute $\psi = \mf_2- \gv{\wx}\gv{\wy}$.
	\end{myitemize} 
	\justify\vspace{-3mm}
	\algoHead{Online:}
	\begin{myitemize}
		\item[--] $P_0,P_j$, for $j \in \EInSet$, compute $\sqr{\starbeta{\wz}}_j = - (\bv{\wx}+\gv{\wx})\sqr{\av{\wy}}_j - (\bv{\wy}+\gv{\wy})\sqr{\av{\wx}}_j + \sqr{\av{\wz}}_j + \sqr{\Chi}_j$.
		\item[--] $P_0,P_1$ $\jsend$ $\sqr{\starbeta{\wz}}_1$ to $P_2$ and  $P_0,P_2$ $\jsend$ $\sqr{\starbeta{\wz}}_2$ to $P_1$.
		\item[--] $P_1, P_2$ compute $\starbeta{\wz} = \sqr{\starbeta{\wz}}_1 +\sqr{\starbeta{\wz}}_2$ and set $\bv{\wz} = \starbeta{\wz} +  \bv{\wx}\bv{\wy} + \psi$.
		\item[--] $P_1,P_2$ $\jsend$ $\bv{\wz}+\gv{\wz}$ to $P_0$.
	\end{myitemize}
\end{protocolbox}
\smallskip 

We differ from BLAZE in several ways. First, we do not simply rely on $P_0$ for the verification information $\starbeta{\wz} + \Chi$, as this may inevitably  lead to abort when $P_0$ is corrupt.  Instead, we find  (a slightly different)  $\starbeta{\wz}$ that, instead of entirely available to $P_0$, will be  available  in $\sqr{\cdot}$-shared form between the two teams $\{P_0,P_1\},\{P_0,P_2\} $, with both servers in $\{P_0,P_i\}$ holding $i$th share $\sqrV{\starbeta{\wz}}{i}$. With this edit,  the $i$th team can $\jsend$ the $i$th share of  $\starbeta{\wz}$ to the third server which computes $\starbeta{\wz}$. Due to the presence of one honest server in each team,  this $\starbeta{\wz}$ is correct and  $P_1,P_2$ directly use it to compute $\bv{\wz}$, with the knowledge of $\psi,\bv{\wx},\bv{\wy}$.  The outcome of our approach is a win-win situation  i.e. either success or $\TTP$ selection. Our approach of computing $\bv{\wz}$ from $\starbeta{\wz}$ is a departure from BLAZE, where the latter suggests computing $\bv{\wz}$ from  the  exchange $P_1,P_2$'s respective share of $\bv{\wz}$ (as in the semi-honest construction)  and use $\starbeta{\wz}$ for verification. Our new $\starbeta{\wz}$ and $\Chi$ are:

\bigskip 
\resizebox{0.94\linewidth}{!}{
	\begin{minipage}{\linewidth}
		\begin{align*} \label{eq:3}
			\Chi &= \gv{\wx}\av{\wy} + \gv{\wy}\av{\wx} + \GammaV{\wx\wy} - \psi \hspace{3mm}\text{and}\\
			\starbeta{\wz} &= -(\bv{\wx}+\gv{\wx})\av{\wy}-(\bv{\wy}+\gv{\wy})\av{\wx}+\av{\wz}+\Chi\\
			&= (-\bv{\wx}\av{\wy} - \bv{\wy}\av{\wx} + \Gammaxy + \av{\wz}) - \psi
			= \bv{\wz} - \bv{\wx}\bv{\wy} - \psi
		\end{align*}
\end{minipage}}
\bigskip 

Clearly, both $P_0$ and $P_i$ can compute $ \sqrV{\starbeta{\wz}}{i} = -(\bv{\wx}+\gv{\wx})\sqr{\av{\wy}}_i-(\bv{\wy}+\gv{\wy})\sqr{\av{\wx}}_i+\sqr{\av{\wz}}_i+\sqr{\Chi}_i$ given $\sqr{\Chi}_i$.   The rest of our discussion explains  how (a)  $i$th share of $\sqr{\Chi}$ can be made available to $\{P_0,P_i\}$ and  (b) $\psi$ can be derived by $P_1,P_2$,  from a multiplication triple. 
Similar to BLAZE,  yet for a different triple, we observe that   $(\md, \me, \mf)$ is a multiplication triple, where $\md = (\gv{\wx}+\av{\wx}), \me =  (\gv{\wy}+\av{\wy}), \mf =(\gv{\wx}\gv{\wy}+\psi)+\Chi$ if and only if $\Chi$ and $\GammaV{\wx\wy}$ are correct. Indeed,

\vspace{-1mm}
\begin{align*}
	\md\me &= (\gv{\wx}+\av{\wx})(\gv{\wy}+\av{\wy}) = \gv{\wx}\gv{\wy} + \gv{\wx}\av{\wy} + \gv{\wy}\av{\wx} + \GammaV{\wx\wy} \\
	&= (\gv{\wx}\gv{\wy} + \psi) + (\gv{\wx}\av{\wy} + \gv{\wy}\av{\wx} + \GammaV{\wx\wy} - \psi) \\
	&= (\gv{\wx}\gv{\wy} + \psi) +\Chi = \mf
\end{align*}
Based on this observation, we compute the above multiplication triple using a multiplication protocol and extract out the values for $ \psi$ and $\Chi$ from the shares of $\mf$ which are bound to be correct. This  can be executed entirely in the preprocessing phase. Specifically, the servers (a) locally obtain $\sgr{\cdot}$-shares of $\md, \me$ as in  \tabref{shareZK}, (b) compute  $\sgr{\cdot}$-shares of $\mf (= \md \me)$, say denoted by $\mf_0,\mf_1,\mf_2$, using an efficient, robust 3-party multiplication protocol, say $\piMulPre$~(abstracted in a functionality \boxref{fig:FMulPre}) and finally (c)  extract out the required preprocessing data {\em locally} as in Eq.~\ref{eq:premap}. We switch to $\sgr{\cdot}$-sharing in this part to be able to use the best robust multiplication protocol of \cite{BGIN19} that supports this form of secret sharing and requires communication of just $3$ elements. Fortunately, the switch does not cost anything, as both the step (a) and (c) (as above) involve local computation and the cost simply reduces to a single run of a multiplication protocol. 

\begin{table}[htb!]
	\centering
	\begin{tabular}{c | c | c | c}
		& $P_0$ & $P_1$ & $P_2$ \\ 
		\small{$\sgr{\val}$} & \small{$(\val_0,\val_1)$} & \small{$(\val_1,\val_2)$} & \small{$(\val_2,\val_0)$}\\
		\toprule
		$\sgr{\md}$	 & $(\sqr{\av{\wx}}_{2},\sqr{\av{\wx}}_{1})$ &$(\sqr{\av{\wx}}_{1},\gv{\wx})  $& $(\gv{\wx}, \sqr{\av{\wx}}_{2})  $\\ 
		\midrule
		$\sgr{\me}$	 & $(\sqr{\av{\wy}}_{2},\sqr{\av{\wy}}_{1})$ &$(\sqr{\av{\wy}}_{1},\gv{\wy})  $& $(\gv{\wy},\sqr{\av{\wy}}_{2})  $\\ 
		\bottomrule
	\end{tabular}
	\vspace{-2mm}
	\caption{\small The $\sgr{\cdot}$-sharing of values $\md$ and $\me$ \label{tab:shareZK}}
\end{table}

\begin{equation}\label{eq:premap}
	\sqr{\Chi}_2 \leftarrow  \mf_0, 
	~~  \sqr{\Chi}_1 \leftarrow \mf_1, 
	~~  \gv{\wx}\gv{\wy} +\psi  \leftarrow \mf_2.
\end{equation}

According to $\sgr{\cdot}$-sharing, both $P_0$ and $P_1$ obtain $\mf_1$ and hence obtain $\sqr{\Chi}_1$. Similarly, $P_0, P_2$ obtain $\mf_0$ and hence $\sqr{\Chi}_2$. Finally, $P_1, P_2$ obtain $\mf_2$ from which they compute $\psi = \mf_2 - \gv{\wx}\gv{\wy}$. This completes the informal discussion.  

To leverage amortization, the {\em send} phase of $\jsend$ alone is executed on the fly and {\em verify} is performed once for multiple instances of $\jsend$. Further, observe that $P_1,P_2$ possess the required shares in the online phase to compute the entire circuit. Hence, $P_0$ can come in only during {\em verify} of $\jsend$ towards $P_1, P_2$, which can be deferred towards the end. Hence, the $\jsend$ of $\bv{\vz} + \gv{\vz}$ to $P_0$ (enabling computation of the verification information) can be performed once, towards the end, thereby requiring a single round for sending $\bv{\vz} + \gv{\vz}$ to $P_0$ for multiple instances. Following this, the {\em verify} of $\jsend$ towards $P_0$ is performed first, followed by  performing the {\em verify} of $\jsend$ towards $P_1, P_2$ in parallel. 

We note that to facilitate a fast online phase for multiplication, our  preprocessing phase leverages a robust multiplication protocol \cite{BGIN19}  in a black-box manner to derive the necessary preprocessing information.  A similar black-box approach is also taken for the dot product protocol in the preprocessing phase. This leaves room for further improvements in the communication cost, which can be obtained by instantiating the black-box with an efficient, robust protocol coupled with the fast online phase.

\smallskip
\begin{systembox}{$\FMulPre$}{3PC: Ideal functionality for $\piMulPre$ protocol}{fig:FMulPre}
	\justify
	$\FMulPre$ interacts with the servers in $\Partyset$ and the adversary $\Sim$. $\FMulPre$ receives $\sgr{\cdot}$-shares of $\md, \me$ from the servers where $P_s$, for $s \in \{0,1,2\}$, holds $\sgr{\md}_s = (\md_s, \md_{(s+1)\%3})$ and $\sgr{\me}_s = (\me_s, \me_{(s+1)\%3})$ such that $\md = \md_0 + \md_1 + \md_2$ and $\me = \me_0 + \me_1 + \me_2$. 
	Let $P_{i}$ denotes the server corrupted by $\Sim$. $\FMulPre$ receives $\sgr{\mf}_i = (\mf_i, \mf_{(i+1)\%3})$ from $\Sim$ where $\mf = \md \me$.
	$\FMulPre$ proceeds as follows:
	\begin{myitemize}
		\item[--] Reconstructs $\md, \me$ using the shares received from honest servers and compute $\mf = \md \me$.
		\item[--] Compute $\mf_{(i+2)\%3} = \mf - \mf_i -\mf_{(i+1)\%3}$ and set the output shares as $\sgr{\mf}_0 = (\mf_0, \mf_1), \sgr{\mf}_1 = (\mf_1, \mf_2),  \sgr{\mf}_2 = (\mf_2, \mf_0)$.
		\item[--] Send $(\OUTPUT, \sgr{\mf}_s)$ to server $P_s \in \Partyset$.
	\end{myitemize}
\end{systembox}
\smallskip

\begin{lemma}[Communication]
	\label{app:piMult}
	Protocol $\piMult$ (\boxref{fig:piMult}) requires an amortized cost of $3\ell$ bits in the preprocessing phase, and $1$ round and amortized cost of $3\ell$ bits in the online phase.
\end{lemma}
\begin{proof}
	In the preprocessing phase, generation of $\av{\wz}$ and $\gv{\wz}$ are non-interactive. This is followed by one execution of $\piMulPre$, which requires an amortized communication cost of $3\ell$ bits. 
	During the online phase, $P_0,P_1$ $\jsend$ $\sqr{\starbeta{\wz}}_1$ to $P_2$, while $P_0,P_2$ $\jsend$ $\sqr{\starbeta{\wz}}_2$ to $P_1$. This requires one round and a communication of $2\ell$ bits. Following this, $P_1,P_2$ $\jsend$ $\bv{\wz}+\gv{\wz}$ to $P_0$, which requires one round and a communication of $\ell$ bits. However, $\jsend$ of $\bv{\wz}+\gv{\wz}$ can be delayed till the end of the protocol, and will require only one round for the entire circuit and can be amortized.
\end{proof}

\paragraph{Reconstruction Protocol}
Protocol $\piRec$ (\boxref{fig:piRec}) allows servers to robustly reconstruct value $\val \in \Z{\ell}$ from its $\shr{\cdot}$-shares. Note that each server misses one share of $\val$ which is held by the other two servers. Consider the case of $P_0$ who requires $\gv{\val}$ to compute $\val$. During the preprocessing, $\ESet$ compute a commitment of $\gv{\val}$, denoted by $\Commit{\gv{\val}}$ and $\jsend$ the same to $P_0$. Similar steps are performed for the values $\sqr{\av{\val}}_2$ and $\sqr{\av{\val}}_1$ that are required by servers $P_1$ and $P_2$ respectively. During the online phase, servers open their commitments to the intended server who accepts the opening that is consistent with the agreed upon commitment. 

\smallskip
\begin{protocolbox}{$\piRec(\Partyset, \shr{\val})$}{3PC: Reconstruction of $\val$ among the servers}{fig:piRec}
	\justify
	\algoHead{Preprocessing:}
	\begin{myitemize}
		\item[--] $P_0,P_j$, for $j \in \EInSet$, compute  $\Commit{\sqr{\av{\val}}_j}$, while $\ESet$ compute  $\Commit{\gv{\val}}$.   
		\item[--] $P_1,P_2$ $\jsend$ $\Commit{\gv{\val}}$ to $P_0$, while $P_0,P_1$ $\jsend$ $\Commit{\sqr{\av{\val}}_1}$ to $P_2$, and $P_0,P_2$ $\jsend$ $\Commit{\sqr{\av{\val}}_2}$ to $P_1$,
	\end{myitemize}
	\justify\vspace{-3mm}
	\algoHead{Online:}
	\begin{myitemize}
		\item[--] $P_0,P_1$ open $\Commit{\sqr{\av{\val}}_1}$ to $P_2$.
		$P_0,P_2$ open  $\Commit{\sqr{\av{\val}}_2}$ to $P_1$.
		$P_1,P_2$ open $\Commit{\gv{\val}}$ to $P_0$.
		\item[--] Each server accepts the opening that is consistent with the agreed upon commitment. $\ESet$ compute $\val = \bv{\val} - \sqr{\av{\val}}_1 - \sqr{\av{\val}}_2$, while $P_0$ computes $\val = (\bv{\val} + \gv{\val}) - \sqr{\av{\val}}_1 - \sqr{\av{\val}}_2 - \gv{\val}$.
	\end{myitemize}
\end{protocolbox}

\begin{lemma}[Communication]
	\label{app:piRec}
	Protocol $\piRec$ (\boxref{fig:piRec}) requires $1$ round and a communication of $6\ell$ bits in the online phase. 
\end{lemma}
\begin{proof}
	The preprocessing phase consists of communication of commitment values using the $\piJmp$ protocol. The hash-based commitment scheme allows generation of a single commitment for several values and hence the cost gets amortised away for multiple instances.
	During the online phase, each server receives an opening for the commitment from other two servers, which requires one round and an overall communication of $6\ell$ bits.
\end{proof}

\paragraph{The Complete 3PC} 
\label{para:3PCmain} 
For the sake of completeness and to demonstrate how GOD is achieved, we show how to compile the above primitives for a general 3PC. A similar approach will be taken for 4PC and each PPML task, and we will avoid repetition. 
In order to compute an arithmetic circuit over $\Z{\ell}$, we  first invoke the key-setup functionality $\FSETUP$ (\boxref{fig:FSETUP}) for key distribution and preprocessing of $\piSh$, $\piMult$ and $\piRec$, as per the given circuit.  During the online phase,  $P_i \in \Partyset$ shares its input $\wx_i$ by executing online steps of $\piSh$ (\boxref{fig:piIn}). This is followed by the circuit evaluation phase, where severs evaluate the gates in the circuit in the topological order, with addition gates (and multiplication-by-a-constant gates)  being computed locally, and multiplication gates being  computed via online of  $\piMult$ (\boxref{fig:piMult}). Finally, servers run the online steps of  $\piRec$ (\boxref{fig:piRec}) on the output wires to reconstruct the function output. 
To leverage amortization, only {\em send} phases of all the $\Jmp$ are run on the flow. At the end of preprocessing, 
the  {\em verify} phase for all possible ordered pair of senders are run. We carry on computation in the online phase only when the {\em verify} phases in the preprocessing are successful. Otherwise, the servers simply send their inputs to the elected $\TTP$, who computes the function and returns the result to all the servers.  Similarly, depending on the output of the {\em verify} at the end of the online phase, either the reconstruction is carried out or a $\TTP$ is identified. In the latter case, computation completes as mentioned before. 

\smallskip
\begin{protocolbox}{3PC}{Complete 3PC}{fig:3pc}
	\justify
	Let $\ckt$ denote the circuit representation of the function to be evaluated. Servers invoke $\FSETUP$ to establish common keys among themselves.
	\begin{myitemize}  
	\item{\bf Preprocessing:}
	Preprocessing phase of $\piSh, \piMult, \piRec$ is executed depending on $\ckt$ topology. For every pair of servers, $P_i, P_j \in \Partyset$, the verification for $\Jmp$ is executed once, for all the $\Jmp$ instances with $P_i, P_j$ as the senders. 
	\item{\bf Online phase:} If $\TTP$ is identified in the preprocessing phase, servers send their inputs to the $\TTP$, which evaluates the function on the clear inputs, and sends the result back to the servers. Else, 
	\begin{inneritemize}
		\item[--] {\em Input sharing: } The online phase of $\piSh$ is executed.
		\item[--]  {\em Evaluation phase: }
		\begin{inneritemize}
			\item[-] Linear gates such as addition and multiplication-by-a-constant are evaluated locally.
			\item[-] For multiplication gates, the online phase of $\piMult$ is executed with only the {\em send} phase of $\Jmp$ being executed for $\Jmp$ towards $P_1, P_2$. All communication towards $P_0$ is carried out in the verification phase. 
		\end{inneritemize}		
		\item[--] {\em Verification phase:} With respect to the multiplication gates, the communication towards $P_0$ is carried out by invoking $\Jmp$ (both {\em send} and {\em verify} phase). This is followed by the {\em verify} phase of $\Jmp$ for all $\Jmp$ instances with $P_1, P_2$ as the receivers. 		
		\item[--] If $\TTP$ is identified at any point during verification, servers send their inputs to the $\TTP$, which evaluates the function on the clear inputs, and sends the result back to the servers.  
		\item[--] {\em Output reconstruction:} If verification succeeds without $\TTP$ identification, servers execute online phase of $\piRec$.
		\end{inneritemize}
	\end{myitemize}
\end{protocolbox}

%
{\em On the security of our framework:} 
We emphasize that we follow the standard traditional (real-world / ideal-world based) security definition of MPC, according to which, in the 4-party setting with 1 corruption, exactly 1 party is assumed to be corrupt, and rest are {\em honest}. As per this definition, disclosing the honest parties's inputs to a selected {\em honest} party is {\em not} a breach of security. Indeed, in our framework, the data sharing and the computation on the shared data is done in a way that any malicious behaviour leads to establishment of a $\TTP$ who is enabled to receive all the inputs and compute the output on the clear. There has been a recent study on the additional requirement of hiding the inputs from a quorum of honest parties (treating them as semi-honest), termed as  Friends-and-Foes (FaF) security notion~\cite{alon20}. This is a stronger security goal than the standard one and it has been shown that one cannot obtain FaF-secure robust 3PC. We leave FaF-secure 4PC for future exploration. 
%

\subsection{Building Blocks for PPML using 3PC}
\label{sec:PPML}
\input{Main_PrivML_3PC}


%% file: Main_PrivML_3PC.tex
This section provides details on robust  realizations of the following building blocks for PPML in 3-server setting-- i) Dot Product, ii) Truncation, iii) Dot Product with Truncation, iv) Secure Comparison, and v) Non-linear Activation functions-- Sigmoid and ReLU. We provide the security proofs in \S\ref{app:sec3pc}. 
We begin by providing details of input sharing and reconstruction in the SOC setting. 

\smallskip
\begin{protocolbox}{$\piShS(\User,\val)$ and $\piRecS(\User, \shr{\val})$}{3PC: Input Sharing and Output Reconstruction}{fig:piSOC}
	\justify
	\algoHead{Input Sharing:}
	\begin{myitemize}
		\item[--] $P_0,P_s$, for $s \in \EInSet$, together sample random $\sqr{\av{\val}}_s \in \Z{\ell}$, while $\ESet$ together sample random $\gv{\val} \in \Z{\ell}$.
		\item[--] $P_0,P_1$ $\jsend$ $\Commit{\sqr{\av{\val}}_1}$ to $P_2$, while $P_0,P_2$ $\jsend$ $\Commit{\sqr{\av{\val}}_2}$ to $P_1$, and $P_1,P_2$ $\jsend$ $\Commit{\gv{\val}}$ to $P_0$.
		\item[--] Each server sends $(\Commit{\sqr{\av{\val}}_1}, \Commit{\sqr{\av{\val}}_2}, \Commit{\gv{\val}})$ to $\User$ who accepts the values that form majority. Also, $P_0, P_s$, for $s \in \{1,2\}$, open $\sqr{\av{\val}}_s$ towards $\User$ while $\ESet$ open $\gv{\val}$ towards $\User$.
		\item[--] $\User$ accepts the consistent opening, recovers $\sqrA{\av{\val}}, \sqrB{\av{\val}}, \gv{\val}$, computes $\bv{\val} = \val + \sqr{\av{\val}}_1 + \sqr{\av{\val}}_2$, and sends $\bv{\val} + \gv{\val}$ to all three servers.
		\item[--] Servers broadcast the received value and accept the majority value if it exists, and a default value, otherwise. $\ESet$ locally compute $\bv{\val}$ from $\bv{\val} +\gv{\val}$ using $\gv{\val}$ to complete the sharing of $\val$. 
	\end{myitemize}
	\justify\vspace{-3mm}
	\algoHead{Output Reconstruction:}
	\begin{myitemize}
		\item[--] Servers execute the preprocessing of $\piRec(\Partyset, \shr{\val})$ to agree upon commitments of $\sqr{\av{\val}}_1, \sqr{\av{\val}}_2$ and $\gv{\val}$.
		\item[--] Each server sends $\bv{\val}+\gv{\val}$ and commitments on $\sqr{\av{\val}}_1, \sqr{\av{\val}}_2$ and $\gv{\val}$ to $\User$, who accepts the values forming majority. 
		\item[--] $P_0, P_i$ for $i \in \EInSet$ open $\sqr{\av{\val}}_i$ to $\User$, while $P_1, P_2$ open $\gv{\val}$ to $\User$.  
		\item[--] $\User$ accepts the consistent opening and computes $\val = (\bv{\val} + \gv{\val}) - \sqr{\av{\val}}_1 - \sqr{\av{\val}}_2 - \gv{\val}$. 
	\end{myitemize}
\end{protocolbox}

\paragraph{Input Sharing and Output Reconstruction in the SOC Setting} 
Protocol $\piShS$ (\boxref{fig:piSOC}) extends input sharing to the SOC setting and allows a user $\User$ to generate the $\shr{\cdot}$-shares of its input $\val$ among the three servers. Note that the necessary commitments to facilitate the sharing are generated in the preprocessing phase by the servers which are then communicated to $\User$, along with the opening, in the online phase. $\User$ selects the commitment forming the majority (for each share) owing to the presence of an honest majority among the servers, and accepts the corresponding shares.
Analogously, protocol $\piRecS$ (\boxref{fig:piSOC}) allows the servers to reconstruct a value $\val$ towards user $\User$. 
In either of the protocols, if at any point, a $\TTP$ is identified, then servers signal the $\TTP$'s identity to $\User$. $\User$ selects the $\TTP$ as the one forming a majority and sends its input in the clear to the $\TTP$, who computes the function output and sends it back to $\User$.

\paragraph{MSB Extraction}
Protocol $\piBitExt$ allows servers to compute the {\em boolean} sharing of the most significant bit ($\MSB$) of a value $\val$ given its arithmetic sharing $\shr{\val}$. To compute the $\MSB$, we use the optimized 2-input Parallel Prefix Adder (PPA) boolean circuit proposed by ABY3~\cite{MR18}. The PPA circuit consists of $2\ell -2$ AND gates and has a multiplicative depth of $\log \ell$. 
%

\begin{table}[htb!]
	\centering
	\resizebox{.95\textwidth}{!}{
		\begin{tabular}{r | c | c | c }
			\toprule
			& $P_0$ & $P_1$ & $P_2$ \\ 
			\midrule
			$\shrB{\val_0[i]}$ & $(0,0,0)$ & $(0,\val_0[i],\val_0[i])$ & $(0, \val_0[i],\val_0[i])$\\ 
			$\shrB{\val_1[i]}$ & $(\val_1[i],0,0)$ & $(\val_1[i],0, 0)$ & $(0,0,0)$\\ 
			$\shrB{\val_2[i]}$ & $(0,\val_2[i],0)$ & $(0,0,0)$ & $(0, \val_2[i],0)$\\ 
			\bottomrule
		\end{tabular}
	}
	\vspace{-2mm}
	\caption{\small The $\shrB{\cdot}$-sharing corresponding to $i^{th}$ bit of $\val_0 = \bv{\val}, \val_1 = -\sqr{\av{\val}}_1$ and $\val_2 = -\sqr{\av{\val}}_2$. Here $i \in \{0,\ldots,\ell-1\}$.\label{tab:sharesBitExt}}
\end{table}

Let $\val_0 = \bv{\val}, \val_1 = -\sqr{\av{\val}}_1$ and $\val_2 = -\sqr{\av{\val}}_2$. Then $\val = \val_0 + \val_1 + \val_2$. Servers first locally compute the boolean shares corresponding to each bit of the values $\val_0, \val_1$ and $\val_2$ according to \tabref{sharesBitExt}.
It has been shown in ABY3 that $\val = \val_0 + \val_1 + \val_2$ can also be expressed as $\val = 2c + s$ where {\sf FA}$(\val_0[i], \val_1[i], \val_2[i]) \rightarrow (c[i],s[i])$ for $i \in \{0,\ldots,\ell-1\}$. Here {\sf FA} denotes a Full Adder circuit while $s$ and $c$ denote the sum and carry bits respectively.
To summarize, servers execute $\ell$ instances of {\sf FA} in parallel to compute $\shrB{c}$ and $\shrB{s}$. The {\sf FA}'s are executed independently and require one round of communication. The final result is then computed as $\MSB(2\shrB{c}+\shrB{s})$ using the optimized PPA circuit.

\begin{lemma}[Communication]
	\label{app:piBitExt2}
	Protocol $\piBitExt$ requires a communication cost of $9\ell - 6$ bits in the preprocessing phase and require $\log \ell + 1$ rounds and an amortized communication of $9\ell - 6$ bits in the online phase.
\end{lemma}
\begin{proof}
	In $\piBitExt$, first round comprises of $\ell$ Full Adder~({\sf FA}) circuits executing in parallel, each comprising of single AND gate. This is followed by the execution of the optimized PPA circuit of ABY3~\cite{MR18}, which comprises of $2\ell-2$ AND gates and has a multiplicative depth of $\log \ell$. Hence the communication cost follows from the multiplication for $3\ell-2$ AND gates.
\end{proof}

\paragraph{Bit to Arithmetic Conversion}
Given the boolean sharing of a bit $\bitb$, denoted as $\shrB{\bitb}$, protocol $\PiBitA$ (\boxref{fig:piBit2A}) allows servers to compute the arithmetic sharing $\shr{\arval{\bitb}}$. Here $\arval{\bitb}$ denotes the equivalent value of $\bitb$ over ring $\Z{\ell}$ (see Notation~\ref{boolequivring}).  
As pointed out in BLAZE, $\arval{\bitb} = \arval{(\bv{\bitb} \xor \av{\bitb})} = \arval{\bv{\bitb}} + \arval{\av{\bitb}} - 2 \arval{\bv{\bitb}} \arval{\av{\bitb}}$. Also $\arval{\av{\bitb}} = \arval{(\sqr{\av{\bitb}}_{1} \xor \sqr{\av{\bitb}}_{2})} = \arval{\sqr{\av{\bitb}}_{1}}+ \arval{\sqr{\av{\bitb}}_{2}} - 2 \arval{\sqr{\av{\bitb}}_{1}} \arval{\sqr{\av{\bitb}}_{2}}$. During the preprocessing phase, $P_0, P_j$ for $j \in \EInSet$ execute $\piJSh$ on $\arval{\sqr{\av{b}}_{j}}$ to generate $\shr{\arval{\sqr{\av{b}}_{j}}}$. Servers then execute $\piMult$ on $\shr{\arval{\sqr{\av{b}}_{1}}}$ and $\shr{\arval{\sqr{\av{b}}_{2}}}$ to generate $\shr{\arval{\sqr{\av{b}}_{1}} \arval{\sqr{\av{b}}_{2}}}$ followed by locally computing $\shr{\arval{\av{\bitb}}}$. During the online phase, $\ESet$ execute $\piJSh$ on $\arval{\bv{\bitb}}$ to jointly generate $\shr{\arval{\bv{\bitb}}}$. Servers then execute $\piMult$ protocol on $\shr{\arval{\bv{\bitb}}}$ and $\shr{\arval{\av{\bitb}}}$ to compute $\shr{\arval{\bv{\bitb}}\arval{\av{\bitb}}}$ followed by locally computing $\shr{\arval{\bitb}}$.

\begin{lemma}[Communication]
	\label{app:piBitA}
	Protocol $\PiBitA$ (\boxref{fig:piBit2A}) requires an amortized communication cost of $9 \ell$ bits in the preprocessing phase and requires $1$ round and an amortized communication of $4 \ell$ bits in the online phase.
\end{lemma}
\begin{proof}
	In the preprocessing phase, servers run two instances of $\piJSh$, which can be done non-interactively (ref. \tabref{SharesAssign}). This is followed by an execution of entire multiplication protocol, which requires $6\ell$ bits to be communicated (Lemma~\ref{app:piMult}). Parallelly, the servers execute the preprocessing phase of $\piMult$, resulting in an additional $3\ell$ bits of communication (Lemma~\ref{app:piMult}).
	During the online phase, $P_1, P_2$ execute $\piJSh$ once, which requires one round and $\ell$ bits to be communicated. In $\piJSh$, the communication towards $P_0$ can be deferred till the end, thereby requiring a single round for multiple instances. This is followed by an execution of the online phase of $\piMult$, which requires one round and a communication of $3\ell$ bits.
\end{proof}

\begin{protocolbox}{$\PiBitA(\Partyset, \shrB{\bitb})$}{3PC: Bit2A Protocol}{fig:piBit2A}
	\justify
	\algoHead{Preprocessing:}
	\begin{myitemize}
		\item[--] $P_0, P_j$ for $j \in \EInSet$ execute $\piJSh$ on $\arval{\sqr{\av{\bitb}}_{j}}$ to generate $\shr{\arval{\sqr{\av{\bitb}}_{j}}}$.
		\item[--] Servers execute $\piMult(\Partyset, \arval{\sqr{\av{\bitb}}_{1}}, \arval{\sqr{\av{\bitb}}_{2}})$ to generate $\shr{\vu}$ where $\vu =\arval{ \sqr{\av{\bitb}}_{1}}\arval{ \sqr{\av{\bitb}}_{2}}$, followed by locally computing $\shr{\arval{\av{\bitb}}} = \shr{\arval{\sqr{\av{\bitb}}_{1}}} + \shr{\arval{\sqr{\av{\bitb}}_{2}}} - 2 \shr{\vu}$.
		\item[--] Servers execute the preprocessing phase of $\piMult(\Partyset, \arval{\bv{\bitb}}, \arval{\av{\bitb}})$ for $\vv = \arval{\bv{\bitb}} \arval{\av{\bitb}}$.
	\end{myitemize}  
	\justify\vspace{-3mm}
	\algoHead{Online:}
	\begin{myitemize}
		\item[--] $P_1, P_2$ execute  $\piJSh(P_1,P_2,\arval{\bv{\bitb}})$ to generate $\shr{\arval{\bv{\bitb}}}$.
		\item[--] Servers execute online phase of $\piMult(\Partyset, \arval{\bv{\bitb}}, \arval{\av{\bitb}})$ to generate $\shr{\vv}$ where $\vv = \arval{ \bv{\bitb}} \arval{ \av{\bitb}}$, followed by locally computing $\shr{\arval{\bitb}} = \shr{\arval{\bv{\bitb}}} + \shr{\arval{\av{\bitb}}} - 2 \shr{\vv}$.
	\end{myitemize}        
\end{protocolbox}

\paragraph{Bit Injection}\label{par:bitInj}
Given the binary sharing of a bit $\bitb$, denoted as $\shrB{\bitb}$, and the arithmetic sharing of $\val \in \Z{\ell}$, protocol $\PiBitInj$ computes $\shr{\cdot}$-sharing of $\bitb\val$. Towards this, servers first execute $\PiBitA$ on $\shrB{\bitb}$ to generate $\shr{\bitb}$. This is followed by servers computing $\shr{\bitb\val}$ by executing $\piMult$ protocol on $\shr{\bitb}$ and $\shr{\val}$.

\begin{lemma}[Communication]
	\label{app:PiBitInj}
	Protocol $\PiBitInj$ requires an amortized communication cost of $12\ell$ bits in the preprocessing phase and requires $2$ rounds and an amortized communication of $7\ell$ bits in the online phase.
\end{lemma}
\begin{proof}
	Protocol $\PiBitInj$ is essentially an execution of $\PiBitA$ (Lemma~\ref{app:piBitA}) followed by one invocation of $\piMult$ (Lemma~\ref{app:piMult}) and the costs follow. 
\end{proof}

\paragraph{Dot Product}
Given the $\shrd$-sharing of vectors $\vecX$ and $\vecY$, protocol $\piDotP$ (\boxref{fig:piDotP}) allows servers to generate $\shrd$-sharing of $\wz = \vecX \band \vecY$ robustly. $\shrd$-sharing of a vector $\vecX$ of size $\nf$, means that each element $\vx_i \in \Z{\ell}$ of $\vecX$, for $i \in [\nf]$, is $\shrd$-shared. 
We borrow ideas from BLAZE for obtaining an online communication cost {\em independent} of $\nf$ and use $\Jmp$ primitive to ensure either success or $\TTP$ selection. Analogous to our multiplication protocol, our dot product offloads one call to a robust dot product protocol to the preprocessing. By extending techniques of \cite{BonehBCGI19, BGIN19}, we give an instantiation for the dot product protocol used in our preprocessing whose (amortized) communication cost is constant, thereby making our preprocessing cost also {\em independent} of $\nf$.

To begin with, $\wz = \vecX \band \vecY$ can be viewed as $\nf$ parallel multiplication instances of the form $\wz_i = \wx_i \wy_i$ for $i \in [n]$, followed by adding up the results. Let $\starbeta{\wz} = \sum_{i=1}^{\nf} \starbeta{\wz_i}$. Then, 

\medskip 
\resizebox{.94\linewidth}{!}{
	\begin{minipage}{\linewidth}
		\begin{equation}
		\starbeta{\wz} = -\sum_{i=1}^{\nf} (\bv{\wx_i}+\gv{\wx_i})\av{\wy_i} - \sum_{i=1}^{\nf} (\bv{\wy_i}+\gv{\wy_i})\av{\wx_i}+ \av{\wz} + {\Chi} \label{eq:4}
		\end{equation}
		\smallskip 
	\end{minipage}}

where $\Chi = \sum_{i=1}^{\nf} (\gv{\wx_i}\av{\wy_i} + \gv{\wy_i}\av{\wx_i} + \GammaV{\wx_i\wy_i} - \psi_i)$.
\medskip 

Apart from the aforementioned modification, the online phase for dot product proceeds similar to that of multiplication protocol. $P_0,P_1$ locally compute $\sqr{\starbeta{\wz}}_1$ as per Eq. \ref{eq:4} and $\jsend$ $\sqr{\starbeta{\wz}}_1$ to $P_2$. $P_1$ obtains $\sqrV{\starbeta{\wz}}{2}$ in a similar fashion. $P_1,P_2$ reconstruct $\starbeta{\wz} = \sqrV{\starbeta{\wz}}{1} + \sqrV{\starbeta{\wz}}{2}$ and compute $\bv{\wz} = \starbeta{\wz} + \sum_{i=1}^{\nf} \bv{\wx_i}\bv{\wy_i} + \psi$. Here, the value $\psi$ has to be correctly generated in the preprocessing phase satisfying Eq. \ref{eq:4}. Finally, $P_1,P_2$ $\jsend$ $\bv{\wz}+\gv{\wz}$ to $P_0$.

We now provide the details for preprocessing phase that enable servers to obtain the required values ($\Chi, \psi$) with the invocation of a dot product protocol in a black-box way. Towards this, let $\vecd = [\md_1, \ldots, \md_{\nf}]$ and $\vece = [\me_1, \ldots, \me_{\nf}]$, where $\md_i = \gv{\wx_i} + \av{\wx_i}$ and $\me_i = \gv{\wy_i} + \av{\wy_i}$ for $i \in [\nf]$, as in the case of multiplication. Then for $\mf = \vecd \band \vece$, 
\begin{align*}
\mf &= \vecd \band \vece = \sum_{i=1}^{\nf} \md_i \me_i = \sum_{i=1}^{\nf} (\gv{\wx_i}+\av{\wx_i})(\gv{\wy_i}+\av{\wy_i})\\
&= \sum_{i=1}^{\nf} (\gv{\wx_i}\gv{\wy_i} + \psi_i) + \sum_{i=1}^{\nf} \Chi_i 
=  \sum_{i=1}^{\nf} (\gv{\wx_i}\gv{\wy_i} + \psi_i) + \Chi \\
&= \sum_{i=1}^{\nf} (\gv{\wx_i}\gv{\wy_i} + \psi_i) + \sqr{\Chi}_1 + \sqr{\Chi}_2
= \mf_2 + \mf_1 + \mf_0.
\end{align*}
where $\mf_2 = \sum_{i=1}^{\nf} (\gv{\wx_i}\gv{\wy_i} + \psi_i), \mf_1 = \sqr{\Chi}_1$ and $\mf_0 = \sqr{\Chi}_2$.
\smallskip

Using the above relation, the preprocessing phase proceeds as follows: $P_0,P_j$ for $j \in \EInSet$ sample a random $\sqr{\av{\wz}}_j \in \Z{\ell}$, while $P_1,P_2$ sample random $\gv{\wz}$. Servers locally prepare $\sgr{\vecd}, \sgr{\vece}$ similar to that of multiplication protocol. Servers then execute a robust 3PC dot product protocol, denoted by $\piDotPPre$ (the task is abstracted away in the functionality Fig.~\ref{fig:FDotPPre}), that takes $\sgr{\vecd}, \sgr{\vece}$ as input and compute $\sgr{\mf}$ with $\mf = \vecd \band \vece$. Given $\sgr{\mf}$, the $\psi$ and $\sqr{\Chi}$ values are extracted as follows (ref. Eq. \ref{eq:Dotpremap}):
\vspace{-2mm}
\begin{equation}\label{eq:Dotpremap}
	\psi = \mf_2 - \sum_{i=1}^{\nf} \gv{\wx_i} \gv{\wy_i},
	~~\sqr{\Chi}_1 = \mf_1, ~~\sqr{\Chi}_2 = \mf_0, 
\end{equation}
%
It is easy to see from the semantics of $\sgr{\cdot}$-sharing that both $P_1,P_2$ obtain $\mf_2$ and hence $\psi$. Similarly, both $P_0,P_1$ obtain $\mf_1$ and hence $\sqr{\Chi}_1$, while $P_0,P_2$ obtain $\sqr{\Chi}_2$.

\smallskip
\begin{protocolbox}{$\piDotP(\Partyset, \{\shr{\wx_i},\shr{\wy_i}\}_{i \in [\nf]})$}{3PC: Dot Product Protocol ($\wz = \vecX \band \vecY$)}{fig:piDotP}
	\justify
	\algoHead{Preprocessing:}
	\begin{myitemize}
		\item[--] $P_0, P_j$, for $j \in \EInSet$, together sample random $\sqr{\av{\wz}}_j \in \Z{\ell}$, while $P_1, P_2$ sample random $\gv{\wz} \in \Z{\ell}$.
		\item[--] Servers locally compute $\sgr{\cdot}$-sharing of $\vecd, \vece$ with $\md_i = \gv{\wx_i} + \av{\wx_i}$ and $\me_i = \gv{\wy_i} + \av{\wy_i}$ for $i \in [\nf]$ as follows:
		\begin{align*} 
			&({\sgr{\md_i}}_0 \text{=} (\sqr{\av{\wx_i}}_{2}, \sqr{\av{\wx_i}}_{1}),
			{\sgr{\md_i}}_1 \text{=} (\sqr{\av{\wx_i}}_{1},\gv{\wx_i}), 
			{\sgr{\md_i}}_2 \text{=} (\gv{\wx_i},\sqr{\av{\wx_i}}_{2}))\\
			&({\sgr{\me_i}}_0 \text{=} (\sqr{\av{\wy_i}}_{2}, \sqr{\av{\wy_i}}_{1}),
			{\sgr{\me_i}}_1 \text{=} (\sqr{\av{\wy_i}}_{1},\gv{\wy_i}), 
			{\sgr{\me_i}}_2 \text{=} (\gv{\wy_i},\sqr{\av{\wy_i}}_{2}))
		\end{align*}
		\item[--] Servers execute $\piDotPPre(\Partyset, \sgr{\vecd}, \sgr{\vece})$ to generate $\sgr{\mf} = \sgr{\vecd \band \vece}$. 
		\item[--] $P_0, P_1$ locally set $\sqr{\Chi}_1 = \mf_1$, while $P_0, P_2$ locally set $\sqr{\Chi}_2 = \mf_0$. $P_1, P_2$ locally compute $\psi = \mf_2 - \sum_{i=1}^{\nf} \gv{\wx_i} \gv{\wy_i}$.
	\end{myitemize}  
	\justify\vspace{-3mm}
	\algoHead{Online:}
	\begin{myitemize}
		\item[--] $P_0,P_j$, for $j \in \EInSet$, compute $\sqrV{\starbeta{\wz}}{j} = -\sum_{i=1}^{\nf} ( (\bv{\wx_i}+\gv{\wx_i})\sqrV{\av{\wy_i}}{j} + (\bv{\wy_i}+\gv{\wy_i})\sqrV{\av{\wx_i}}{j}) + \sqrV{\av{\wz} }{j}+ \sqrV{\Chi}{j}$. 
		\item[--] $P_0,P_1$ $\jsend$ $\sqr{\starbeta{\wz}}_1$ to $P_2$ and  $P_0,P_2$ $\jsend$ $\sqr{\starbeta{\wz}}_2$ to $P_1$.
		\item [--] $\ESet$ locally compute $\starbeta{\wz} = \sqrV{\starbeta{\wz}}{1}+\sqrV{\starbeta{\wz}}{2}$ and set $ \bv{\wz} = \starbeta{\wz} + \sum_{i=1}^{\nf}(\bv{\wx_i}\bv{\wy_i}) + \psi$. 
		\item[--] $P_1,P_2$ $\jsend$ $\bv{\wz}+\gv{\wz}$ to $P_0$.
	\end{myitemize}        
\end{protocolbox}

\smallskip
\begin{systembox}{$\FDotPPre$}{3PC: Ideal functionality for $\piDotPPre$ protocol}{fig:FDotPPre}
	\justify
	$\FDotPPre$ interacts with the servers in $\Partyset$ and the adversary $\Sim$. $\FDotPPre$ receives $\sgr{\cdot}$-shares of vectors $\vecd = (\md_1, \ldots, \md_{\nf}), \vece = (\me_1, \ldots, \me_{\nf})$ from the servers. 
	Let $\val_{j, s}$ for $j \in [\nf], s \in \{0, 1,  2\}$ denote the share of $\val_j$ such that $\val_j = \val_{j, 0} + \val_{j, 1} + \val_{j, 2}$.
	Server $P_s$, for $s \in \{0,1,2\}$, holds $\sgr{\md_j}_s = (\md_{j,s}, \md_{j, (s+1)\%3})$ and $\sgr{\me_j}_s = (\me_{j,s}, \me_{j, (s+1)\%3})$ where $j \in [\nf]$. Let $P_i$ denotes the server corrupted by $\Sim$. $\FMulPre$ receives $\sgr{\mf}_i = (\mf_i, \mf_{(i+1)\%3})$ from $\Sim$ where $\mf = \vecd \band \vece$.
	$\FDotPPre$ proceeds as follows:
	\begin{myitemize}
		\item[--] Reconstructs $\md_j, \me_j$, for $j \in [\nf]$, using the shares received from honest servers and compute $\mf = \sum_{j=1}^{n} \md_j \me_j$.
		\item[--] Compute $\mf_{(i+2)\%3} = \mf - \mf_i -\mf_{(i+1)\%3}$ and set the output shares as $\sgr{\mf}_0 = (\mf_0, \mf_1), \sgr{\mf}_1 = (\mf_1, \mf_2),  \sgr{\mf}_2 = (\mf_2, \mf_0)$.
		\item[--] Send $(\OUTPUT, \sgr{\mf}_s)$ to server $P_s \in \Partyset$.
	\end{myitemize}
\end{systembox}

The ideal world functionality for realizing $\piDotPPre$ is presented in \boxref{fig:FDotPPre}. 
A trivial way to instantiate $\piDotPPre$ is to treat a dot product operation as $\nf$ multiplications. However, this results in a communication cost that is linearly dependent on the feature size. 
Instead, we instantiate $\piDotPPre$ by a semi-honest dot product protocol followed by a verification phase to check the correctness. For the verification phase, we extend the techniques of \cite{BonehBCGI19, BGIN19} to provide support for verification of dot product tuples. Setting the verification phase parameters appropriately gives a $\piDotPPre$ whose (amortized) communication cost is independent of the feature size. Details appear in \S\ref{appsec:piDotP}.

\begin{lemma}[Communication]
	\label{app:piDotP}
	Protocol $\piDotP$ (\boxref{fig:piDotP}) requires an amortized communication of $3 \ell $ bits in the preprocessing phase and requires $1$ round and an amortized communication of $3 \ell$ bits in the online phase.
\end{lemma}
\begin{proof}
	During the preprocessing phase, servers execute $\piDotPPre$. This requires communicating $3 \ell$ bits for a single semi-honest dot product protocol and $\Order(\frac{\sqrt{\nf}}{\sqrt{\nm}})$ extended ring elements for its verification. By appropriately setting the values of $\nf, \nm$, the cost of communicating $\Order(\frac{\sqrt{\nf}}{\sqrt{\nm}})$ elements can be amortized away, thereby resulting in an amortized communication cost of $3 \ell$ bits in the preprocessing phase. 
	The online phase follows similarly to that of $\piMult$, the only difference being that servers combine their shares corresponding to all the $\nf$ multiplications into one and then exchange.  This requires one round and an amortized communication of $3 \ell$ bits.
\end{proof}

\paragraph{Truncation} 
Working over fixed-point values, repeated multiplications using FPA arithmetic can lead to an overflow resulting in loss of significant bits of information. This put forth the need for truncation~\cite{MohasselZ17,MR18,ASTRA,FLASH,BLAZE} that re-adjusts the shares after multiplication so that FPA semantics are maintained. As shown in SecureML~\cite{MohasselZ17}, the method of truncation would result in loss of information on the least significant bits and affect the accuracy by a very minimal amount. 

For truncation, servers execute $\piTrunc$ (\boxref{fig:piTr}) to generate $(\sqr{\vr}, \shr{\vrt})$-pair, where $\vr$ is a random ring element, and $\vrt$ is the truncated value of $\vr$, i.e the value $\vr$ right-shifted by $d$ bit positions. Recall that $d$ denotes the number of bits allocated for the fractional part in the FPA representation. Given $(\vr, \vrt)$, the truncated value of $\val$, denoted as $\trunc{\val}$, is computed as $\trunc{\val} = \trunc{(\val - \vr)} + \vrt$. The correctness and accuracy of this method was shown in ABY3~\cite{MR18}.

Protocol $\piTrunc$ is inspired from \cite{MR18, Trident} and proceeds as follows to generate $(\sqr{\vr}, \shr{\vrt})$. Analogous to the approach of ABY3 \cite{MR18}, servers generate a boolean sharing of an $\ell$-bit value $\vr = \vr_1 \xor \vr_2$, non-interactively. Each server truncates its share of $\vr$ locally to obtain a boolean sharing of $\vrt$ by removing the lower $d$ bits. To obtain the arithmetic shares of $(\vr, \vrt)$ from their boolean sharing, we do not, however, rely on the approach of ABY3 as it requires more rounds. Instead, we implicitly perform a {\em boolean to arithmetic conversion}, as was proposed in Trident \cite{Trident}, to obtain the arithmetic shares of $(\vr, \vrt)$. This entails performing two dot product operations and constitutes the cost for $\piTrunc$. 

\begin{mypbox}{$\piTrunc(\Partyset)$}{3PC: Generating Random Truncated Pair $(\vr, \vrt)$}{fig:piTr}
	\justify	
	\begin{myitemize}
		\item[--] To generate each bit $\vr[i]$ of $\vr$ for $i \in \{0, \ldots, \ell-1\}$, $P_0, P_j$ for $j \in \EInSet$ sample random $\vr_j[i] \in \Z{}$  and define $\vr[i] = \vr_1[i] \xor \vr_2[i]$.
		\item[--] Servers generate $\shrd$-shares of $\arval{(\vr_j[i])}$ for $i \in \{0, \ldots, \ell-1\}, j \in  \EInSet$ non-interactively following \tabref{SharesAssign}.
		\item[--] Define $\vecX$ and $\vecY$ such that $\vx =  2^{i-d+1}\arval{(\vr_1[i])}$ and $\vy_i = \arval{(\vr_2[i])}$,  respectively, for $i \in \{d,\ldots,\ell-1\}$.
		Define $\vecP$ and $\vecQ$ such that $\vp_i = 2^{i+1}\arval{(\vr_1[i])}$ and $\vq_i = \arval{(\vr_2[i])}$, respectively, for $i \in \{0, \ldots, \ell-1\}$.		
		Servers execute $\piDotP$ to compute $\shrd$-shares of $\sA = \vecX\band \vecY$ and  $\sB = \vecP\band \vecQ$. 
		\item[--] Servers locally compute $\shr{{\vrt}} = \sum_{i=d}^{\ell-1} 2^{i-d}(\shr{\arval{({\vr_1[i]})}}+\shr{\arval{(\vr_2[i])}}) - \shr{\sA}$, and $\shr{\vr} = \sum_{i=0}^{\ell-1} 2^{i}(\shr{\arval{({\vr_1[i]})}}+\shr{\arval{({\vr_2[i]})}}) - \shr{\sB}$. 
		\item[--] $P_0$ locally computes  $\bv{\vr} = \vr+\av{\vr}$.   $P_0,P_1$ set $\sqr{\vr}_1 =  -\sqr{\av{\vr}}_1$ and $P_0,P_2$ set $\sqr{\vr}_2 =  \bv{\vr}-\sqr{\av{\vr}}_2$. 
		
	\end{myitemize}  
\end{mypbox}

We now give details for generating $(\sqr{\vr}, \shr{\vrt})$. For this, servers proceed as follows: $P_0, P_j$ for $j \in \EInSet$ sample random $\vr_j \in \Z{\ell}$. Recall that the bit at $i\text{th}$ position in $\vr$ is denoted as $\vr[i]$. Define $\vr[i] = \vr_1[i] \xor \vr_2[i]$ for $i \in \{0,\ldots,\ell-1\}$. For $\vr$ defined as above, we have $\vrt[i] = \vr_1[i+d] \xor \vr_2[i+d]$ for $i \in \{0,\ldots,\ell-d-1\}$. Further,  

\begin{footnotesize}
	\begin{align} \label{eq:trunc_r}
		\vr & =  \sum_{i=0}^{\ell-1} 2^{i} \vr[i] =  \sum_{i=0}^{\ell-1} 2^{i} (\vr_1[i]\xor\vr_2[i]) \nonumber \\
		& =  \sum_{i=0}^{\ell-1} 2^{i} \left( \arval{(\vr_1[i])} + \arval{(\vr_2[i])} - 2 \arval{(\vr_1[i])} \cdot \arval{(\vr_2[i])} \right) \nonumber \\
		& = \sum_{i=0}^{\ell-1} 2^{i} \left( \arval{(\vr_1[i])} + \arval{(\vr_2[i])} \right) - \sum_{i=0}^{\ell-1} \left( 2^{i+1} \arval{(\vr_1[i])} \right) \cdot \arval{(\vr_2[i])} 
	\end{align}	
\end{footnotesize}

Similarly, for $\vrt$ we have the following,
\begin{footnotesize}
	\begin{align}\label{eq:trunc_rd}
		\vrt &= \sum_{i=d}^{\ell-1} 2^{i-d} \left( \arval{(\vr_1[i])} + \arval{(\vr_2[i])} \right)
		-  \sum_{i=d}^{\ell-1} \left( 2^{i-d+1} \arval{(\vr_1[i])} \right) \cdot \arval{(\vr_2[i])}
	\end{align}	
\end{footnotesize}

The servers non-interactively generate $\shrd$-shares (arithmetic shares) for each bit of $\vr_1$ and $\vr_2$ 
as in \tabref{SharesAssign}.
Given their $\shrd$-shares, the servers execute $\piDotP$ twice to compute $\shrd$-share of $\sA = \sum_{i=d}^{\ell-1} (2^{i-d+1}\arval{(\vr_1[i])}) \cdot \arval{(\vr_2[i])}$,  and $\sB = \sum_{i=0}^{\ell-1} (2^{i+1}\arval{(\vr_1[i])}) \cdot \arval{(\vr_2[i])}$. Using these values, the servers can locally compute the $\shrd$-shares for ($\vr,\vrt$) pair following Equation~\ref{eq:trunc_r} and \ref{eq:trunc_rd}, respectively. Note that servers need $\sqr{\cdot}$-shares of $\vr$ and not  $\shrd$-shares. The $\sqr{\cdot}$-shares can be computed from the $\shrd$-shares locally as follows. Let $(\av{\vr},\bv{\vr},\gv{\vr})$ be the values corresponding to the $\shrd$-shares of $\vr$. Since $P_0$ knows the entire value $\vr$ in clear, and it knows $\av{\vr}$, it can locally compute $\bv{\vr}$. Now, the servers set $\sqr{\cdot}$-shares as: $\sqr{\vr}_1 = -\sqr{\av{\vr}}_1$ and $\sqr{\vr}_2 = \bv{\vr}-\sqr{\av{\vr}}_2$. The protocol appears in \boxref{fig:piTr}.

\begin{lemma}[Communication]
	\label{app:piTrunc}
	Protocol $\piTrunc$ (\boxref{fig:piTr}) requires an amortized communication of $12 \ell$ bits.
\end{lemma}
\begin{proof}
	All the operations in $\piTrunc$ are non-interactive except for the two dot product calls required to compute $\sA, \sB$. The cost thus follows from Lemma \ref{app:piDotP}.
\end{proof}

\paragraph{Dot Product with Truncation} 
Given the $\shrd$-sharing of vectors $\vecX$ and $\vecY$, protocol $\piDotPTr$ (\boxref{fig:piDotPTr}) allows servers to generate $\shr{\trunc{\wz}}$, where $\trunc{\wz}$ denotes the truncated value of $\wz = \vecX \band \vecY$. A naive way is to compute the dot product using $\piDotP$, followed by performing truncation using the ($\vr, \vr^d$) pair. Instead, we follow the optimization of BLAZE where the online phase of $\piDotP$ is modified to integrate the truncation using $(\vr, \vr^d)$ at no additional cost.

The preprocessing phase now consists of the execution of one instance of $\piTrunc$ (\boxref{fig:piTr}) and the preprocessing corresponding to $\piDotP$ (\boxref{fig:piDotP}). In the online phase, servers enable $\ESet$ to obtain $\wz^\star-\vr$ instead of $\starbeta{\wz}$, where $\wz^\star = \starbeta{\wz} - \av{\wz}$. Using $\wz^\star-\vr$, both $\ESet$ then compute $(\wz - \vr)$ locally, truncate it to obtain $\trunc{(\wz-\vr)}$ and execute $\piJSh$ to generate $\shr{\trunc{(\wz-\vr)}}$. Finally, servers locally compute the result as $\shr{\trunc{\wz}} = \shr{\trunc{(\wz-\vr)}} + \shr{\vrt}$. 
The formal details for $\piDotPTr$ protocol appear in \boxref{fig:piDotPTr}. 

\begin{protocolbox}{$\piDotPTr(\Partyset, \{\shr{\wx_i},\shr{\wy_i}\}_{i \in [\nf]})$}{3PC: Dot Product Protocol with Truncation}{fig:piDotPTr}
	\justify
	\algoHead{Preprocessing:}
	\begin{myitemize}
		\item[--] Servers execute the preprocessing of $\piDotP(\Partyset, \{\shr{\wx_i},\shr{\wy_i}\}_{i \in [\nf]})$.                 
		\item[--] In parallel, servers execute $\piTrunc(\Partyset)$ to generate the truncation pair $(\sqr{\vr}, \shr{\vrt})$. 
	\end{myitemize}
	\justify\vspace{-3mm}
	\algoHead{Online:}
	\begin{myitemize}
		\item[--] $P_0, P_j$, for $j \in \EInSet$, compute $\sqrV{\Psi}{j} = -\sum_{i=1}^{\nf} ((\bv{\wx_i}+\gv{\wx_i})\sqrV{\av{\wy_i}}{j} + (\bv{\wy_i}+\gv{\wy_i})\sqrV{\av{\wx_i}}{j})  - \sqr{\vr}_j$ and set $\sqrV{{(\wz - \vr)^\star}}{j} = \sqrV{\Psi}{j} + \sqrV{\Chi}{j}$. 
		\item[--] $P_1, P_0$ $\jsend$ $\sqrV{{(\wz - \vr)}^\star}{1}$ to $P_2$ and $P_2, P_0$ $\jsend$ $\sqrV{{(\wz - \vr)}^\star}{2}$ to $P_1$.
		\item[--] $P_1,P_2$ locally compute $(\wz - \vr)^\star = \sqrV{{(\wz - \vr)}^\star}{1} + \sqrV{{(\wz - \vr)}^\star}{2}$ and set $(\wz - \vr) =  {(\wz - \vr)}^\star + \sum_{i=1}^{\nf}(\bv{\wx_i}\bv{\wy_i})  + \psi$. 
		\item[--] $P_1,P_2$ locally truncate $(\wz - \vr)$ to obtain $\trunc{(\wz - \vr)}$ and execute $\piJSh(P_1, P_2, \trunc{(\wz - \vr)})$ to generate $\shr{\trunc{(\wz - \vr)}}$.
		\item[--] Servers locally compute $ \shr{\wz} = \shr{\trunc{(\wz - \vr)}} + \shr{\vrt}$ .
	\end{myitemize}
\end{protocolbox}

\begin{lemma}[Communication]
	\label{app:piDotPTr}
	Protocol $\piDotPTr$ (\boxref{fig:piDotPTr}) requires an amortized communication of $15 \ell$ bits in the preprocessing phase and requires $1$ round and an amortized communication of $3 \ell$ bits in the online phase.
\end{lemma}
\begin{proof}
	During the preprocessing phase, servers execute the preprocessing phase of $\piDotP$, resulting in an amortized communication of $3 \ell $ bits (Lemma~\ref{app:piDotP}). In parallel, servers execute one instance of $\piTrunc$ protocol resulting in an additional communication of $12 \ell$ bits (Lemma~\ref{app:piTrunc}).
	
	The online phase follows from that of $\piDotP$ protocol except that, now, $\ESet$ compute additive shares of $\wz - \vr$, where $\wz = \vecX \band \vecY$, which is achieved using two executions of $\piJmp$ in parallel. This requires one round and an amortized communication cost of $2 \ell$ bits. $\ESet$ then jointly share the truncated value of $\wz - \vr$ with $P_0$, which requires one round and $\ell$ bits. However, this step can be deferred till the end for multiple dot product with truncation instances, which amortizes the cost. 
	\
\end{proof}

\paragraph{Secure Comparison}
Secure comparison allows servers to check whether $\wx < \wy$, given their $\shrd$-shares. In FPA representation,  checking $\wx < \wy$ is equivalent to checking the $\MSB$ of $\val = \wx - \wy$. Towards this, servers locally compute $\shr{\val} = \shr{\wx} - \shr{\wy}$ and extract the $\MSB$ of $\val$ using $\piBitExt$. In case an arithmetic sharing is desired, servers can apply $\PiBitA$ (\boxref{fig:piBit2A}) protocol on the outcome of $\piBitExt$ protocol.

\paragraph{Activation Functions} 
We now elaborate on two of the most prominently used activation functions: i) Rectified Linear Unit (ReLU) and (ii) Sigmoid (Sig).

{\em (i) ReLU:}
The ReLU function, $\ReLU(\val) = \maxv(0, \val)$, can be viewed as $\ReLU(\val) = \overline{\bitb} \cdot \val$, where bit $\bitb = 1$ if $\val < 0$ and $0$ otherwise. Here $\overline{\bitb}$ denotes the complement of $\bitb$. 
Given $\shr{\val}$, servers execute $\piBitExt$ on $\shr{\val}$ to generate $\shrB{\bitb}$. $\shrB{\cdot}$-sharing of $\overline{\bitb}$ is locally computed by setting $\bv{\overline{\bitb}} = 1 \xor \bv{\bitb}$. Servers execute $\PiBitInj$ protocol on $\shrB{\overline{\bitb}}$ and $\shr{\val}$ to obtain the desired result.

\begin{lemma}[Communication]
	\label{app:piReLU}
	Protocol $\ReLU$ requires an amortized communication of $21\ell-6$ bits in the preprocessing phase and requires $\log \ell + 3$ rounds and an amortized communication of $16\ell - 6$ bits in the online phase.
\end{lemma}
\begin{proof}
	One instance of $\ReLU$ protocol comprises of execution of one instance of $\piBitExt$, followed by $\PiBitInj$. The cost, therefore, follows from Lemma~\ref{app:piBitExt2},  and Lemma~\ref{app:PiBitInj}. 
\end{proof}

{\em (ii) Sig:}
In this work, we use the MPC-friendly variant of the Sigmoid function~\cite{MohasselZ17, MR18, ASTRA}.
Note that $\Sig(\val) = \overline{\bitb_1} \bitb_2 (\val + 1/2) + \overline{\bitb_2}$, where $\bitb_1 = 1$ if $\val + 1/2 < 0$ and $\bitb_2 = 1$ if $\val - 1/2 < 0$. To compute $\shr{\Sig(\val)}$, servers proceed in a similar fashion as in ReLU, and hence, we skip the details.

The formal details of the MPC-friendly variant of the Sigmoid function~\cite{MohasselZ17, MR18, ASTRA} is given below:
\begin{align*}
	\Sig(\val) = \left\{
	\begin{array}{lll}
		0                  & \quad \val < -\frac{1}{2} \\
		\val + \frac{1}{2} & \quad - \frac{1}{2} \leq \val \leq \frac{1}{2} \\
		1                  & \quad \val > \frac{1}{2}
	\end{array}
	\right.
\end{align*}

\begin{lemma}[Communication]
	\label{app:piSig}
	Protocol $\Sig$ requires an amortized communication of $39\ell - 9$ bits in the preprocessing phase and requires $\log \ell + 4$ rounds and an amortized communication of $29\ell-9$ bits in the online phase.
\end{lemma}
\begin{proof}
	An instance of $\Sig$ protocol involves the execution of the following protocols in order-- i) two parallel instances of $\piBitExt$ protocol, ii) once instance of $\piMult$ protocol over boolean value, and iii) one instance of $\PiBitInj$ and $\PiBitA$ in parallel. The cost follows from Lemma~\ref{app:piBitExt2}, Lemma~\ref{app:piBitA} and Lemma~\ref{app:PiBitInj}. 
\end{proof}

\paragraph{Maxpool, Matrix Operations and Convolutions}
The goal of maxpool is to find the maximum value in a vector $\vecX$ of $m$ values. Maximum between two elements $\vx_i$, $\vx_j$ can be computed by applying secure comparison, which returns a binary sharing of a bit $\bitb$ such that $\bitb= 0$ if $\vx_i > \vx_j$, or $1$, otherwise, followed by computing $(\bitb)^{\bf B}(\vx_j - \vx_i) + \vx_i$, which can be performed using bit injection (\ref{par:bitInj}).
To find the maximum value in vector $\vecX$, the servers first group the values in $\vecX$ into pairs and securely compare each pair to obtain the maximum of the two. This results in a vector of size $m / 2$. This process is repeated for $\BigO{\log m}$ rounds to obtain the maximum value in the entire vector. 

Linear matrix operations, such as addition of two matrices $\Mat{A}, \Mat{B}$ to generate matrix $\Mat{C} = \Mat{A} + \Mat{B}$, can be computed by extending the scalar operations (addition, in this case) with respect to each element of the matrix. Matrix multiplication, on the other hand, can be expressed as a collection of dot products, where the element in the $i^{\text{th}}$ row and $j^{\text{th}}$ column of $\Mat{C} = \Mat{A} \times \Mat{B}$, where $\Mat{A}, \Mat{B}$ are matrices of dimension $\vp \times \vq$, $\vq \times \vr$, respectively, can be computed as a dot product of the $i^{\text{th}}$ row of $\Mat{A}$ and the $j^{\text{th}}$ column of $\Mat{B}$. Thus, computing $\Mat{C}$ of dimension $\vp \times \vr$ requires $\vp \vr$ dot products whose communication cost (amortized) is equal to that of computing $\vp \vr$ multiplications in our case. This improves the cost of matrix multiplication over the naive approach which requires $\vp \vq \vr$ multiplications. 

Convolutions form an important building block in several neural network architectures and can be represented as matrix multiplications, as explained in the example below. Consider a 2-dimensional convolution ($\cv$) of a $3 \times 3$ input matrix $\Mat{X}$ with a kernel $\Mat{K}$ of size $2 \times 2$. This can be represented as a matrix multiplication as follows. 


\[ 
\cv \left(
\begin{bmatrix}
\vx_1  &  \vx_2  & \vx_3 \\
\vx_4  &  \vx_5  & \vx_6 \\
\vx_7  &  \vx_8  & \vx_9      
\end{bmatrix}
,
\begin{bmatrix}
\vk_1  &  \vk_2      \\
\vk_3  &  \vk_4      
\end{bmatrix} 
\right)
= 
\begin{bmatrix}
\vx_1  &  \vx_2  & \vx_4 & \vx_5 \\
\vx_2  &  \vx_3  & \vx_5 & \vx_6 \\
\vx_4  &  \vx_5  & \vx_7 & \vx_8 \\   
\vx_5  &  \vx_6  & \vx_8 & \vx_9     
\end{bmatrix}
\begin{bmatrix}
\vk_1  \\
\vk_2  \\
\vk_3  \\
\vk_4      
\end{bmatrix} 
\]


Generally, convolving a $f \times f$ kernel over a $w \times h$ input with $p \times p$ padding using $s \times s$ stride having $i$ input channels and $o$ output channels, is equivalent to performing a matrix multiplication on matrices of dimension $(w^{\prime} \cdot h^{\prime}) \times (i \cdot f \cdot f)$ and $(i \cdot f \cdot f) \times (o)$ where $w^{\prime} = \dfrac{w-f+2p}{
	s} + 1$ and $h^{\prime} = \dfrac{h-f+2p}{
	s} + 1$. We refer readers to \cite{WaghGC19} (cf. ``Linear and Convolutional Layer'') and \cite{ConvStanford} for more details.

%% file: Main_4PC.tex
In this section, we extend our 3PC results to the 4-party case and observe substantial efficiency gain. First, the use of broadcast is eliminated. Second, the preprocessing of multiplication becomes substantially computationally light, eliminating the multiplication protocol (used in the preprocessing) altogether. Third, we achieve a dot product protocol with communication cost independent of the size of the vector, completely eliminating the complex machinery required as in the 3PC case. At the heart of our 4PC constructions lies an efficient  4-party $\Jmp$ primitive, denoted as $\JmpF$, that allows two servers to send a common value to a third server robustly. 

This section is organized as follows. We begin with the secret-sharing semantics for $4$ servers, for which we only use an extended version of $\shr{\cdot}$-sharing. We then explain the joint message passing primitive for four servers, followed by our 4PC protocols. We conclude this section with a detailed analysis about achieving private robustness.

\paragraph{Secret Sharing Semantics} 
For a value $\val$, the shares for $P_0, P_1$ and $P_2$ remain the same as that for 3PC case. That is, $P_0$ holds $(\sqrA{\av{\val}}, \sqrB{\av{\val}}, \bv{\val}+\gv{\val})$ while $P_i$ for $i \in \EInSet$ holds $( \sqr{\av{\val}}_i, \bv{\val}, \gv{\val})$. The shares for the fourth server $P_3$ is defined as $(\sqrA{\av{\val}}, \sqrB{\av{\val}}, \gv{\val})$. Clearly,  the secret is defined as $\val = \bv{\val} - \sqrA{\av{\val}} - \sqrB{\av{\val}}$.

\smallskip
\subsection{4PC Joint Message Passing Primitive}
The $\JmpF$ primitive enables two servers $P_i$, $P_j$ to send a common value $\val \in \Z{\ell}$ to a third server $P_k$, or identify a $\TTP$ in case of any inconsistency. This primitive is analogous to $\Jmp$~(\boxref{fig:piJmp}) in spirit but is significantly  optimized and free from broadcast calls.  Similar to the 3PC counterpart, each server maintains a bit and  $P_i$ sends the value, and $P_j$ the hash of it to $P_k$. $P_k$ sets its inconsistency bit to $1$ when  the (value, hash) pair is inconsistent. This  is followed by relaying the bit to all the servers, who exchange it among themselves and agree on the bit that forms majority ($1$ indicates the presence of inconsistency, and $0$ indicates consistency). The presence of an honest majority among $P_i, P_j, P_l$, guarantees agreement on the presence/absence of an inconsistency as conveyed by $P_k$.
Observe that inconsistency can only be caused either due to a corrupt sender sending an incorrect value (or hash), or a corrupt receiver falsely announcing the presence of inconsistency. Hence, the fourth server, $P_l$, can safely be employed as $\TTP$. The ideal functionality appears in \boxref{fig:JmpFFunc}, and the protocol appears in \boxref{fig:p4pcjmp}.

\begin{notation} \label{jmp4-send}
	We say that $P_i,P_j$ $\jsendf$ $\val$ to $P_k$ when they invoke $\piJmpF(P_i, P_j, P_k, \val,P_l)$. 
\end{notation}

%
We note that the end goal of $\JmpF$ primitive relates closely to the bi-convey primitive of FLASH \cite{FLASH}. Bi-convey allows two servers $S_1,S_2$ to convey a value to a server $R$, and in case of an inconsistency, a pair of honest servers mutually identify each other, followed by exchanging their internal randomness to recover the clear inputs, computing the circuit, and sending the output to all.
Note, however, that $\JmpF$ primitive is more efficient and differs significantly in techniques from the bi-convey primitive. Unlike in bi-convey, in case of an inconsistency, $\JmpF$ enables servers to learn the $\TTP$'s identity unanimously. Moreover, bi-convey demands that honest servers, identified during an inconsistency, exchange their internal randomness (which comprises of the shared keys established during the key-setup phase) to proceed with the computation. This enforces the need for a fresh key-setup every time inconsistency is detected. 
On the efficiency front,  $\JmpF$  simply halves the communication cost of bi-convey, giving a $2\times$ improvement. 

\begin{systembox}{$\FJmpF$}{4PC: Ideal functionality for $\JmpF$ primitive}{fig:JmpFFunc}
	\justify
	$\FJmpF$ interacts with the servers in $\Partyset$ and the adversary $\Sim$. 
	\begin{myitemize}
		\item[\bf Step 1:] $\FJmp$ receives $(\INPUT,\val_s)$ from senders $P_s$ for $s \in \{i,j\}$, $(\INPUT,\bot)$ from receiver $P_k$ and fourth server $P_l$, while it receives $(\SELECT,\ttp)$ from $\Sim$. Here $\ttp$ is a boolean value, with a $1$ indicating that $\TTP = P_l$ should be established. 
		\item[\bf Step 2:] If $\val_i =\val_j$ and $\ttp = 0$, or if $\Sim$ has corrupted $P_l$, set $\msg_i = \msg_j = \msg_l = \bot, \msg_k = \val_i$ and go to {\bf Step 4}.
		\item[\bf Step 3:] Else : Set $\msg_i = \msg_j = \msg_k = \msg_l = P_l$.
		\item[\bf Step 4:] Send $(\OUTPUT, \msg_s)$ to $P_s$ for $s \in \{0,1,2, 3\}$.
	\end{myitemize}
\end{systembox}

\begin{mypbox}{$\piJmpF(P_i, P_j, P_k, \val,P_l)$}{4PC: Joint Message Passing Primitive}{fig:p4pcjmp}
	$P_s \in \Partyset$ initializes an inconsistency bit $\bitb_s = 0$. If $P_s$ remains silent instead of sending $\bitb_s$ in any of the following rounds, the recipient sets $\bitb_s$ to $1$.
	
	\justify
	{\em Send Phase:} $P_i$ sends $\val$ to $P_k$.
	
	\noindent {\em Verify Phase:} $P_j$ sends $\Hash(\val)$ to $P_k$. 
	\begin{myitemize} 
		\item[--]  
		$P_k$ sets $\bitb_k = 1$ if the received values are inconsistent or if the value is not received. 
		\item[--] $P_k$ sends $\bitb_k$ to all servers. $P_s$ for $s \in \{i, j, l\}$ sets $\bitb_s = \bitb_k$.
		\item[--] $P_s$ for $s \in \{i, j, l\}$ mutually exchange their bits. $P_s$ resets $\bitb_s = \bitb^{\prime}$ where $\bitb^{\prime}$ denotes the bit which appears in majority among $\bitb_i, \bitb_j, \bitb_l$.
		\item[--] All servers set $\TTP = P_l$ if $\bitb^{\prime} = 1$, terminate otherwise.
	\end{myitemize}
\end{mypbox}

\begin{lemma}[Communication]
	\label{app:piJmpF}
	Protocol $\piJmpF$ (\boxref{fig:p4pcjmp}) requires $1$ round and an amortized communication of $\ell$ bits in the online phase.
\end{lemma}
\begin{proof}
	Server $P_i$ sends the value $\val$ to $P_k$ while $P_j$ sends hash of the same to $P_k$. This accounts for one round of communication. Values sent by $P_j$ for several instances can be concatenated and hashed to obtain a single value. Hence the cost of sending the hash gets amortized over multiple instances. Similarly, the two round exchange of inconsistency bits to establish a $\TTP$ can be combined for multiple instances, thereby amortizing this cost. Thus, the amortized cost of this protocol is $\ell$ bits.
\end{proof}

\subsection{4PC Protocols}
In this section, we revisit the protocols from 3PC (\S\ref{sec:3PC}) and suggest optimizations leveraging the presence of an additional honest party in the system. We provide security proofs in \S\ref{app:sec4PC}. 

\paragraph{Sharing Protocol}  
To enable $P_i$ to share a value $\val$, protocol $\pifinput$ (\boxref{fig:4pcsharing}) proceeds similar to that of 3PC case with the addition that $P_3$ also samples the values $\sqr{\av{\val}}_1, \sqr{\av{\val}}_2, \gv{\val}$ using the shared randomness with the respective servers. On a high level, $P_i$ computes $\bv{\val} = \val + \sqr{\av{\val}}_1 + \sqr{\av{\val}}_2$ and sends $\bv{\val}$ (or $\bv{\val} + \gv{\val}$) to another server and they together $\jsendf$  this information to the intended servers. The formal protocol for sharing a value $\val$ by $P_i$ is given in \boxref{fig:4pcsharing} 

\begin{protocolbox}{$\pifinput (P_i, \val)$}{4PC: Generating $\shr{\val}$-shares by server $P_i$}{fig:4pcsharing}
	
	\algoHead{Preprocessing:}
	\begin{myitemize}
		\item[--] If $P_i = P_0$ : $P_0,P_3,P_j$, for $j \in \EInSet$, together sample random $\sqr{\av{\val}}_j \in \Z{\ell}$, while $\Partyset$ sample random $\gv{\val} \in \Z{\ell}$.
		\item[--] If $P_i = P_1$ : $P_0,P_3,P_1$ together sample random $\sqr{\av{\val}}_1 \in \Z{\ell}$, while $\Partyset$ sample a random $\sqr{\av{\val}}_2 \in \Z{\ell}$. Also, $\ESet, P_3$ sample random $\gv{\val} \in \Z{\ell}$.
		\item[--] If $P_i = P_2$: Analogous to the case when $P_i = P_1$.
		\item[--] If $P_i = P_3$: $P_0, P_3, P_j$, for $j \in \{1,2\}$, sample random $\sqr{\av{\val}}_j \in \Z{\ell}$. $P_1, P_2, P_3$ together sample random $\gv{\val} \in \Z{\ell}$.  
	\end{myitemize}
	\justify\vspace{-3mm}
	\algoHead{Online:}
	\begin{myitemize}
		\item[--] If $P_i = P_0$ : $P_0$ computes $\bv{\val} = \val+\av{\val}$ and sends $\bv{\val}$ to $P_1$. $P_0, P_1$ $\jsendf$ $\bv{\val}$ to $P_2$. 
		\item[--] If $P_i = P_j$, for $j \in\EInSet$ : $P_j$ computes $\bv{\val} = \val+\av{\val}$, sends $\bv{\val}$ to $P_{3-j}$. 
		$P_1, P_2$ $\jsendf$ $\bv{\val}+\gv{\val}$ to $P_0$. 
		\item[--] If $P_i = P_3$: $P_3$ sends $\bv{\val} +  \gv{\val} = \val + \av{\val} + \gv{\val} $ to $P_0$. 
		$P_3, P_0$ $\jsendf$ $\bv{\val}+\gv{\val}$ to both $P_1$ and $P_2$. 
	\end{myitemize}
	
\end{protocolbox}

\begin{lemma}[Communication]
	\label{app:pifinput}
	In the online phase, $\pifinput$ (\boxref{fig:4pcsharing}) requires $2$ rounds and an amortized communication of $2 \ell$ bits when $\PSet$ share a value, whereas it requires an amortized communication of $3 \ell$ bits when $P_3$ shares a value.
\end{lemma}
\begin{proof}
	The proof for $\PSet$ sharing a value follows from \ref{app:piSh}. For the case when $P_3$ wants to share a value $\val$, it first sends $\bv{\val} + \gv{\val}$ to $P_0$ which requires one round and $\ell$ bits of communication. This is followed by $2$ parallel calls to $\piJmpF$ which together require one round and an amortized communication of $2 \ell$ bits. 
\end{proof}

\vspace{-4mm}
\paragraph{Joint Sharing Protocol} 
Protocol $\pifjsh$ enables a pair of (unordered) servers $(P_i, P_j)$ to jointly generate a $\shrd$-sharing of value $\val \in \Z{\ell}$ known to both of them. In case of an inconsistency, the server outside the computation serves as a $\TTP$. The protocol is described in \boxref{fig:p4pcjsh}.

When $P_3, P_0$ want to jointly share a value $\val$ which is available in the preprocessing phase, protocol $\pifjsh$ can be performed with a single element of communication (as opposed to $2$ elements in \boxref{fig:p4pcjsh}).
$P_0, P_3$ can jointly share $\val$ as follows. $P_0, P_3, P_1$ sample a random $\vr \in \Z{\ell}$ and set $\sqrA{\av{\val}} = \vr$. $P_0, P_3$ set $\sqrB{\av{\val}} = - (\vr + \val)$ and $\jsendf$ $\sqrB{\av{\val}}$ to $P_2$. This is followed by servers locally setting $\gv{\val} = \bv{\val} = 0$.

We further observe that servers can generate a $\shr{\cdot}$-sharing of $\val$ non-interactively when $\val$ is available with $P_0, P_1, P_2$. For this, servers set $\sqrA{\av{\val}} = \sqrB{\av{\val}} = \gv{\val} = 0$ and $\bv{\val} = \val$. We abuse notation and use $\pifjsh(P_0, P_1, P_2, \val)$ to denote this sharing.

\begin{lemma}[Communication]
	\label{app:pifjsh}
	In the online phase, $\pifjsh$ (\boxref{fig:p4pcjsh}) requires $1$ round and an amortized communication of $2 \ell$ bits when $(P_3, P_s)$ for $s \in \{0, 1, 2\}$ share a value, and requires an amortized communication of $\ell$ bits, otherwise.
\end{lemma}
\begin{proof}
	When $(P_3, P_s)$ for $s \in \{0, 1, 2\}$ want to share a value $\val$, there are two parallel calls to $\piJmpF$ which requires an amortized communication of $2 \ell$ bits and one round. In the other cases, $\piJmpF$ is invoked only once, resulting in an amortized communication of $\ell$ bits.
\end{proof}

\begin{protocolbox}{$\pifjsh (P_i, P_j, \val)$}{4PC: $\shrd$-sharing of a value $\val \in \Z{\ell}$ jointly by $P_i, P_j$}{fig:p4pcjsh}
	\algoHead{Preprocessing:}
	\begin{myitemize}
		\item[--] If $(P_i, P_j) = (P_1, P_2)$ : $P_1, P_2, P_3$ sample $\gv{\val} \in \Z{\ell}$. Servers locally set $\sqr{\av{\val}}_1 = \sqr{\av{\val}}_2 = 0$.
		\item[--] If $(P_i, P_j) = (P_s, P_0)$, for $s \in \EInSet$ : Servers execute the preprocessing of $\pifinput(P_s, \val)$. Servers locally set $\gv{\val} = 0$.
		\item[--] If $(P_i, P_j) = (P_s, P_3)$, for $s \in \{0,1,2\}$ : Servers execute the preprocessing of $\pifinput(P_s, \val)$.
	\end{myitemize}
	\justify\vspace{-3mm}
	\algoHead{Online:}
	\begin{myitemize}
		\item[--] If $(P_i, P_j) = (P_1, P_2)$ : $P_1, P_2$ set $\bv{\val} = \val$ and $\jsendf$ $\bv{\val}+\gv{\val}$ to $P_0$. 
		\item[--] If $(P_i, P_j) = (P_s, P_0)$, for $s \in \{1,2,3\}$ : $P_s, P_0$ compute $\bv{\val} = \val + \sqr{\av{\val}}_1 + \sqr{\av{\val}}_2$ and $\jsendf$ $\bv{\val}$ to $P_k$, where $(k \in \EInSet) \wedge (k \neq s)$.
		\item[--] If $(P_i, P_j) = (P_s, P_3)$, for $s \in \EInSet$: 
		$P_3, P_s$ compute $\bv{\val}$ and $\bv{\val} + \gv{\val}$. $P_s, P_3$ $\jsendf$ $\bv{\val}$ to $P_k$, where $(k \in \EInSet) \wedge (k \neq s)$. In parallel, $P_s, P_3$ $\jsendf$ $\bv{\val} + \gv{\val}$ to $P_0$.
	\end{myitemize}
\end{protocolbox}

\paragraph{$\sgr{\cdot}$-sharing Protocol}
\label{app:sgrsh}
In some protocols, $P_3$ is required to generate $\sgr{\cdot}$-sharing of a value $\val$ in the preprocessing phase, where $\sgr{\cdot}$-sharing of $\val$ is same as that defined in 3PC (where $\val = \val_0 + \val_1 + \val_2$, and $P_0$ possesses $(\val_0, \val_1)$, $P_1$ possesses $(\val_1, \val_2)$, and $P_2$ possess $(\val_2, \val_0)$) with the addition that $P_3$ now possesses $(\val_0,\val_1,\val_2)$. We call the resultant protocol $\pifsgrsh$ and it appears in \boxref{fig:p4pcash}.

\smallskip
\begin{protocolbox}{$\pifsgrsh (P_3, \val)$}{4PC: $\sgr{\cdot}$-sharing of value $\val$ by $P_3$}{fig:p4pcash}
	
	\algoHead{Preprocessing}:
	\begin{myitemize}
		\item[--] Servers $P_0, P_3, P_1$ sample a random $\val_1 \in \Z{\ell}$, while servers $P_0, P_3, P_2$ sample a random $\val_0 \in \Z{\ell}$. 
		\item[--] $P_3$ computes $\val_2 = \val - \val_0 - \val_1$ and sends $\val_2$ to $P_2$. $P_3, P_2$ $\jsendf$ $\val_2$ to $P_1$. 
	\end{myitemize}
\end{protocolbox}

\medskip
Note that servers can locally convert $\sgr{\val}$ to $\shr{\val}$ by setting their shares as shown in \tabref{sgrtoshr}.

\bigskip 
\begin{table}[htb!]
	\centering
	\resizebox{.95\textwidth}{!}{
		\begin{tabular}{r | c | c | c | c }
			\toprule
			& $P_0$ & $P_1$ & $P_2$ & $P_3$ \\ 
			\midrule
			$\shr{\val}$ & $(-\val_1, -\val_0, 0)$ & $(-\val_1, \val_2,  -\val_2)$ & $(-\val_0, \val_2, -\val_2)$ & $(-\val_0, -\val_1, -\val_2)$\\ 
			\bottomrule
		\end{tabular}
	}
	\vspace{-2mm}
	\caption{\small Local conversion of shares from $\sgr{\cdot}$-sharing to $\shr{\cdot}$-sharing for a value $\val$. Here, $\sqrA{\av{\val}} = -\val_1, \sqrB{\av{\val}} = -\val_0, \bv{\val} = \val_2, \gv{\val} = -\val_2$. \label{tab:sgrtoshr}}
\end{table}

\begin{lemma}[Communication]
	\label{app:sgrshl}
	Protocol $\pifsgrsh$ (\boxref{fig:p4pcash}) requires $2$ rounds and an amortized communication of $2 \ell$ bits.
\end{lemma}
\begin{proof}
	Communicating $\val_2$ to $P_2$ requires $\ell$ bits and $1$ round. This is followed by one invocation of $\piJmpF$ which requires $\ell$ bits and $1$ round. Thus, the amortized communication cost is $2 \ell$ bits and two rounds.
\end{proof}

\paragraph{Multiplication Protocol} 
Given the $\shr{\cdot}$-shares of $\wx$ and $\wy$, protocol $\pifmul$ (\boxref{fig:4pcmul}) allows servers to compute $\shr{\wz}$ with $\wz = \wx \wy$. When compared with the state-of-the-art 4PC GOD protocol of FLASH~\cite{FLASH}, our solution improves communication in both, the preprocessing and online phase, from $6$ to $3$ ring elements. Moreover, our communication cost matches with the state-of-the-art 4PC protocol of Trident~\cite{Trident} that only provides security with fairness.

Recall that the goal of preprocessing in 3PC multiplication was to enable $P_1, P_2$ obtain $\psi$, and $P_0, P_i$ for $i \in \EInSet$ obtain $\sqr{\Chi}_i$ where $\Chi = \gv{\wx}\av{\wy} + \gv{\wy}\av{\wx} + \Gammaxy - \psi$. Here $\psi$ is a random value known to both $P_1, P_2$. With the help of $P_3$, we let the servers obtain the respective preprocessing data as follows: $P_0, P_3, P_1$ together samples random $\sqr{\Gammaxy}_1 \in \Z{\ell}$. $P_0,P_3$ locally compute $\Gammaxy = \av{\wx} \av{\wy}$, set $\sqr{\Gammaxy}_2 = \Gammaxy - \sqr{\Gammaxy}_1$ and $\jsendf$ $\sqr{\Gammaxy}_2$ to $P_2$. $P_1, P_2, P_3$ locally sample $\psi, \vr$ and generate $\sqr{\cdot}$-shares of $\psi$ by setting $\sqr{\psi}_1 = r$ and $\sqr{\psi}_2 = \psi - \vr$.
Then $P_j, P_3$ for $j \in \{1,2\}$ compute $\sqr{\Chi}_j = \gv{\wx}\sqr{\av{\wy}}_j + \gv{\wy}\sqr{\av{\wx}}_j + \sqr{\Gammaxy}_j - \sqr{\psi}_j$ and $\jsendf$ $\sqr{\Chi}_j$ to $P_0$. The online phase is similar to that of 3PC, apart from $\piJmpF$ being used instead of $\piJmp$ for communication. Since $P_3$ is not involved in the online computation phase, we can safely assume $P_3$ to serve as the $\TTP$ for the $\piJmpF$ executions in the online phase. 

\begin{protocolbox}{$\pifmul (\Partyset, \shr{\wx}, \shr{\wy})$}{4PC: Multiplication Protocol ($\wz = \wx \cdot \wy$)}{fig:4pcmul}
	\justify
	\algoHead{Preprocessing:}
	\begin{myitemize}
		\item[--] $P_0,P_3,P_j$, for $j \in \EInSet$, sample random $\sqr{\av{\wz}}_{j}\in \Z{\ell}$, and $P_0, P_1, P_3$ sample random $\sqr{\Gammaxy}_1 \in \Z{\ell}$.
		\item[--]  	$P_1, P_2, P_3$ sample random $\gv{\wz}, \psi, \vr \in \Z{\ell}$ and set $\sqrA{\psi} = \vr,\allowbreak \sqrB{\psi} = \psi - \vr$.
		\item[--] $P_0, P_3$ set $\sqr{\Gammaxy}_2 = \Gammaxy - \sqr{\Gammaxy}_1$, where $\Gammaxy = \av{\wx}\av{\wy}$. $P_0, P_3$ $\jsendf$ $\sqr{\Gammaxy}_2$ to $P_2$.
		\item[--] $P_3, P_j$, for $j \in \EInSet$, set $\sqr{\Chi}_j = \gv{\wx}\sqr{\av{\wy}}_j + \gv{\wy}\sqr{\av{\wx}}_j + \sqr{\Gammaxy}_j - \sqr{\psi}_j$. 
		$P_1, P_3$ $\jsendf$ $\sqr{\Chi}_1$ to $P_0$, while $P_2, P_3$ $\jsendf$ $\sqr{\Chi}_2$ to $P_0$.
	\end{myitemize}
	\justify\vspace{-3mm}
	\algoHead{Online:}
	\begin{myitemize}
		\item[--] $P_0, P_j$, for $j \in \EInSet$, compute $\sqr{\starbeta{\wz}}_j = -(\bv{\wx}+\gv{\wx})\sqr{\av{\wy}}_j - (\bv{\wy}+\gv{\wy})\sqr{\av{\wx}}_j + \sqr{\av{\wz}}_j + \sqr{\Chi}_j$. 
		\item[--] $P_1, P_0$ $\jsendf$ $\sqr{\starbeta{\wz}}_1$ to $P_2$, and $P_2, P_0$ $\jsendf$ $\sqr{\starbeta{\wz}}_2$ to $P_1$.
		\item[--] $P_j$, for $j \in \EInSet$, computes $\starbeta{\wz} = \sqr{\starbeta{\wz}}_1 + \sqr{\starbeta{\wz}}_2$ and sets $\bv{\wz} = \starbeta{\wz} + \bv{\wx}\bv{\wy} + \psi$.
		\item[--] $P_1, P_2$ $\jsendf$ $\bv{\wz}+\gv{\wz}$ to $P_0$.
	\end{myitemize}	
\end{protocolbox}
\smallskip

\begin{lemma}[Communication]
	\label{app:pifmul}
	$\pifmul$ (\boxref{fig:4pcmul}) requires an amortized communication of $3 \ell$ bits in the preprocessing phase, and $1$ round with an amortized communication of $3 \ell$ bits in the online phase.
\end{lemma}
\begin{proof}
	In the preprocessing phase, the servers execute $\piJmpF$ to $\jsendf$ $\sqrB{\Gammaxy}$ to $P_2$ resulting in amortized communication of $\ell$ bits. This is followed by $2$ parallel invocations of $\piJmpF$ to $\jsendf$ $\sqrA{\Chi}, \sqrB{\Chi}$ to $P_0$ which require an amortized communication of $2 \ell$ bits. Thus, the amortized communication cost in preprocessing is $3 \ell$ bits. In the online phase, there are $2$ parallel invocations of $\piJmpF$ to $\jsendf$ $\sqrA{\starbeta{\wz}}, \sqrB{\starbeta{\wz}}$ to $P_2, P_1$, respectively, which requires amortized communication of $2 \ell$ bits and one round. This is followed by another call to $\piJmpF$ to $\jsendf$ $\bv{\wz} + \gv{\wz}$ to $P_0$ which requires one more round and amortized communication of $\ell$ bits. However, $\jsendf$ of $\bv{\wz} + \gv{\wz}$ can be delayed till the end of the protocol, and will require only one round for multiple multiplication gates and hence, can be amortized. Thus, the total number of rounds required for multiplication in the online phase is one with an amortized communication of $3 \ell$ bits.
\end{proof}

\paragraph{Reconstruction Protocol} 
Given $\shr{\val}$, protocol $\pifrec$ (\boxref{fig:4pcrec}) enables servers to robustly reconstruct the value $\val$ among the servers. Note that every server lacks one share for reconstruction and the same is available with three other servers. Hence, they communicate the missing share among themselves, and the majority value is accepted. As an optimization, two among the three servers can send the missing share while the third one can send a hash of the same for verification. Notice that, as opposed to the 3PC case, this protocol does not require commitments. 
The formal protocol for reconstruction is given in \boxref{fig:4pcrec}.

\begin{protocolbox}{$\pifrec (\Partyset, \shr{\val})$}{4PC: Reconstruction of $\val$ among the servers}{fig:4pcrec}
	\justify 
	\algoHead{Online}
	\begin{myitemize}
		\item[--] $P_0$ receives $\gv{\val}$ from $\ESet$ and $\Hash(\gv{\val})$ from $P_3$.
		\item[--] $P_1$ receives $\sqrB{\av{\val}}$ from $P_2, P_3$ and $\Hash(\sqrB{\av{\val}})$ from $P_0$. 
		\item[--] $P_2$ receives $\sqrA{\av{\val}}$ from $P_0, P_3$ and $\Hash(\sqrA{\av{\val}})$ from $P_1$. 
		\item[--] $P_3$ receives $\bv{\val} + \gv{\val}$ from $P_0, P_1$ and $\Hash(\bv{\val} + \gv{\val})$ from $P_2$. 
		\item[--] $P_i \in \Partyset$ selects the missing share forming the majority among the values received and reconstructs the output.
	\end{myitemize}
	
\end{protocolbox}

\begin{lemma}[Communication]
	\label{app:pifrec}
	$\pifrec$ (\boxref{fig:4pcrec}) requires an amortized communication of $8 \ell$ bits and $1$ round in the online phase.
\end{lemma}
\begin{proof}
	Each $P_s$ for $s \in \{0,1,2,3\}$ receives the missing share in clear from two other servers, while the hash of it from the third. As before, the missing share sent by the third server can be concatenated over multiple instances and hashed to obtain a single value. Thus, the amortized communication cost is $2 \ell$ bits per server, resulting in a total cost of $8 \ell$ bits.
\end{proof}

\smallskip
\subsection{Building Blocks for PPML using 4PC}
\label{sec:4PPML}
This section provides details on robust  realizations of the PPML building blocks in 4-server setting (for the same blocks as in  \S\ref{sec:PPML}). We provide the security proofs in \S\ref{app:sec4PC}. 

\paragraph{Input Sharing and Output Reconstruction in SOC Setting}
\label{ba} 
We extend input sharing and reconstruction in the SOC setting as follows. To generate $\shr{\cdot}$-shares for its input $\val$, $\User$ receives each of the shares $\sqrA{\av{\val}}, \sqrB{\av{\val}}$, and $\gv{\val}$ from three out of the four servers as well as a random value $\vr \in \Z{\ell}$ sampled together by $\PSet$ and accepts the values that form the majority. $\User$ locally  computes $\vu = \val + \sqrA{\av{\val}} + \sqrB{\av{\val}} + \gv{\val} + \vr$ and sends $\vu$ to all the servers. Servers then execute a two round byzantine agreement (BA)~\cite{PeaseSL80} to agree on $\vu$ (or $\bot$). 
%
At a high-level, the BA protocol proceeds as follows. Let us denote the value received by $P_i$ from $\User$ as $\vu_i$. To agree on $\vu$ received from $\User$, the servers first arrive on an agreement regarding each $\vu_i$ received by $P_i$. This is followed by selecting the majority value among $\vu_1, \vu_2, \vu_3, \vu_4$. For servers to agree on $\vu_i$, $P_i$ first sends $\vu_i$ to all servers. This is followed by $P_j \in \Partyset \backslash P_i$ exchanging $\vu_i$ among themselves. Thus, each $P_j \in \Partyset \backslash P_i$ receives three versions of $\vu_i$ and sets the majority value among the three values received as $\vu_i$. Since there can be at most one corruption among the servers, the majority rule ensures that all honest servers are on the same page. Once each of the values are agreed on, every server takes the majority among $\vu_1, \vu_2, \vu_3, \vu_4$ as the value sent by $\User$. If no value appears in majority, a default value is chosen. 
We refer the readers to \cite{PeaseSL80} for the formal details of the agreement protocol.
On successful completion of BA, $P_0$ computes $\bv{\val} + \gv{\val}$ from $\vu$ while $P_1, P_2$ compute $\bv{\val}$ from $\vu$ locally. For the reconstruction of a value $\val$, servers send their $\shr{\cdot}$-shares of $\val$ to $\User$, who selects the majority value for each share and reconstructs the output. At any point, if a $\TTP$ is identified, the servers proceed as follows. All servers send their $\shr{\cdot}$-share of the input to the $\TTP$. $\TTP$ picks the majority value for each share and computes the function output. It then sends this output to $\User$. $\User$ also receives the identity of the $\TTP$ from  all servers and accepts the output received from the $\TTP$ forming majority.

\paragraph{Bit Extraction Protocol}
\label{app:bitext4}
This protocol enables the servers to compute a boolean sharing of the most significant bit (MSB) of a value $\val \in \Z{\ell}$ given the arithmetic sharing $\shr{\val}$. To compute the MSB, we use the optimized Parallel Prefix Adder (PPA) circuit from ABY3 \cite{MR18}, which takes as input two boolean values and outputs the MSB of the sum of the inputs.
The circuit requires $2 (\ell - 1)$ AND gates and has a multiplicative depth of $\log \ell$. The protocol for bit extraction ($\pifbitext$) involves computing the boolean PPA circuit using the protocols described in \S\ref{sec:4PC}. The two inputs to this boolean circuit are generated as follows. The value $\val$ whose MSB needs to be extracted can be represented as the sum of two values as $\val = \bv{\val} + (-\av{\val})$ where the first input to the circuit will be $\bv{\val}$ and the second input will be $-\av{\val}$. Since $\bv{\val}$ is held by $P_1, P_2$, servers execute $\pifjsh$ to generate $\shrB{\bv{\val}}$. Similarly, $P_0, P_3$ possess $\av{\val}$, and servers execute $\pifjsh$ to generate $\shrB{-\av{\val}}$. Servers in $\Partyset$ use the $\shrB{\cdot}$-shares of these two inputs ($\bv{\val}, -\av{\val}$) to compute the optimized PPA circuit which outputs the $\shrB{\MSB(\val)}$.

\begin{lemma}[Communication]
	\label{app:pifbitext}
	The protocol $\pifbitext$ requires an amortized communication of $7 \ell - 6$ bits in the preprocessing phase, and $\log \ell$ rounds with amortized communication of $7 \ell - 6$ bits in the online phase.
\end{lemma}
\begin{proof}
	Generation of boolean sharing of $\av{\val}$ requires $\ell$ bits in the preprocessing phase (since $\pifjsh$ with $P_0, P_3$ can be achieved with $\ell$ bits of communication in the preprocessing phase), and generation of boolean sharing of $\bv{\val}$ requires $\ell$ bits and one round (which can be deferred towards the end of the protocol thereby requiring one round for several instances) in the online phase. Further, the boolean PPA circuit to be computed requires $2(\ell - 1)$ AND gates. Since each AND gate requires $\pifmul$ to be executed, it requires an amortized communication of $6 \ell - 6$ bits in both the preprocessing phase and the online phase. Thus, the overall communication is $7 \ell - 6$ bits, in both, the preprocessing and online phase. The circuit has a multiplicative depth of $\log \ell$ which results in $\log \ell $ rounds in the online phase.
\end{proof}

\paragraph{Bit2A Protocol}
\label{app:bit2A4}
This protocol enables servers to compute the arithmetic sharing of a bit $\bitb$ given its boolean sharing. Let $\arval{\bitb}$ denote the value of $\bitb$ in the ring $\Z{\ell}$. We observe that $\arval{\bitb}$ can be written as follows. $\arval{\bitb} = \arval{(\av{\bitb} \xor \bv{\bitb})} = \arval{\av{\bitb}} + \arval{\bv{\bitb }} - 2 \arval{\av{\bitb}} \arval{\bv{\bitb}}$.
Thus, to obtain an arithmetic sharing of $\arval{\bitb}$, the servers can compute an arithmetic sharing of $\arval{\bv{\bitb }}$, $\arval{\av{\bitb}}$ and $\arval{\bv{\bitb }} \arval{\av{\bitb}}$. This can be done as follows. $P_0, P_3$ execute $\pifjsh$ on  $\arval{\av{\bitb }}$ in the preprocessing phase to generate $\shr{\arval{\av{\bitb }}}$. Similarly, $P_1, P_2$ execute $\pifjsh$ on  $\arval{\bv{\bitb }}$ in the online phase to generate $\shr{\arval{\bv{\bitb }}}$. This is followed by $\pifmul$ on $\shr{\arval{\bv{\bitb}}}, \shr{\arval{\av{\bitb}}}$, followed by local computation to obtain $\shr{\arval{\bitb}}$.

\smallskip
\begin{protocolbox}{$\pifbita (\Partyset, \shrB{\bitb})$}{4PC: Bit2A Protocol}{fig:p4pcbitA}
	
	\algoHead{Preprocessing}:
	\begin{myitemize}
		\item[--] Servers execute $\pifsgrsh(P_3, \arval{\ve})$ (\boxref{fig:p4pcash}) where $\ve = \av{\bitb} \xor \gv{\bitb}$. Let the shares be $\sgr{\arval{\ve}}_0 = (\ve_0, \ve_1), \sgr{\arval{\ve}}_1 = (\ve_1, \ve_2), \sgr{\arval{\ve}}_2 = (\ve_2, \ve_0), \sgr{\arval{\ve}}_3 = (\ve_0, \ve_1, \ve_2)$.
		\item[--] Verification of $\sgr{\arval{\ve}}$-sharing is performed as follows:
		\begin{inneritemize}
			\item[--] $P_1, P_2, P_3$ sample a random $\vr \in \Z{\ell}$ and a bit $\vr_{\bitb} \in \Z{1}$.
			\item[--] $P_1, P_2$ compute $\wx_1 = \gv{\bitb} \xor \vr_{\bitb}$, and $\jsendf$ $\wx_{1}$ to $P_0$.
			\item[--] $P_1, P_3$ compute $\wy_1 = (\ve_1 + \ve_2)(1 - 2 \arval{\vr_{\bitb}}) + \arval{\vr_{\bitb}} + \vr$, and $\jsendf$ $\wy_{1}$ to $P_0$.
			\item[--] $P_2, P_3$ compute $\wy_2 = \ve_0(1 - 2 \arval{\vr_{\bitb}}) - \vr$, and $\jsendf$ $\Hash(\wy_2)$ to $P_0$.
			\item[--] $P_0$ computes $\wx = \ve \xor \vr_{\bitb} = \sqrA{\av{\bitb}} \xor \sqrB{\av{\bitb}} \xor \wx_{1}$ and checks if $\Hash(\arval{\wx} - \wy_{1}) = \Hash(\wy_{2})$. 
			\item[--] If verification fails, $P_0$ sets $\flag = 1$, else it sets $\flag = 0$. $P_0$ sends $\flag$ to $P_1$. Next, $P_1, P_0$ $\jsendf$ $\flag$ to $P_2$ and $P_3$. Servers set $\TTP = P_1$ if $\flag = 1$.
		\end{inneritemize}
		\item[--] If verification succeeds, servers locally convert $\sgr{\arval{\ve}}$ to $\shr{\arval{\ve}}$ by setting their shares according to \tabref{sgrtoshr}.
	\end{myitemize}
	\justify\vspace{-3mm}
	\algoHead{Online}:
	\begin{myitemize}
		\item[--] Servers execute $\pifjsh(P_0, P_1, P_2, \arval{\vc})$ where $\vc = \bv{\bitb} \xor \gv{\bitb}$.
		\item[--] Servers execute $\pifmul(\Partyset, \shr{\arval{\ve}}, \shr{\arval{\vc}})$ to generate $\shr{\arval{\ve} \arval{\vc}}$.
		\item[--] Servers compute $\shr{\arval{\bitb}} = \shr{\arval{\ve}} + \shr{\arval{\vc}} - 2 \shr{\arval{\ve} \arval{\vc}}$.
	\end{myitemize} 
\end{protocolbox}

While the above approach serves the  purpose, we now provide an improved version, which further helps in reducing the online cost. We observe that $\arval{\bitb}$ can be written as follows. $\arval{\bitb} = \arval{(\av{\bitb} \xor \bv{\bitb})} = \arval{((\av{\bitb} \xor \gv{\bitb}) \xor (\bv{\bitb} \xor \gv{\bitb}))} = \arval{(\ve \xor \vc)} = \arval{\ve} + \arval{\vc} - 2  \arval{\ve} \arval{\vc}$ where $\ve = \av{\bitb} \xor \gv{\bitb}$ and $\vc = \bv{\bitb} \xor \gv{\bitb}$. Thus, to obtain an arithmetic sharing of $\arval{\bitb}$, $P_3$ generates $\sgr{\cdot}$-sharing of $\arval{\ve}$. To ensure the correctness of the shares, the servers $P_0, P_1, P_2$ check whether the following equation holds: $\arval{(\ve \xor \vr_{\bitb})} = \arval{\ve} + \arval{\vr_{\bitb}} - 2 \arval{\ve} \arval{\vr_{\bitb}}$. If the verification fails, a $\TTP$ is identified. Else, this is followed by servers locally converting $\sgr{\arval{\ve}}$-shares to $\shr{\arval{\ve}}$ according to \tabref{sgrtoshr}, followed by multiplying $\shr{\arval{\ve}}, \shr{\arval{\vc}}$ and locally computing $\shr{\arval{\bitb}} = \shr{\arval{\ve}} + \shr{\arval{\vc}} - 2 \shr{\arval{\ve} \arval{\vc}}$. Note that during $\pifjsh(P_0, P_1, P_2, \arval{\vc})$ since $\av{\arval{\vc}}$ and $\gv{\arval{\vc}}$ are set to $0$, the preprocessing of multiplication can be performed locally. The formal protocol appears in \boxref{fig:p4pcbitA}.

\begin{lemma}[Communication]
	\label{app:pifbita}
	$\pifbita$ (\boxref{fig:p4pcbitA}) requires an amortized communication of $3 \ell + 4$ bits in the preprocessing phase, and $1$ round with amortized communication of $3 \ell$ bits in the online phase.
\end{lemma}
\begin{proof}
	During preprocessing, one instance of $\pifsgrsh$ requires $2 \ell$ bits of communication. Further, sending $\wx_1$ requires $1$ bit, while sending $\wy_{1}$ requires $\ell$ bits. Sending of $\Hash(\wy_{2})$ can be amortized over several instances. Finally, communicating $\flag$ requires $3$ bits. Thus, the overall amortized communication cost in preprocessing phase is $3 \ell + 4$ bits.
	In the online phase, joint sharing of $\arval{\vc}$ can be performed non-interactively. The only cost is due to the online phase of multiplication which requires $3 \ell$ bits and one round. Thus, the amortized communication cost in the online phase is $3 \ell$ bits with one round of communication. 
\end{proof}

\paragraph{Bit Injection Protocol}
\label{app:bitinj4}
Given the boolean sharing of a bit $\bitb$, denoted as $\shrB{\bitb}$, and the arithmetic sharing of $\val \in \Z{\ell}$, protocol $\piBitInjF$ (\boxref{fig:p4pcbitinj}) computes $\shr{\cdot}$-sharing of $\bitb\val$. This can be naively computed by servers first executing $\pifbita$ on $\shrB{\bitb}$ to generate $\shr{\bitb}$, followed by servers computing $\shr{\bitb\val}$ by executing $\pifmul$ protocol on $\shr{\bitb}$ and $\shr{\val}$. Instead, we provide an optimized variant which helps in reducing the communication cost of the naive approach in, both, the preprocessing and online phase. We give the details next.

Let $\wz = \arval{\bitb} \val$, where $\arval{\bitb}$ denotes the value of $\bitb$ in $\Z{\ell}$. Then, during the computation of $\shr{\wz}$, we observe the following:
\begin{align*}
	\wz &= \arval{\bitb} \val = \arval{(\av{\bitb} \xor \bv{\bitb})} (\bv{\val} - \av{\val})\\ 
	&= \arval{((\av{\bitb} \xor \gv{\bitb}) \xor (\bv{\bitb} \xor \gv{\bitb}))} ((\bv{\val} + \gv{\val}) - (\av{\val} + \gv{\val}))\\	
	&= \arval{(\vc_{\bitb} \xor \ve_{\bitb})} (\vc_{\val} - \ve_{\val}) = (\arval{\vc_{\bitb}} + \arval{\ve_{\bitb}} - 2 \arval{\vc_{\bitb}} \arval{\ve_{\bitb}}) (\vc_{\val} - \ve_{\val})\\
	&= \arval{\vc_{\bitb}} \vc_{\val} - \arval{\vc_{\bitb}} \ve_{\val} + (\vc_{\val} - 2 \arval{\vc_{\bitb}} \vc_{\val}) \arval{\ve_{\bitb}} + (2 \arval{\vc_{\bitb}} - 1) \arval{\ve_{\bitb}} \ve_{\val}
\end{align*}
where $\vc_{\bitb} = \bv{\bitb} \xor \gv{\bitb}, ~\ve_{\bitb} = \av{\bitb} \xor \gv{\bitb}, ~\vc_{\val} = \bv{\val} + \gv{\val}$ and $\ve_{\val} = \av{\val} + \gv{\val}$. The protocol proceeds with $P_3$ generating $\sgr{\cdot}$-shares of $\arval{\ve_{\bitb}}$ and $\ve_{\wz} = \arval{\ve_{\bitb}} \ve_{\val}$, followed by verification of the same by $P_0, P_1, P_2$. If verification succeeds, then to enable $P_2$ to compute $\bv{\wz} = \wz + \av{\wz}$, $P_1, P_0$ $\jsendf$ the missing share of $\bv{\wz}$ to $P_2$. Similarly for $P_1$. Next, $P_1, P_2$ reconstruct $\bv{\wz}$, and $\jsendf$ $\bv{\wz} + \gv{\wz}$ to $P_0$ completing the protocol. 

\begin{lemma}[Communication]
	\label{app:PifBitInj}
	Protocol $\piBitInjF$ requires an amortized communication cost of $6 \ell + 4$ bits in the preprocessing phase, and requires $1$ round with an amortized communication of $3\ell$ bits in the online phase.
\end{lemma}
\begin{proof}
	The preprocessing phase requires two instances of $\pifsgrsh$ which require $4 \ell$ bits of communication. Verifying correctness of $\sgr{\arval{\ve_{\bitb}}}$ requires $\ell + 1$ bits, whereas for $\sgr{\ve_{\wz}}$ we require $\ell$ bits. Finally, communicating the $\flag$ requires $3$ bits. This results in the amortized communication of  $6 \ell + 4$ bits in the preprocessing phase. 
	The online phase consists of three calls to $\piJmpF$ which requires $3 \ell$ bits of amortized communication. Note that the last call can be deferred towards the end of the computation, thereby requiring a single round for multiple instances. Thus, the number of rounds required in the online phase is one. 
\end{proof}

\begin{protocolbox}{$\piBitInjF (\Partyset, \shrB{\bitb}, \shr{\val})$}{4PC: Bit Injection Protocol}{fig:p4pcbitinj}
	Let $\vc_{\bitb} = \bv{\bitb} \xor \gv{\bitb}, ~\ve_{\bitb} = \av{\bitb} \xor \gv{\bitb}, ~\vc_{\val} = \bv{\val} + \gv{\val}, ~\ve_{\val} = \av{\val} + \gv{\val}$ and $\ve_{\wz} = \arval{\ve_{\bitb}} \ve_{\val}$.
	\justify\vspace{-2mm}
	\algoHead{Preprocessing}:
	\begin{myitemize}
		\item[--] $P_0, P_3, P_j$ for $j \in \EInSet$ sample $\sqr{\av{\wz}}_1 \in \Z{\ell}$ while $P_1, P_2, P_3$ sample $\gv{\wz} \in \Z{\ell}$.
		\item[--] Servers execute $\pifsgrsh(P_3, \arval{\ve_{\bitb}})$ and $\pifsgrsh(P_3, \ve_{\wz})$. Shares of $\sgr{\ve_{\val}}$ are set locally as $\ve_{{\val}_0} = \sqrB{\av{\val}}, \ve_{{\val}_1} = \sqrA{\av{\val}}, \ve_{{\val}_3} = \gv{\val}$. 
		\item[--] Servers verify correctness of $\sgr{\arval{\ve_{\bitb}}}$ using steps similar to $\pifbita$ (\boxref{fig:p4pcbitA}). Correctness of $\sgr{\ve_{\wz}}$ is verified as follows. 
		\begin{inneritemize}
			\item[--] $P_0, P_3, P_j$ for $j \in \EInSet$ sample a random $\vr_j \in \Z{\ell}$ while $P_1, P_2, P_3$ sample a random $\vr_0 \in \Z{\ell}$. $P_0, P_3$ set $\va_0 = \vr_1 - \vr_2$, $P_1, P_3$ set $\va_1 = \vr_0 - \vr_1$ and $P_2, P_3$ set $\va_2 = \vr_2 - \vr_0$. 
			\item[--] $P_1, P_3$ compute $\vx_1  = \ve_{{\val}_2} \ve_{{\bitb}_2}+ \ve_{{\val}_1} \ve_{{\bitb}_2} + \ve_{{\val}_2} \ve_{{\bitb}_1} + \va_1$.
			\item[--] $P_2, P_3$ compute $\vx_2  = \ve_{{\val}_0} \ve_{{\bitb}_0} + \ve_{{\val}_0} \ve_{{\bitb}_2} + \ve_{{\val}_2} \ve_{{\bitb}_0} + \va_2$.
			\item[--] $P_0$ computes $\vx_0  = \ve_{{\val}_1} \ve_{{\bitb}_1}+ \ve_{{\val}_1} \ve_{{\bitb}_0} + \ve_{{\val}_0} \ve_{{\bitb}_1} + \va_0$.
			\item[--] $P_1, P_3$ $\jsendf$ $\vy_1 = \vx_1 - \ve_{{\wz}_1}$ to $P_0$, while $P_2, P_3$ $\jsendf$ $\Hash(-\vy_2)$ to $P_0$,  where $\vy_2 = \vx_2 - \ve_{\wz_2}$.  
			\item[--] $P_0$ computes $\vy_0 = \vx_0 - \ve_{\wz_0}$, and checks if $\Hash(\vy_0 + \vy_1) = \Hash(-\vy_2)$. 
			\item[--] If verification fails, $P_0$ sets $\flag = 1$, else it sets $\flag = 0$. $P_0$ sends $\flag$ to $P_1$. Next, $P_1, P_0$ $\jsendf$ $\flag$ to $P_2$ and $P_3$. Servers set $\TTP = P_1$ if $\flag = 1$.
		\end{inneritemize}
	\end{myitemize}
	\justify\vspace{-3mm}
	\algoHead{Online}:
	\begin{myitemize}
		\item[--] $P_0, P_1$ compute $\vu_1 = - \arval{\vc_{\bitb}} \ve_{{\val}_1} + (\vc_{\val} - 2 \arval{\vc_{\bitb}} \vc_{\val}) \arval{\ve_{{\bitb}_1}} + (2 \arval{\vc_{\bitb}} - 1) \ve_{{\wz}_1} + \sqrA{\av{\wz}}$, and $\jsendf$ $\vu_1$ to $P_2$.
		\item[--] $P_0, P_2$ compute $\vu_2 = - \arval{\vc_{\bitb}} \ve_{{\val}_0} + (\vc_{\val} - 2 \arval{\vc_{\bitb}} \vc_{\val}) \arval{\ve_{{\bitb}_0}} + (2 \arval{\vc_{\bitb}} - 1) \ve_{{\wz}_0} + \sqrB{\av{\wz}}$, and $\jsendf$ $\vu_2$ to $P_1$.
		\item[--] $P_1, P_2$ compute $\bv{\wz} = \vu_1 + \vu_2 - \arval{\vc_{\bitb}} \ve_{{\val}_2} + (\vc_{\val} - 2 \arval{\vc_{\bitb}} \vc_{\val}) \arval{\ve_{{\bitb}_2}} + (2 \arval{\vc_{\bitb}} - 1) \ve_{{\wz}_2} + \arval{\vc_{\bitb}} \vc_{\val}$. 
		\item[--] $P_1, P_2$ $\jsendf$ $\bv{\wz} + \gv{\wz}$ to $P_0$.
	\end{myitemize} 
	
\end{protocolbox}

\paragraph{Dot Product} 
Given $\shr{\cdot}$-shares of two $\nf$-sized vectors $\vecX, \vecY$, protocol $\pifdotp$ (\boxref{fig:p4pcdotp}) enables servers to compute $\shr{\wz}$ with $\wz = \vecX \band \vecY$. The protocol is essentially similar to $\nf$ instances of multiplications of the form $\wz_i = \wx_i \wy_i$ for $i \in [\nf]$. But instead of communicating values corresponding to each of the $\nf$ instances, servers locally sum up the shares and communicate a single value. This helps to obtain a communication cost independent of the size of the vectors.

In more detail, the dot product protocol proceeds as follows.
During the preprocessing phase, similar to the multiplication protocol $P_0, P_1, P_3$ sample a random $\sqr{\Gammaxdy}_1$. $P_0, P_3$ compute $\Gammaxdy = \sum_{i=1}^{\nf} \av{\wx_i} \av{\wy_i}$ and $\jsendf$ $\sqr{\Gammaxdy}_2 = \Gammaxdy - \sqr{\Gammaxdy}_1$ to $P_2$. 
$P_1, P_2, P_3$ sample a random $\psi$, and generate its $\sqr{\cdot}$-shares locally. Servers $P_3, P_j$ for $j \in \EInSet$ then compute $\sqr{\Chi}_j = \sum_{i=1}^{\nf} (\gv{\wx_i}\sqr{\av{\wy_i}}_j + \gv{\wy_i}\sqr{\av{\wx_i}}_j) +  \sqr{\Gammaxdy}_j - \sqr{\psi}_j$, and $\jsendf$ $\sqr{\Chi}_j$ to $P_0$.
The formal protocol is given in \boxref{fig:p4pcdotp}. 

\begin{protocolbox}{$\pifdotp (\Partyset, \{\shr{\wx_i}, \shr{\wy_i}\}_{i \in [\nf]})$}{4PC: Dot Product Protocol ($\wz = \vecX \band \vecY$)}{fig:p4pcdotp}
	\justify
	\algoHead{Preprocessing}:
	\begin{myitemize}
		\item[--] $P_0,P_3,P_j$, for $j \in \EInSet$, sample random $\sqr{\av{\wz}}_{j}\in \Z{\ell}$, while $P_0, P_1, P_3$ sample random $\sqr{\Gammaxdy}_1 \in \Z{\ell}$.
		\item[--] $P_1$, $P_2$, $P_3$ together sample random $\gv{\wz}, \psi, \vr \in \Z{\ell}$ and set $\sqrA{\psi} = \vr, ~\sqrB{\psi} = \psi - \vr$. 
		\item[--] $P_0, P_3$ compute $\sqr{\Gammaxdy}_2 = \Gammaxdy - \sqr{\Gammaxdy}_1$, where $\Gammaxdy = \sum_{i=1}^{\nf} \av{\wx_i}\av{\wy_i}$. 
		$P_0, P_3$ $\jsendf$ $\sqr{\Gammaxdy}_2$ to $P_2$.
		\item[--] $P_3, P_j$, for $j \in \EInSet$, set $\sqr{\Chi}_j = \sum_{i=1}^{\nf} (\gv{\wx_i}\sqr{\av{\wy_i}}_j +  \gv{\wy_i}\sqr{\av{\wx_i}}_j)\allowbreak + \sqr{\Gammaxdy}_j - \sqr{\psi}_j$. 
		\item[--] $P_1, P_3$ $\jsendf$ $\sqr{\Chi}_1$ to $P_0$, and $P_2, P_3$ $\jsendf$ $\sqr{\Chi}_2$ to $P_0$.
	\end{myitemize}
	\justify\vspace{-3mm}
	\algoHead{Online}:
	\begin{myitemize}
		\item[--] $P_0, P_j$, for $j \in \EInSet$, compute $\sqr{\starbeta{\wz}}_j = - \sum_{i=1}^{\nf}  ((\bv{\wx_i}+\gv{\wx_i})\sqr{\av{\wy_i}}_j + (\bv{\wy_i}+\gv{\wy_i})\sqr{\av{\wx_i}}_j) + \sqr{\av{\wz}}_j + \sqr{\Chi}_j$. 
		\item[--] $P_1, P_0$ $\jsendf$ $\sqr{\starbeta{\wz}}_1$ to $P_2$, while $P_2, P_0$ $\jsendf$ $\sqr{\starbeta{\wz}}_2$ to $P_1$.
		\item[--] $P_j$ for $j \in \EInSet$ computes $\starbeta{\wz} = \sqr{\starbeta{\wz}}_1 + \sqr{\starbeta{\wz}}_2$ and sets $\bv{\wz} = \starbeta{\wz} + \sum_{i=1}^{\nf} (\bv{\wx_i}\bv{\wy_i}) + \psi$.
		\item[--] $P_1, P_2$ $\jsendf$ $\bv{\wz}+\gv{\wz}$ to $P_0$.
	\end{myitemize}
\end{protocolbox}

\begin{lemma}[Communication]
	\label{app:pifdotp}
	$\pifdotp$ (\boxref{fig:p4pcdotp}) requires an amortized communication of $3 \ell$ bits in the preprocessing phase, and $1$ round and an amortized communication of $3 \ell$ bits in the online phase.
\end{lemma}
\begin{proof}
	The preprocessing phase requires three calls to $\piJmpF$, one to $\jsendf$ $\sqrB{\Gammaxdy}$ to $P_2$, and two to $\jsendf$ $\sqrA{\Chi}, \sqrB{\Chi}$ to $P_0$. Each invocation of $\piJmpF$ requires $\ell$ bits resulting in the amortized communication cost of preprocessing phase to be $3 \ell$ bits.
	In the online phase, there are $2$ parallel invocations of $\piJmpF$ to $\jsendf$ $\sqrA{\starbeta{\wz}}, \sqrB{\starbeta{\wz}}$ to $P_2, P_1$, respectively, which require amortized communication of $2 \ell$ bits and one round. This is followed by another call to $\piJmpF$ to $\jsendf$ $\bv{\wz} + \gv{\wz}$ to $P_0$ which requires one more round and amortized communication of $\ell$ bits. As in the multiplication protocol, this step can be delayed till the end of the protocol and clubbed for multiple instances. Thus, the online phase requires one round and an amortized communication of $3 \ell$ bits. 
\end{proof}

\paragraph{Truncation}
\label{app:trunc4}
Given the $\shrd$-sharing of a value $\val$ and a random truncation pair $(\sqr{\vr}, \shr{\vrt})$, the $\shrd$-sharing of the truncated value $\trunc{\val}$ (right shifted value by, say, $d$  positions) can be computed as follows. Servers open the value $(\val - \vr)$, truncate it and add it to $\shr{\trunc{\vr}}$ to obtain $\shr{\trunc{\val}}$. The protocol for generating the truncation pair $(\sqr{\vr}, \shr{\vrt})$ is described in \boxref{fig:p4pctrgen}.
%

\begin{mypbox}{$\piftrgen (\Partyset)$}{4PC: Generating Random Truncated Pair $(\vr, \vrt)$}{fig:p4pctrgen}
	\justify
	\begin{myitemize}
		\item[--] $P_0, P_3, P_j$, for $j \in \EInSet$ sample random $R_j \in \Z{\ell}$. $P_0, P_3$ sets $\vr = R_1 + R_2$ while $P_j$ sets $\sqr{\vr}_j = R_j$. 
		\item[--] $P_0, P_3$ locally truncate $\vr$ to obtain $\vr^d$ and execute $\pifjsh(P_0, P_3, \vr^d)$ to generate $\shr{\vr^d}$.
	\end{myitemize}	
\end{mypbox}

\begin{lemma}[Communication]
	\label{app:piftrgen}
	$\piftrgen$ (\boxref{fig:p4pctrgen}) requires $1$ round and an amortized communication of $\ell$ bits in the online phase.
\end{lemma}
\begin{proof}
	The cost follows directly from that of $\piJmpF$~(Lemma~\ref{app:piJmpF} and \ref{app:pifjsh}). 
\end{proof}


\paragraph{Dot Product with Truncation}
\label{app:dotpt4}
Protocol $\pifdotpt$~(\boxref{fig:p4pcdotpt}) enables servers to generate $\shrd$-sharing of the truncated value of $\wz = \vecX \band \vecY$, denoted as $\trunc{\wz}$, given the $\shrd$-sharing of $\nf$-sized vectors $\vecX$ and $\vecY$. This protocol is similar to the 3PC protocol.

\begin{protocolbox}{$\pifdotpt (\Partyset, \{\shr{\wx_i}, \shr{\wy_i}\}_{i \in [\nf]})$}{4PC: Dot Product Protocol with Truncation}{fig:p4pcdotpt}
	\justify
	\algoHead{Preprocessing}:
	\begin{myitemize}
		\item[--] Servers execute the preprocessing phase of $\pifdotp(\Partyset,\allowbreak \{\shr{\wx_i}, \shr{\wy_i}\}_{i \in [\nf]})$.
		\item[--] Servers execute $\piftrgen(\Partyset)$ to generate the truncation pair $(\sqr{\vr}, \shr{\vr^d})$. $P_0$ obtains the value $\vr$ in clear.
	\end{myitemize}	
	\justify\vspace{-3mm}
	\algoHead{Online}:
	\begin{myitemize} 
		\item[--] $P_0, P_j$, for $j \in \EInSet$, compute $\sqrV{\Psi}{j} = -\sum_{i=1}^{\nf} ((\bv{\wx_i}+\gv{\wx_i})\sqrV{\av{\wy_i}}{j} + (\bv{\wy_i}+\gv{\wy_i})\sqrV{\av{\wx_i}}{j})  - \sqrV{\vr}{j}$ and sets $\sqrV{{(\wz - \vr)}^\star}{j} = \sqrV{\Psi}{j} + \sqrV{\Chi}{j}$. 
		\item[--] $P_1, P_0$ $\jsendf$ $\sqrV{{(\wz - \vr)}^\star}{1}$ to $P_2$ and $P_2, P_0$ $\jsendf$ $\sqrV{{(\wz - \vr)}^\star}{2}$ to $P_1$.
		\item[--] $P_1,P_2$ locally compute $(\wz - \vr)^\star = \sqrV{{(\wz - \vr)}^\star}{1} + \sqrV{{(\wz - \vr)}^\star}{2}$ and set $(\wz - \vr) =  {(\wz - \vr)}^\star + \sum_{i=1}^{\nf}(\bv{\wx_i}\bv{\wy_i})  + \psi$. 
		\item[--] $P_1,P_2$ locally truncate $(\wz - \vr)$ to obtain $\trunc{(\wz - \vr)}$ and execute $\pifjsh(P_1, P_2, \trunc{(\wz - \vr)})$ to generate $\shr{\trunc{(\wz - \vr)}}$.
		\item[--] Servers locally compute $ \shr{\trunc{\wz}} = \shr{\trunc{(\wz - \vr)}} + \shr{\vrt}$ .
	\end{myitemize}
\end{protocolbox}

\begin{lemma}[Communication]
	\label{app:pifdotpt}
	$\pifdotpt$ (\boxref{fig:p4pcdotpt}) requires an amortized communication of $4 \ell$ bits in the preprocessing phase, and $1$ round with amortized communication of $3 \ell$ bits in the online phase.
\end{lemma}
\begin{proof}
	The preprocessing phase comprises of the preprocessing phase of $\pifdotp$ and $\piftrgen$	which results in an amortized communication of $3 \ell + \ell = 4 \ell$ bits. 
	The online phase follows from that of $\pifdotp$ protocol except that, now, $\ESet$ compute $\sqr{\cdot}$-shares of $\wz - \vr$. This requires one round and an amortized communication cost of $2 \ell$ bits. $\ESet$ then jointly share the truncated value of $\wz - \vr$ with $P_0$, which requires $1$ round and $\ell$ bits. However, this step can be deferred till the end for multiple instances, which results in amortizing this round. The total amortized communication is thus $3 \ell$ bits in online phase. 
\end{proof}

\paragraph{Activation Functions} \label{appsec:actFunc4}
Here, as in the 3PC case, we consider two activation functions --  {\em ReLU} and {\em Sig}.
\begin{lemma}[Communication]
	\label{app:pifReLU}
	Protocol for $\ReLU$ requires an amortized communication of $13 \ell - 2$ bits in the preprocessing phase and requires $\log \ell + 1$ rounds and an amortized communication of $10 \ell - 6$ bits in the online phase.
\end{lemma}
\begin{proof}
	One instance of $\ReLU$ protocol comprises of execution of one instance of $\pifbitext$, followed by $\piBitInjF$. The cost, therefore, follows from Lemma~\ref{app:pifbitext},  and Lemma~\ref{app:PifBitInj}. 
\end{proof}

\begin{lemma}[Communication]
	\label{app:pifSig}
	Protocol for $\Sig$ requires an amortized communication of $23 \ell - 1$ bits in the preprocessing phase and requires $\log \ell + 2$ rounds and an amortized communication of $20 \ell - 9$ bits in the online phase.
\end{lemma}
\begin{proof}
	An instance of $\Sig$ protocol involves the execution of the following protocols in order-- i) two parallel instances of $\pifbitext$ protocol, ii) one instance of $\pifmul$ protocol over boolean value, and iii) one instance of $\piBitInjF$ and $\pifbita$ in parallel. The cost follows from Lemma~\ref{app:pifbitext}, Lemma~\ref{app:pifbita} and Lemma~\ref{app:PifBitInj}. 
\end{proof}


%% file: Main_5_Bench.tex
In this section, we empirically show the practicality of our protocols for PPML. We consider training and inference for Logistic Regression, and inference for $3$ different Neural Networks~(NN). NN training requires additional tools to allow mixed world computations, which we leave as future work. We refer readers to SecureML~\cite{MohasselZ17}, ABY3~\cite{MR18}, BLAZE~\cite{BLAZE}, FALCON~\cite{Falcon} for a detailed description of the training and inference steps for the aforementioned ML algorithms. All our benchmarking is done over the publicly available MNIST~\cite{MNIST10} and CIFAR-10~\cite{CIFAR10} dataset. For training, we use a batch size of $B = 128$ and define $1$ KB = $8192$ bits.

In 3PC, we compare our results against the best-known framework BLAZE that provides fairness in the same setting. 
We observe that the technique of making the dot product cost independent of feature size can also be applied to BLAZE to obtain better costs. Hence, for a fair comparison, we additionally report these improved values for BLAZE. Further, we only consider the PPA circuit based variant of bit extraction for BLAZE since we aim for high throughput; the GC based variant results in huge communication and is not efficient for deep NNs. Our results imply that we get  GOD at no additional cost compared to BLAZE. For 4PC, we compare our results with two best-known works FLASH \cite{FLASH} (which is robust) and Trident \cite{Trident} (which is fair). Our results halve the cost of  FLASH and are on par with Trident.

\paragraph{Benchmarking Environment} 
We use a 64-bit ring ($\Z{64}$). The benchmarking is performed over a WAN that was instantiated using n1-standard-8 instances of Google Cloud\footnote{https://cloud.google.com/}, with machines located in East Australia ($P_0$), South Asia ($P_1$), South East Asia ($P_2$), and West Europe ($P_3$). The machines are equipped with 2.3 GHz Intel Xeon E5 v3 (Haswell) processors supporting hyper-threading, with 8 vCPUs, and 30 GB of RAM Memory and with a bandwidth of $40$ Mbps. The average round-trip time ($\rtt$) was taken as the time for communicating 1 KB of data between a pair of parties, and the $\rtt$ values were as follows.
\begin{center} 
	\resizebox{0.46\textwidth}{!}
	{
		\begin{tabular}{c c c c c c}
			\toprule
			$P_0$-$P_1$ & $P_0$-$P_2$ & $P_0$-$P_3$ & $P_1$-$P_2$ & $P_1$-$P_3$ & $P_2$-$P_3$\\
			\midrule
			$151.40 ms$ & $59.95 ms$ & $275.02 ms$ & $92.94 ms$  & $173.93 ms$  & $219.37 ms$\\
			\bottomrule 
		\end{tabular}
	}
\end{center}

\paragraph{Software Details} 
We implement our protocols using the publicly available ENCRYPTO library~\cite{ENCRYPTO} in C++17. We obtained the code of BLAZE and FLASH from the respective authors and executed them in our environment. The collision-resistant hash function was instantiated using SHA-256. We have used multi-threading, and our machines were capable of handling a total of 32 threads. Each experiment is run for 20 times, and the average values are reported.

\paragraph{Datasets} 
We use the following datasets:
\begin{myitemize}
	\item[-] MNIST~\cite{MNIST10} is a collection of $28 \times 28$ pixel, handwritten digit images along with a label between $0$ and $9$ for each image. It has $60,000$ and respectively, $10,000$ images in the training and test set. We evaluate logistic regression, and NN-1, NN-2 (cf. \S\ref{sec:nninf}) on this dataset.
	\item[-] CIFAR-10~\cite{CIFAR10} consists of $32 \times 32$ pixel images of $10$ different classes such as dogs, horses, etc. There are $50,000$ images for training and $10,000$ for testing, with $6000$ images in each class. We evaluate NN-3 (cf. \S\ref{sec:nninf}) on this dataset.
\end{myitemize}

\paragraph{Benchmarking Parameters}
We use {\em throughput} ($\TP$) as the benchmarking parameter following BLAZE and ABY3~\cite{MR18} as it would help to analyse the effect of improved communication and round complexity in a single shot. Here, $\TP$ denotes the number of operations (``iterations" for the case of training and ``queries" for the case of inference) that can be performed in unit time. We consider minute as the unit time since  most of our protocols over WAN requires more than a second to complete. An {\em iteration} in ML training consists of a {\em forward propagation} phase followed by a {\em backward propagation} phase. In the former phase, servers compute the output from the inputs. At the same time, in the latter, the model parameters are adjusted according to the difference in the computed output and the actual output. The inference can be viewed as one forward propagation of the algorithm alone. In addition to $\TP$, we provide the online and overall communication and latency for all the benchmarked ML algorithms.

We observe that due to our protocols' asymmetric nature, the  communication load is unevenly distributed among all the servers, which leaves several communication channels under-utilized. Thus, to improve the performance, we perform {\em load balancing}, where we run several parallel execution threads, each with roles of the servers changed. This helps in utilizing all channels and improving the performance.
We report the communication and runtime of the protocols for online phase and total (= preprocessing + online).

\subsection{Logistic Regression}
In Logistic Regression, one iteration comprises updating the weight vector $\vecW$ using the gradient descent algorithm (GD). It is updated according to the function given below:
	$\vecW = \vecW - \frac{\alpha}{B} \Mat{X}_i^{T}  \circ \left(\Sig(\Mat{X}_i \circ \vecW)-\Mat{Y}_i\right).$
where $\alpha$ and $\Mat{X}_i$ denote the learning rate, and a subset, of batch size B, randomly selected from the entire dataset in the $i$th iteration, respectively. The forward propagation comprises of computing the value $\Mat{X}_i \circ \vecW$ followed by an application of a sigmoid function on it. The weight vector is updated in the backward propagation, which internally requires the computation of a series of matrix multiplications, and can be achieved using a dot product. The update function can be computed  using $\shr{\cdot}$ shares as:
	$\shr{\vecW} = \shr{\vecW} - \frac{\alpha}{B} \shr{\Mat{X}_j^{T}}  \circ (\Sig(\shr{\Mat{X}_j} \circ \shr{\vecW}) - \shr{\Mat{Y}_j}).$
We summarize our results in \tabref{log_reg}. 

	\begin{table}[htb!]
		\centering
		\resizebox{\textwidth}{!}{
			\begin{tabular}{r | c | r | r | r | r | r }
				\toprule
				\multirow{2}[2]{*}{Setting} & \multirow{2}[2]{*}{Ref.} & 
				\multicolumn{3}{c|}{Online ($\TP$ in $\times 10^3$)} & 
				\multicolumn{2}{c}{Total} \\ \cmidrule{3-7}
				& & Latency (s) & Com~[KB] & $\TP$ & Latency (s) & Com~[KB] \\
				\midrule 
				\multirow{2}{*}{\makecell[r]{3PC\\Training}}     
				&  BLAZE        &  $0.74$  & $50.26$ & $4872.38$ & $0.93$ & $203.35$ \\ 
				&  {\bf SWIFT} &  $1.05$  & $50.32$ & $4872.38$ & $1.54$  & $203.47$ \\ \midrule
				\multirow{2}{*}{\makecell[r]{3PC\\Inference}}     
				&  BLAZE        & $0.66$  & $0.28$ & $7852.05$ & $0.84$ & $0.74$ \\ 
				&  {\bf SWIFT} & $0.97$  & $0.34$ & $6076.46$ & $1.46$  & $0.86$\\ \midrule
				\multirow{2}{*}{\makecell[r]{4PC\\Training}}     
				&  FLASH        & $0.83$ & $88.93$ & $5194.18$ & $1.11$ & $166.75$ \\ 
				&  {\bf SWIFT} & $0.83$ & $41.32$ & $11969.48$ & $1.11$ & $92.91$ \\ \midrule
				\multirow{2}{*}{\makecell[r]{4PC\\Inference}}     
				&  FLASH        & $0.76$ & $0.50$ & $7678.40$ & $1.04$ & $0.96$ \\ 
				&  {\bf SWIFT} & $0.75$ & $0.27$ & $15586.96$ & $1.03$ & $0.57$ \\ 
				\bottomrule
			\end{tabular}
		}
		\vspace{-2mm}
		\caption{\small Logistic Regression training and inference. $\TP$ is given in (\#it/min) for training and (\#queries/min) for inference. \label{tab:log_reg}}
	\end{table}

We observe that the online $\TP$ for the case of 3PC inference is slightly lower compared to that of BLAZE. This is because the total number of rounds for the inference phase is slightly higher in our case due to the additional rounds introduced by the verification mechanism (aka {\em verify} phase which also needs broadcast). This gap becomes less evident for protocols with more number of rounds, as is demonstrated in the case of NN (presented next), where verification for several iterations is clubbed together, making the overhead for verification insignificant. 

For the case of 4PC, our solution outperforms FLASH in terms of communication as well as throughput. Concretely, we observe a $2\times$ improvement in $\TP$ for inference and a $2.3\times$ improvement for training. For Trident~\cite{Trident}, we observe a drop of $15.86 \%$ in $\TP$ for inference due to the extra rounds required for verification to achieve GOD. This loss is, however, traded-off with the stronger security guarantee. For training, we are on par with Trident as the effect of extra rounds becomes less significant for more number of rounds, as will also be evident from the comparisons for NN inference.  

As a final remark, note that our 4PC sees roughly $2.5\times$ improvement compared to our 3PC for logistic regression. 

\subsection{NN Inference} 
\label{sec:nninf}
We consider the following popular neural networks for benchmarking. These are chosen based on the different range of model parameters and types of layers used in the network. We refer readers to \cite{Falcon} for a detailed architecture of the neural networks.

\noindent {\bf NN-1:} This is a $3$-layered fully connected network with ReLU activation after each layer. This network has around $118$K parameters and is chosen from ~\cite{MR18,BLAZE}. 

\noindent {\bf NN-2:} This network, called LeNet \cite{lenet}, contains $2$ convolutional layers and $2$ fully connected layers with ReLU activation after each layer, additionally followed by maxpool for convolutional layers. This network has approximately $431$K parameters.

\noindent {\bf NN-3:} This network, called VGG16 \cite{vgg16}, was the runner-up of ILSVRC-2014 competition. This network has $16$ layers in total and comprises of fully-connected, convolutional, ReLU activation and maxpool layers. This network has about $138$ million parameters. 

\begin{table}[htb!]
	\centering
	\resizebox{\textwidth}{!}{
		\begin{tabular}{r | c | r | r | r | r | r }
			\toprule
			\multirow{2}[2]{*}{Network} & \multirow{2}[2]{*}{Ref.} &  \multicolumn{3}{c|}{Online} & \multicolumn{2}{c}{Total}  \\ \cmidrule{3-7 }
			& & Latency (s) & Com~[MB] & $\TP$ & Latency (s) & Com~[MB]\\
			\midrule 
			\multirow{2}{*}{\makecell[r]{NN-1}}    
			&  BLAZE        & $1.92$  & $0.04$ & $49275.19$ & $2.35$ & $0.11$ \\ 
			&  {\bf SWIFT} & $2.22$ & $0.04$ & $49275.19$ & $2.97$  & $0.11$\\ \midrule
			\multirow{2}{*}{\makecell[r]{NN-2}}    
			&  BLAZE        & $4.77$  & $3.54$ & $536.52$ & $5.61$ & $9.59$ \\ 
			&  {\bf SWIFT} & $5.08$ & $3.54$ & $536.52$ & $6.22$ & $9.59$\\ 
			\midrule 
			\multirow{2}{*}{\makecell[r]{NN-3}}    
			&  BLAZE        & $15.58$ & $52.58$ & $36.03$ & $18.81$ & $148.02$\\ 
			&  {\bf SWIFT} & $15.89$ & $52.58$ & $36.03$ & $19.29$ & $148.02$\\ 
			\bottomrule
		\end{tabular}
	}
	\vspace{-2mm}
	\caption{\small 3PC NN Inference.  $\TP$ is given in (\#queries/min).\label{tab:nn_inf_3pc}}
\end{table}

\tabref{nn_inf_3pc} summarise our benchmarking results for 3PC NN inference. As illustrated, the performance of our 3PC framework is on par with BLAZE while providing better security guarantee. 

\begin{table}[htb!]
	\centering
	\resizebox{\textwidth}{!}{
		\begin{tabular}{r | c | r | r | r | r | r }
			\toprule
			\multirow{2}[2]{*}{Network} & \multirow{2}[2]{*}{Ref.} &  \multicolumn{3}{c|}{Online} & \multicolumn{2}{c}{Total}  \\ \cmidrule{3-7 }
			& & Latency (s) & Com~[MB] & $\TP$ & Latency (s) & Com~[MB]\\
			\midrule 
			\multirow{2}{*}{\makecell[r]{NN-1}}    
			&  FLASH        & $1.70$ & $0.06$ & $59130.23$ & $2.17$ & $0.12$ \\ 
			&  {\bf SWIFT} & $1.70$ & $0.03$ & $147825.56$ & $2.17$ & $0.06$ \\ \midrule
			\multirow{2}{*}{\makecell[r]{NN-2}}    
			&  FLASH        & $3.93$ & $5.51$  & $653.67$ & $4.71$ & $10.50$ \\ 
			&  {\bf SWIFT} & $3.93$ & $2.33$ & $1672.55$ & $4.71$ & $5.40$ \\ 
			\midrule 
			\multirow{2}{*}{\makecell[r]{NN-3}}    
			&  FLASH        & $12.65$ & $82.54$ & $43.61$ & $15.31$ & $157.11$\\ 
			&  {\bf SWIFT} & $12.50$ & $35.21$ & $110.47$ & $15.14$ & $81.46$\\ 
			\bottomrule
		\end{tabular}
	}
	\vspace{-2mm}
	\caption{\small 4PC NN Inference.  $\TP$ is given in (\#queries/min).\label{tab:nn_inf_4pc}}
\end{table}

\tabref{nn_inf_4pc} summarises NN inference for 4PC setting. Here, we outperform FLASH in every aspect, with the improvement in $\TP$ being at least $2.5\times$ for each NN architecture. Further, we are on par with Trident~\cite{Trident} because the extra rounds required for verification get amortized with an increase in the number of rounds required for computing NN inference. This establishes the practical relevance of our work.

As a final remark, note that our 4PC sees roughly $3\times$ improvement compared to our 3PC for NN inference. This reflects the improvements brought in by the additional honest server in the system.

%% file: Main_Conclusion.tex
In this work, we presented an efficient framework for PPML that achieves the strongest security of GOD or robustness. 
Our 3PC protocol builds upon the recent work of BLAZE~\cite{BLAZE} and achieves almost similar (in some cases, better) performance albeit improving the security guarantee. For the case of 4PC, we outperform the best-known-- (a)  robust protocol of FLASH~\cite{FLASH} by $2\times$ performance-wise and (b) fair protocol of Trident~\cite{Trident}  by uplifting its security.

We leave the problem of extending our framework to support mixed-world conversions as well as to design protocols to support algorithms like Decision Trees, k-means Clustering etc. as open problem. 

%% file: Appendix_Prelims.tex
\subsection{Shared Key Setup} \label{app:keysetup}

Let $F : \{0, 1\}^{\csec} \times \{0, 1\}^{\csec} \rightarrow X$ be a secure  pseudo-random function (PRF), with co-domain $X$ being $\Z{\ell}$. The set of keys established between the servers for 3PC is as follows:
\smallskip
\begin{myitemize}
	\item[--] One key shared between every pair-- $k_{01}, k_{02},\allowbreak k_{12}$ for the servers $(P_0, P_1), (P_0, P_2) $and$ (P_1, P_2)$, respectively.
	\item[--] One shared key known to all the servers-- $k_{\Partyset}$. 
\end{myitemize}

Suppose $P_0,P_1$ wish to sample a random value $r \in \Z{\ell}$ non-interactively. To do so they invoke $F_{k_{01}}(id_{01})$ and obtain $r$. Here, $id_{01}$ denotes a counter maintained by the servers, and is updated after every PRF invocation. The appropriate keys used to sample is implicit from the context, from the identities of the pair that sample or from the fact that it is sampled by all, and, hence, is omitted.

\begin{systembox}{$\FSETUP$}{3PC: Ideal functionality for shared-key setup}{fig:FSETUP}
	\justify
	$\FSETUP$ interacts with the servers in $\Partyset$ and the adversary $\Sim$. $\FSETUP$ picks random keys $\Key{ij}$ for $i,j \in \{0,1,2\}$ and $\Key{\Partyset}$. Let $\wy_s$ denote the keys corresponding to server $P_s$. Then
	\begin{myitemize}
		\item[--] $\wy_s = (k_{01}, k_{02}$ and $\Key{\Partyset})$ when $P_s = P_0$.
		\item[--] $\wy_s = (k_{01}, k_{12}$ and $\Key{\Partyset})$ when $P_s = P_1$.
		\item[--] $\wy_s = (k_{02}, k_{12}$ and $\Key{\Partyset})$ when $P_s = P_2$.
	\end{myitemize}
	\begin{description}
		\item {\bf Output: } Send $(\OUTPUT, \wy_s)$ to every $P_s \in \Partyset$.
		%
	\end{description}
\end{systembox}

The key setup is modelled via a functionality $\FSETUP$ (\boxref{fig:FSETUP}) that  can be realised using any secure  MPC protocol. Analogously, key setup functionality for 4PC is given in \boxref{fig:FSETUPf}.

\begin{systembox}{$\FSETUPf$}{4PC: Ideal functionality for shared-key setup}{fig:FSETUPf}
	\justify
	$\FSETUPf$ interacts with the servers in $\Partyset$ and the adversary $\Sim$. $\FSETUPf$ picks random keys $\Key{ij}$ and $\Key{ijk}$ for $i,j,k \in \{0,1,2\}$ and $\Key{\Partyset}$. Let $\wy_s$ denote the keys corresponding to server $P_s$. Then
	\begin{myitemize}
		\item[--] $\wy_s = (\Key{01}, \Key{02}, \Key{03}, \Key{012}, \Key{013}, \Key{023}$ and $\Key{\Partyset})$ when $P_s = P_0$.
		\item[--] $\wy_s = (\Key{01}, \Key{12}, \Key{13}, \Key{012},  \Key{013}, \Key{123}$ and $\Key{\Partyset})$ when $P_s = P_1$.
		\item[--] $\wy_s = (\Key{02}, \Key{12}, \Key{23}, \Key{012},  \Key{023}, \Key{123}$ and $\Key{\Partyset})$ when $P_s = P_2$.
		\item[--] $\wy_s = (\Key{03}, \Key{13}, \Key{23}, \Key{013}, \Key{023}, \Key{123}$ and $\Key{\Partyset})$ when $P_s = P_3$.
	\end{myitemize}
	\begin{description}
		\item {\bf Output: } Send $(\OUTPUT, \wy_s)$ to every $P_s \in \Partyset$.
		%
	\end{description}
\end{systembox}

To generate a {\em 3-out-of-3} additive sharing of $0$ i.e. $\sz_s$ for $s \in \{0, 1, 2\}$ such that $P_s$ holds $\sz_s$, and $\sz_0 + \sz_1 +\sz_2 = 0$, servers proceed as follows. Every pair of servers, $P_s, P_{(s+1) \% 3}$, non-interactively generate $\vr_s$, as described earlier, and each $P_s$ sets $\sz_s = \vr_s - \vr_{(s-1)\%3}$.

\subsection{Collision Resistant Hash Function}
Consider a hash function family $\Hash = \mathcal{K}\times \mathcal{L} \rightarrow \mathcal{Y}$. The hash function $\Hash$ is said to be collision resistant if, for all probabilistic polynomial-time adversaries $\Adv$, given the description of $\Hash_k$ where {$k \in_R \mathcal{K}$}, there exists a negligible function $\negl()$ such that $\Pr[ (x_1,x_2) \leftarrow \Adv(k):(x_1 \ne x_2) \wedge \Hash_k(x_1)=\Hash_k(x_2)] \leq \negl(\csec)$, where $m = \mathsf{poly}(\csec)$ and $x_1,x_2 \in_R \{0,1\}^m$.

\subsection{Commitment Scheme}
Let $\commit(x)$ denote the commitment of a value $x$. The commitment scheme $\commit(x)$ possesses two properties; {\em hiding} and {\em binding}. The former ensures privacy of the value $\val$ given just its commitment $\commit(\val)$, while the latter prevents a corrupt server from opening the commitment to a different value $x' \ne x$. The practical realization of  a commitment scheme is via a hash function $\mathcal{H}()$ given below, whose  security can be proved in the random-oracle model (ROM)--  for  $(c, o) =  (\mathcal{H}(x||r), \allowbreak x||r) = \commit (x; r)$.

%% file: Appendix_DotPPre.tex
As mentioned earlier, a trivial way to instantiate $\piDotPPre$ is to treat a dot product operation as $\nf$ multiplications. However, this results in a communication cost that is linearly dependent on the feature size. Instead, we instantiate $\piDotPPre$ by a semi-honest dot product protocol followed by a verification phase to check the correctness. For the verification phase, we extend the techniques of \cite{BonehBCGI19, BGIN19} to provide support for verification of dot product tuples. Setting the verification phase parameters appropriately gives a $\piDotPPre$ whose (amortized) communication cost is independent of the feature size. We provide the details next.

To realize $\FDotPPre$, the approach is to run a semi-honest dot product protocol followed by a verification phase to check the correctness of the output. 
For verification, the work of \cite{BonehBCGI19} provides techniques to verify the correctness of $\nm$ multiplication triples (and degree-two relations) at a cost of $\Order(\sqrt{\nm})$ extended ring elements, albeit with abort security. While \cite{BGIN19} improves their techniques to provide {\em robust} verification for multiplication, we show how to extend the techniques in \cite{BGIN19} to robustly verify the correctness of $\nm$ dot product tuples (dot product being a degree two relation), with vectors of dimension $\nf$, at a cost of $\Order(\sqrt{\nf \nm})$ extended ring elements. Thus, the cost to realize one instance of $\FDotPPre$ can be brought down to only the cost of a semi-honest dot product computation (which is $3$ ring elements and independent of the vector dimension), where the cost due to verification can be amortized away by setting $\nf, \nm$ appropriately.

Given vectors $\vecd = (\md_1, \ldots, \md_{\nf}), \vece = (\me_1, \ldots, \me_{\nf})$, let server $P_i$, for $i \in \{0,1,2\}$, hold $\sgr{\md_j}_i = (\md_{j,i}, \md_{j, (i+1)\%3})$ and $\sgr{\me_j}_i = (\me_{j,i}, \me_{j, (i+1)\%3})$ where $j \in [\nf]$ (henceforth, we omit the use of $\%3$ in the subscript as it is understood from the context). The semi-honest dot product protocol proceeds as follows. The servers, using the shared key setup, non-interactively generate \emph{3-out-of-3} additive shares of zero (as described in \ref{app:keysetup}), i.e $P_i$ has $\sz_i$, such that $\sz_0 + \sz_1 +\sz_2 = 0$. Then, each $P_i$ locally computes \emph{3-out-of -3}  additive share of $\mf = \vecd \band \vece$ as: 
\begin{align} \label{eq:dotpver1}
	\mf_i =  \sz_i + \sum_{j = 1}^{\nf} \left( \md_{j, i} \cdot \me_{j, i} + \md_{j, i} \cdot \me_{j, {i+1}} + \md_{j, {i+1}} \cdot \me_{j, i} \right)
\end{align}
Now, to complete the $\sgr{\cdot}$-sharing of $\mf$, $P_i$ sends $\mf_i$ to  $P_{i-1}$. 
To check the correctness of the computation $\sgr{\mf} = \sgr{\vecd \band \vece}$, each $P_i \in \Partyset$ needs to prove that the $\mf_i$ it sent in the semi-honest protocol satisfies \ref{eq:dotpver1}, i.e.
\begin{align} \label{eq:dotpver2}
	\sz_i + \sum_{j = 1}^{\nf} \left( \md_{j, i} \cdot \me_{j, i} + \md_{j, i} \cdot \me_{j, {i+1}} + \md_{j, {i+1}} \cdot \me_{j, i} \right) - \mf_i = 0
\end{align}
This difference in the expected message that should be sent (computed using $P_i$'s correct input shares) and actual message that is sent by $P_i$ is captured by a circuit $\cc$, defined below. 

\begin{align} 
	\begin{split} 
		& \cc \left(\{\md_{j, i}, \md_{j, {i+1}}, \me_{j, i}, \me_{j, {i+1}} \}_{j=1}^{\nf}, \sz_i, \mf_i \right) \label{eq:cc} \\ 
		&= \sz_i + \sum_{j = 1}^{\nf} \left( \md_{j, i} \cdot \me_{j, i} + \md_{j, i} \cdot \me_{j, {i+1}} + \md_{j, {i+1}} \cdot \me_{j, i} \right) - \mf_i
	\end{split}
\end{align}

Here, $\cc$ takes as input $\ic = 4 \nf + 2$ values: $\sgr{\cdot}$-shares of $\vecd, \vece$ held by $P_i$, i.e. $\{\md_{j, i}, \md_{j, {i+1}}, \me_{j, i}, \me_{j, {i+1}} \}_{j=1}^{\nf}$, the additive share of zero, $\sz_i$, that $P_i$ holds, and the additive share $\mf_i$ sent by $P_i$. For correct computation with respect to $P_i$, we require the difference in the expected message and the actual message to be $0$, i.e.,
\begin{align} \label{eq:dotpver3}
	\cc \left(\{\md_{j, i}, \md_{j, {i+1}}, \me_{j, i}, \me_{j, {i+1}} \}_{j=1}^{\nf}, \sz_i, \mf_i \right) = 0
\end{align}

We now explain how to verify the correctness for $\nm$ dot product tuples assuming that the operations are carried out over a prime-order field. The verification can be extended to support operations over rings following the techniques of \cite{BonehBCGI19, BGIN19}.
To verify the correctness for $\nm$ dot product tuples, $\{\vecd_k, \vece_k, \mf_k\}_{k=1}^{\nm}$ where $\mf_k = \vecd_k \band \vece_k$, the output of $\cc$ (which is the difference in the expected and actual message sent) for each of the corresponding dot product tuple must be $0$. 
To check correctness of all dot products {\em at once}, it suffices to check if a random linear combination of the output of each $\cc$ (for each dot product) is $0$. This is because the random linear combination of the differences will be $0$ with high probability if $\mf_k = \vecd_k \band \vece_k$ for each $k \in \{1, \ldots, \nm\}$. 
We remark that the definition of $c(\cdot)$ in \cite{BGIN19} enables the verification of only multiplication triples. With the re-definition of $\cc$ as in \ref{eq:cc}, we can now verify the correctness of dot products while the rest of the verification steps remain similar to that in \cite{BGIN19}. We elaborate on the details, next.

A verification circuit, constructed as follows, enables $P_i$ to prove the correctness of the additive share of $\mf$ that it sent, for $\nm$ instances of dot product at once. Note that the proof system is designed for the distributed-verifier setting where the proof generated by $P_i$ will be shared among $P_{i-1}, P_{i+1}$, who can together verify its correctness. First, a sub-circuit $\cg$ is defined as: group $\nL$ small $\cc$ circuits and take a random linear combination of the values on their output wires. Since each $\cc$ circuit takes  $\ic = 4\nf+2$ inputs as described earlier, $\cg$ takes in $\ic \nL$ inputs. Precisely, $\cg$ is defined as follows:
\begin{align*}
	\cg(x_1, \ldots, x_{\ic \nL}) = \sum_{k=1}^{\nL} \rt_k \cdot \cc(x_{(k-1) \ic + 1}, \ldots, x_{(k-1) \ic + \ic})
\end{align*}

Since there are total $\nm$ dot products to be verified, there will be $\nM = \nm /\nL$ sub-circuits $\cg$. Looking ahead, this grouping technique enables obtaining a sub-linear communication cost for verification because the communication cost turns out to be $\Order(\ic \nL + \nM)$ and setting $\ic \nL = \nM$ gives the desired result. The sub-circuits $\cg$ make up the circuit $\cG$ which outputs a random linear combination of the values on the output wires of each $\cg$, i.e:
\begin{align*}
	\cG(x_1, \ldots, x_{\ic \nm}) = \sum_{k=1}^{\nM} \re_k \cdot \cg(x_{(k-1) \ic \nL + 1}, \ldots, x_{(k-1) \ic \nL + \ic \nL}) 
\end{align*}
Here, $\rt_k$ and $\re_k$ are randomly sampled (non-interactively) by all parties.
To prove correctness, $P_i$ needs to prove that $\cG$ outputs $0$. For this, $P_i$ defines $f_1\ldots,f_{\ic \nL}$ random polynomials of degree $\nM$, one for each input wire of $\cg$. For $\ell \in \{1,\ldots,\nM\}$ and $j \in \{1,\ldots,\ic \nL\}$,  $f_j(0)$ is chosen randomly and $f_j(\ell) = x_{(\ell-1)\ic+j}$ (i.e the $j\text{th}$ input of the $\ell\text{th}$ $\cg$ gate). 
$P_i$ further defines a $2\nM$ degree polynomial $p(\cdot)$ on the output wires of $\cg$, i.e $p(\cdot) = \cg (f_1, \ldots, f_{\ic \nL})$ where $p(\ell)$ for $\ell \in \{1, \ldots, \nM\}$ is the output of the $\ell$th $\cg$ gate. The additional $\nM+1$ points required to interpolate the $2\nM$ degree polynomial $p$, are obtained by evaluating $f_1, \ldots, f_{\ic \nL}$ on $\nM+1$ additional points, followed by an application of $\cg$ circuit. The proof generated by $P_i$ consists of $f_1(0), \ldots, f_{\ic \nL}(0)$ and the coefficients of $p$. Recall that since we are in the distributed-verifier setting, the prover $P_i$ additively shares the proof with $P_{i-1}, P_{i+1}$. Note here, that shares of $f_1(0), \ldots, f_{\ic \nL}(0)$ can be generated non-interactively. 

To verify the proof, verifiers $P_{i-1}, P_{i+1}$ need to check if the output of $\cG$  is $0$. This can be verified by computing the output of $\cG$ as $b= \sum_{\ell=1}^{\nM} \re_{\ell} \cdot p(\ell)$ and checking if $b = 0$, where $\re_{\ell}$'s are non-interactively sampled by all after the proof is sent. If $p$ is defined correctly, then this is indeed a random linear combination of the outputs of all the $\cg$-circuits. This necessitates the second check to verify the correctness of $p$ as per its definition i.e $p(\cdot) = \cg (f_1(\cdot), \ldots, f_{\ic \nL}(\cdot))$. This is performed by checking if $p(r) = \cg (f_1(r), \ldots, f_{\ic \nL}(r))$ for a random $r \notin \{1, \ldots, \nM\}$ (for privacy to hold) sampled non-interactively by all after the proof is sent. 
For the first check, verifiers can locally compute additive shares of $b$ (using the additive shares of coefficients of $p$ obtained as part of the proof) and reconstruct $b$ to check for equality with $0$.
For the second, verifiers locally compute additive shares of $p(r)$ using the shares of coefficients of $p$, and shares of $f_1(r), \ldots, f_{\ic \nL}(r)$ by interpolating $f_1, \ldots, f_{\ic \nL}$ using ($P_i$'s) inputs to the $\cc$-circuits which are implicitly additively shared between them (owing to the replicated sharing property). Verifiers exchange these values among themselves, reconstruct it and check if $p(r) = \cg (f_1(r), \ldots, f_{\ic \nL}(r))$.
Note that, the messages computed and exchanged by the verifiers, depend only on  the proof sent by $P_i$ and the random values ($r, \re$) sampled by all. These messages can also be independently computed by $P_i$. Thus, in order to prevent a verifier from falsely rejecting a correct proof, we use $\Jmp$ to exchange these messages.
To optimize the communication cost further, it suffices if a single verifier computes the output of verification. 

\paragraph{Setting the parameters:} The proof sent by $P_i$ consists of the constant terms  $f_j(0)$ for $j \in \{1, \ldots, \ic \nL\}$ and $2\nM+1$ coefficients of $p$. The former can be can be generated non-interactively. Hence, $P_i$ needs to communicate $2\nM+1$ elements to the verifiers (one of which can be performed non-interactively). The message sent by the verifier consists of the additive share of $\sum_{\ell=1}^{\nM} \re_{\ell} \cdot p(\ell)$ (for the first check) and $f_1(r), \ldots , f_{\ic \nL}(r), p(r)$ (for the second check). Thus, the verifier communicates $\ic \nL + 2$ elements. As the proof is executed three times, each time with one party acting as the prover and the other two acting as the verifiers, overall, each party communicates $\ic \nL + 2 \nM + 3$ elements. Setting $\ic \nL = 2 \nM$ and $\nM = \frac{m}{\nL} $ results in the total communication required for verifying $m$ dot products to be $\Order(\sqrt{\nf \nm})$. Thus, verifying a single dot product has an amortized cost of $\Order \left( \sqrt{\frac{\nf}{\nm}} \right)$ which can be made very small by appropriately setting the values of $\nf, \nm$. Thus, the (amortized) cost of a maliciously secure dot product protocol can be made equal to that of a semi-honest dot product protocol, which is $3$ ring elements. 

To support verification over rings \cite{BGIN19}, verification operations are carried out on the extended ring $\Z{\ell}/f(x)$, which is the ring of all polynomials with coefficients in $\Z{\ell}$ modulo a polynomial $f$, of degree $d$, irreducible over $\Z{ }$. Each element in $\Z{\ell}$ is lifted to a $d$-degree polynomial in $\Z{\ell}[x]/f(x)$ (which results in blowing up the communication by a factor $d$). Thus, the per party communication amounts to $(\ic \nL + 2 \nM + 3)d$ elements of $\Z{\ell}$ for verifying $\nm$ dot products of vector size $\nf$ where $\ic = 4 \nf + 2$. Further, the probability of a cheating prover is bounded by $\frac{2^{(\ell - 1) d} \cdot 2 \nM + 1}{2^{\ell d} - \nM}$ (cf. Theorem 4.7 of \cite{BGIN19}). 
Thus, if $\gv{}$ is such that $2^{\gv{}} \geq 2\nM$, then the cheating probability is 
\begin{align*}
	\frac{2^{(\ell - 1) d} \cdot 2 \nM + 1}{2^{\ell d} - \nM} \leq \frac{2^{(\ell - 1) d} \cdot 2^{\gv{}} + 1}{2^{\ell d} - M} \approx 2^{-(d - \gv{})}
\end{align*}
We note that both, \cite{BGIN19} and our technique require a communication cost of $\Order(\sqrt{\nm \nf})$ ring elements for verifying $\nm$ dot products of vector size $\nf$. This is because multiplication is a special case of dot product with $\nf = 1$. However, since our verification is for dot products, we can get away with performing only $\nm$ semi-honest dot products whose cost is equivalent to computing $\nm$ semi-honest multiplications, whereas \cite{BGIN19} requires to execute $\nm \nf$ multiplications (as their technique can only verify correctness of multiplications), resulting in a dot product cost dependent on the vector size.
Concretely, to get $40$ bits of statistical security and for a vector size of $2^{10}$ (CIFAR-10\cite{CIFAR10} dataset), the aforementioned parameters can be set as given in \tabref{zkver}. 

\begin{table}[htb!]
	\centering
	\resizebox{.7\textwidth}{!}{
		\begin{tabular}{r | r | r | r | r }
			\toprule
			$\nm$ &  $\nM$ & $\gv{}$ & $d$ & Cost (per dot product) \\
			\midrule 
			$2^{20}$ & $2^{16}$ & $17$  & $57$ & $7.125$ \\ \midrule 
			$2^{30}$ & $2^{21}$ & $22$  & $62$ & $0.242$ \\ \midrule 
			$2^{40}$ & $2^{26}$ & $27$  & $67$ & $0.008$  \\ \midrule
			$2^{50}$ & $2^{31}$ & $32$  & $72$ & $0.0002$  \\ \bottomrule 
		\end{tabular}
	}
	\vspace{-2mm}
	\caption{\small Cost of verification in terms of the number of ring elements communicated per dot product, and parameters for vector size $\nf = 2^{10}$ and $40$ bits of statistical security. Here, $\nm$ - $\#$dot products to be verified, $\nM$- $\# \cg$ sub-circuits, $d$-degree of extension. \label{tab:zkver} }
\end{table}

It is possible to further bring down the communication cost required for verifying $\nm$ dot product tuples to $\Order(\log (\nf \nm))$ at the expense of requiring more rounds by further extending the technique of \cite{BonehBCGI19}, which we leave as an exercise.
We refer readers to \cite{BGIN19} for formal details. 

%% file: Appendix_Security.tex
In this section, we provide detailed security proofs for our constructions in both the 3PC and 4PC domains. We prove security using the real-world/ ideal-word simulation based technique.
We provide proofs in the $\FSETUP$-hybrid model for the case of 3PC, where $\FSETUP$~(\boxref{fig:FSETUP}) denotes the ideal functionality for the three server shared-key setup. Similarly, 4PC proofs are provided in $\FSETUPf$-hybrid model~(\boxref{fig:FSETUPf}).

Let $\Adv$ denote the real-world adversary corrupting at most one server in $\Partyset$, and $\Sim$ denote the corresponding ideal world adversary. The strategy for simulating the computation of function $f$ (represented by a circuit $\ckt$) is as follows: The simulation begins with the simulator emulating the shared-key setup~($\FSETUP / \FSETUPf$) functionality and giving the respective keys to the adversary. This is followed by the input sharing phase in which $\Sim$ extracts the input of $\Adv$, using the known keys, and sets the inputs of the honest parties to be $0$. $\Sim$ now knows all the inputs and can compute all the intermediate values for each of the building blocks in clear. Also, $\Sim$ can obtain the output of the $\ckt$ in clear. $\Sim$ now proceeds simulating each of the building block in topological order using the aforementioned values (inputs of $\Adv$, intermediate values and circuit output).

In some of our sub protocols, adversary is able to decide on which among the honest parties should be chosen as the Trusted Third Party ($\TTP$) in that execution of the protocol. To capture this, we consider {\em corruption-aware} functionalities~\cite{AsharovL17} for the sub-protocols, where the functionality is provided the identity of the corrupt server as an auxiliary information. 

For modularity, we provide the simulation steps for each of the sub-protocols separately. These steps, when carried out in the respective order, result in the simulation steps for the entire 3/4PC protocol.
If a $\TTP$ is identified during the simulation of any of the sub-protocols, simulator will stop the simulation at that step. In the next round, the simulator receives the input of the corrupt party in clear on behalf of the $\TTP$ for the 3PC case, whereas it receives the input shares from adversary for 4PC.

\subsection{Security Proofs for 3PC protocols}
\label{app:sec3pc}
\input{Appendix_3PC_Security}


\subsection{Security Proofs for 4PC protocols}
\label{app:sec4PC}
\input{Appendix_4PC_Security}


%% file: Appendix_3PC_Security.tex
The ideal functionality $\Funct$ for 3PC appears in \boxref{fig:Funct}.

\begin{systembox}{$\Funct$}{3PC: Ideal functionality for evaluating a function $f$}{fig:Funct}
	\justify
	$\Funct$ interacts with the servers in $\Partyset$ and the adversary $\Sim$. Let $f$ denote the functionality to be  computed. Let $\wx_s$ be the input corresponding to the server $P_s$, and $\wy_s$ be the corresponding output, i.e $(\wy_0,\wy_1,\wy_2) = f(\wx_0,\wx_1,\wx_2).$ 
	\begin{myitemize}
		\item[\bf Step 1:] $\Funct$ receives $(\INPUT,\wx_s)$ from $P_s \in \Partyset$, and computes $(\wy_0,\wy_1,\wy_2) = f(\wx_0,\wx_1,\wx_2).$
		\item[\bf Step 2:] $\Funct$ sends  $(\OUTPUT, \wy_s)$ to $P_s \in \Partyset$.
	\end{myitemize}
\end{systembox}

\vspace{-2mm}
\subsubsection{Joint Message Passing ($\Jmp$) Protocol} 
\label{appsec:piJmp}
This section provides the security proof for the $\Jmp$ primitive, which forms the crux for achieving GOD in our constructions. Let $\FJmp$~(\boxref{fig:JmpFunc}) denote the ideal functionality and let $\Sim_{\Jmp}^{P_s}$ denote the corresponding simulator for the case of corrupt $P_s \in \Partyset$.
We begin with the case for a corrupt sender, $P_i$. The case for a corrupt $P_j$ is similar and hence we omit details for the same. 
\begin{simulatorsplitbox}{$\Sim_{\Jmp}^{P_i}$}{Simulator $\Sim_{\Jmp}^{P_i}$ for corrupt sender $P_i$}{fig:JmpSimi}
	~
	\begin{myitemize}
		\item[--] $\Sim_{\Jmp}^{P_i}$ initializes $\ttp = \bot$ and receives $\val_i$ from $\Adv$ on behalf of $P_k$.
		\item[--] In case, $\Adv$ fails to send a value  $\Sim_{\Jmp}^{P_i}$ broadcasts \texttt{"(accuse,$P_i$)"}, sets $\ttp = P_j$, $\val_i = \bot$, and skip to the last step. 
		\item[--] Else, it checks if $\val_i = \val$, where $\val$ is the value computed by $\Sim_{\Jmp}^{P_i}$ based on the interaction with $\Adv$, and using the knowledge of the shared keys. If the values are equal, $\Sim_{\Jmp}^{P_i}$ sets $b_k = 0$, else, sets $b_k=1$, and sends the same to $\Adv$ on the behalf of $P_k$.
		\item[--] If $\Adv$ broadcasts \texttt{"(accuse,$P_k$)"}, $\Sim_{\Jmp}^{P_i}$ sets $\val_i = \bot$, $\ttp = P_j$, and skips to the last step.
		\item[--] $\Sim_{\Jmp}^{P_i}$ computes and sends $b_j$ to $\Adv$ on behalf of $P_j$ and receives $b_{\Adv}$ from $\Adv$ on behalf of honest $P_j$.
		\item[--] If $\Sim_{\Jmp}^{P_i}$ does not receive a $b_{\Adv}$ on behalf of $P_j$, it broadcasts \texttt{"(accuse,$P_i$)"}, sets $\val_i = \bot$, $\ttp = P_k$. If $\Adv$ broadcasts \texttt{"(accuse,$P_j$)"}, $\Sim_{\Jmp}^{P_i}$ sets $\val_i = \bot$, $\ttp = P_k$. If $\ttp$ is set, skip to the last step.
		\item[--] If $(\val_i = \val)$ and $b_{\Adv} = 1$, $\Sim_{\Jmp}^{P_i}$ broadcasts $\Hash_j = \Hash(\val)$ on behalf of $P_j$. 
		\item[--] Else if $\val_i \neq \val_j$ :  $\Sim_{\Jmp}^{P_i}$ broadcasts $\Hash_j = \Hash(\val)$ and $\Hash_k = \Hash(\val_i)$ on behalf of $P_j$ and $P_k$, respectively. If $\Adv$ does not broadcast,  $\Sim_{\Jmp}^{P_i}$ sets $\ttp = P_k$.  Else if, $\Adv$ broadcasts a value $\Hash_{\Adv}$:
		\begin{inneritemize}
			\item If $\Hash_{\Adv} \neq \Hash_j$ : $\Sim_{\Jmp}^{P_i}$ sets $\ttp = P_k$.
			\item Else if $\Hash_{\Adv} \neq \Hash_k$ : $\Sim_{\Jmp}^{P_i}$ sets $\ttp = P_j$.
		\end{inneritemize} 
		\item[--] $\Sim_{\Jmp}^{P_i}$ invokes $\FJmp$ on $(\INPUT,\val_i)$ and $(\SELECT, \ttp)$ on behalf of $\Adv$.
	\end{myitemize}
\end{simulatorsplitbox}

The case for a corrupt receiver, $P_k$ is provided in \boxref{fig:JmpSimk}.

\begin{simulatorbox}{$\Sim_{\Jmp}^{P_k}$}{Simulator $\Sim_{\Jmp}^{P_k}$ for corrupt receiver $P_k$}{fig:JmpSimk}
	~
	\begin{myitemize}
		\item[--] $\Sim_{\Jmp}^{P_k}$ initializes $\ttp = \bot$, computes $\val$ honestly and  sends $\val$ and $\Hash(\val)$ to $\Adv$ on behalf of $P_i$ and $P_j$, respectively.  
		\item[--] If $\Adv$ broadcasts \texttt{"(accuse,$P_i$)"}, set $\ttp = P_j$, else if $\Adv$ broadcasts \texttt{"(accuse,$P_j$)"}, set $\ttp = P_i$. If both messages are broadcast, set $\ttp = P_i$. If $\ttp$ is set skip to the last step.
		\item[--] On behalf of $P_i,P_j$, $\Sim_{\Jmp}^{P_k}$ receives $b_{\Adv}$ from $\Adv$. Let $b_{i}$ (resp. $b_{j}$) denote the bit received by $P_i$ (resp. $P_j$) from $\Adv$. 
		\item[--] If $\Adv$ failed to send bit $b_{\Adv}$ to $P_i$, $\Sim_{\Jmp}^{P_k}$ broadcasts \texttt{"(accuse,$P_k$)"}, set $\ttp = P_j$. Similarly, for $P_j$. If both $P_i, P_j$ broadcast \texttt{"(accuse,$P_k$)"}, set $\ttp = P_i$. If $\ttp$ is set, skip to the last step. 
		\item[--] If $b_{i} \vee b_{j} = 1$ : $\Sim_{\Jmp}^{P_k}$ broadcasts $\Hash_i, \Hash_j$ where $\Hash_i = \Hash_j = \Hash(\val)$ on behalf of $P_i,P_j$, respectively. 
		\item[--] If $\Adv$ does not broadcast $\Sim_{\Jmp}^{P_k}$ sets $\ttp = \bot$. If $\Adv$ broadcasts a value $\Hash_{\Adv}$:
		\begin{inneritemize}
			\item If $\Hash_{\Adv} \neq \Hash_i$ : $\Sim_{\Jmp}^{P_k}$ sets $\ttp = P_j$.
			\item Else if $\Hash_{\Adv} = \Hash_i = \Hash_j$ : $\Sim_{\Jmp}^{P_k}$  sets $\ttp = P_i$. 
		\end{inneritemize}
		\item [--] $\Sim_{\Jmp}^{P_k}$ invokes $\FJmp$ on $(\INPUT,\bot)$ and $(\SELECT, \ttp)$ on behalf of $\Adv$.
	\end{myitemize}
\end{simulatorbox}

\vspace{-5mm}
\subsubsection{Sharing Protocol} 
\label{appsec:FpiSh}
The case for a corrupt $P_0$ is provided in \boxref{fig:JShSim0}.
\begin{simulatorsplitbox}{$\Sim_{\Sh}^{P_0}$}{Simulator $\Sim_{\Sh}^{P_0}$ for corrupt  $P_0$}{fig:JShSim0}
	\justify
	\algoHead{Preprocessing:} 
	$\Sim_{\Sh}^{P_0}$ emulates $\FSETUP$ and gives the keys $(k_{01},k_{02},k_{\Partyset})$ to $\Adv$. The values that are commonly held along with $\Adv$ are sampled using appropriate shared key. Otherwise, values are sampled randomly.
   	\justify\vspace{-3mm}	
	\algoHead{Online:}
	\begin{myitemize}
		\item[--] If the dealer $P_s = P_0$:
		\begin{inneritemize}
			\item $\Sim_{\Sh}^{P_0}$ receives $\bv{\val}$ on behalf of $P_1$ and sets $\msg = \val$ accordingly. 
			\item Steps for $\piJmp$ protocol are simulated according to $\Sim_{\Jmp}^{P_i}$ (\figref{JmpSimi}), where $P_0$ plays the role of one of the senders.	
		\end{inneritemize}
		\item[--] If the dealer $P_s = P_1$:
		\begin{inneritemize}
			\item $\Sim_{\Sh}^{P_0}$ sets $\val = 0$ by assigning $\bv{\val} = \av{\val}$.
			\item Steps for $\piJmp$ protocol are simulated similar to $\Sim_{\Jmp}^{P_k}$ (\figref{JmpSimk}), with $P_0$ acting as the receiver. 
		\end{inneritemize}
		\item[--] If the dealer if $P_2$ : Similar to the case when $P_s = P_1$.
	\end{myitemize}
\end{simulatorsplitbox}

\smallskip
The case for a corrupt $P_1$ is provided in \boxref{fig:JShSim1}. The case for a corrupt $P_2$ is similar. 
\begin{simulatorbox}{$\Sim_{\Sh}^{P_1}$}{Simulator $\Sim_{\Sh}^{P_1}$ for corrupt $P_1$}{fig:JShSim1}
	\justify
	\algoHead{Preprocessing:} 
	$\Sim_{\JSh}^{P_1}$ emulates $\FSETUP$ and gives the keys $(k_{01},k_{12},k_{\Partyset})$ to $\Adv$. The values that are commonly held along with $\Adv$ are sampled using appropriate shared key. Otherwise, values are sampled randomly.
	\justify\vspace{-3mm}
	\algoHead{Online:}
	\begin{myitemize}
		\item[--] If dealer $P_s = P_1$ : $\Sim_{\Sh}^{P_1}$  receives $\bv{\val}$ from $\Adv$ on behalf of $P_2$.
		\item[--] If $P_s = P_0$ :  $\Sim_{\Sh}^{P_1}$ sets $\val = 0$ by assigning $\bv{\val} = \av{\val}$ and sends $\bv{\val}$ to $\Adv$ on behalf of $P_s$. 
		\item[--] If $P_s = P_2$ : Similar to the case where $P_s = P_0$. 
		\item[--] Steps of $\piJmp$, in all the steps above, are simulated similar to $\Sim_{\Jmp}^{P_i}$ (\figref{JmpSimi}), ie. the case of corrupt sender. 
	\end{myitemize}
\end{simulatorbox}

\subsubsection{Multiplication Protocol} 
\label{appsec:FpiMult}
The case for a corrupt $P_0$ is provided in \boxref{fig:MultSim0}.

\begin{simulatorbox}{$\Sim_{\Mult}^{P_0}$}{Simulator $\Sim_{\Mult}^{P_0}$ for corrupt $P_0$}{fig:MultSim0}
	\justify
	\algoHead{Preprocessing:} 
	\begin{myitemize}
		\item[--] $\Sim_{\Mult}^{P_0}$ samples $\sqr{\av{\vz}}_1, \sqr{\av{\vz}}_2$ and $\gv{\vz}$ on behalf of $P_1,P_2$ and generates the $\sgr{\cdot}$-shares of $\md,\me$ honestly.
		\item[--] $\Sim_{\Mult}^{P_0}$ emulates $\FMulPre$, and extracts $\psi$, $\sqr{\Chi}_1,\sqr{\Chi}_2$ on behalf of $P_1,P_2$.
	\end{myitemize}	
    \justify\vspace{-3mm}
	\algoHead{Online:}
	\begin{myitemize}
		\item[--] $\Sim_{\Mult}^{P_0}$ computes $\sqr{\starbeta{\vz}}_1, \sqr{\starbeta{\vz}}_2$ and steps of $\piJmp$ are simulated according to $\Sim_{\Jmp}^{P_i}$ with $\Adv$ as one of the sender for both $\sqr{\starbeta{\vz}}_1$, and $\sqr{\starbeta{\vz}}_2$. 
		\item[--] $\Sim_{\Mult}^{P_0}$ computes $\bv{\vz} + \gv{\vz}$ on behalf of $P_1,P_2$ and  steps of $\piJmp$ are simulated according to $\Sim_{\Jmp}^{P_k}$ with $\Adv$ as the receiver for $\bv{\vz} + \gv{\vz}$.  

	\end{myitemize}
\end{simulatorbox}

\smallskip
The case for a corrupt $P_1$ is provided in \boxref{fig:MultSim1}. The case for a corrupt $P_2$ is similar. 

\begin{simulatorsplitbox}{$\Sim_{\Mult}^{P_1}$}{Simulator $\Sim_{\Mult}^{P_1}$ for corrupt $P_1$}{fig:MultSim1}
	\justify
	\algoHead{Preprocessing:} 
	\begin{myitemize}
		\item[--] $\Sim_{\Mult}^{P_1}$ samples $\sqr{\av{\vz}}_1, \gv{\vz}$ and $\sqr{\av{\vz}}_2$ on behalf of $P_0,P_2$. $\Sim_{\Mult}^{P_1}$ generates the $\sgr{\cdot}$-shares of $\md,\me$ honestly.
		\item[--] $\Sim_{\Mult}^{P_1}$ emulates $\FMulPre$, and extracts $\psi, \sqr{\Chi}_1, \sqr{\Chi}_2$ on behalf of $P_0, P_2$.
	\end{myitemize}
	\justify\vspace{-3mm}
	\algoHead{Online:}
	\begin{myitemize}
		\item[--] $\Sim_{\Mult}^{P_1}$ computes $\sqr{\starbeta{\vz}}_1, \sqr{\starbeta{\vz}}_2$ on behalf of $P_0,P_2$, and steps of $\piJmp$ are simulated according to $\Sim_{\Jmp}^{P_i}$ with $\Adv$ as one of the  sender for $\sqr{\starbeta{\vz}}_1$, and as the receiver for $\sqr{\starbeta{\vz}}_2$.  
		\item[--] $\Sim_{\Mult}^{P_1}$ computes $\bv{\vz} + \gv{\vz}$ on behalf of $P_2$ and  steps of $\piJmp$ are simulated according to $\Sim_{\Jmp}^{P_i}$ with $\Adv$ one of the senders for $\bv{\vz} + \gv{\vz}$.  
	
	\end{myitemize}
\end{simulatorsplitbox}

\subsubsection{Reconstruction Protocol}
The case for a corrupt $P_0$ is provided in \boxref{fig:ReconFSim0}. The case for a corrupt $P_1,P_2$ is similar. 

\begin{simulatorbox}{$\Sim_{\Rec}$}{Simulator $\Sim_{\Rec}$ for corrupt $P_0$ }{fig:ReconFSim0}
	\justify
	\algoHead{Preprocessing:} 
	\begin{myitemize}
		\item[--] $\Sim_{\Rec}$ computes commitments on $\sqr{\av{\val}}_1, \sqr{\av{\val}}_2$ and $\gv{\val}$ on behalf of $P_1,P_2$, using the respective shared keys.
		\item[--]The steps of $\piJmp$ are simulated similar to $\Sim_{\Jmp}^{P_k}$ with $\Adv$ acting as the receiver for $\Commit{\gv{\val}}$, and $\Sim_{\Jmp}^{P_i}$ with $\Adv$ acting as one of the senders for $\Commit{\sqr{\av{\val}}_1}$ and $ \Commit{\sqr{\av{\val}}_2}$.
	\end{myitemize}	
    \justify\vspace{-3mm}
	\algoHead{Online:}
	\begin{myitemize}
		\item[--] $\Sim_{\Rec}$ receives openings for $\Commit{\sqr{\av{\val}}_1}, \Commit{\sqr{\av{\val}}_2} $ on behalf of $P_2$ and $P_1$, respectively. 
		\item[--] $\Sim_{\Rec}$ opens $\Commit{\gv{\val}}$ to $\Adv$ on behalf of $P_1,P_2$.

	\end{myitemize}
\end{simulatorbox}

\subsubsection{Joint Sharing Protocol}
The case for a corrupt $P_0$ is provided in \boxref{fig:JShSim0}. The case for a corrupt $P_1,P_2$ is similar. 
\begin{simulatorbox}{$\Sim_{\JSh}$}{Simulator $\Sim_{\JSh}$ for corrupt $P_0$ }{fig:JShSim0}
	\justify
	\algoHead{Preprocessing:} 	
	$\Sim_{\Sh}^{P_0}$ emulates $\FSETUP$ and gives the keys $(k_{01},k_{02},k_{\Partyset})$ to $\Adv$. The values that are commonly held along with $\Adv$ are sampled using appropriate shared key. Otherwise, values are sampled randomly.
	\justify\vspace{-3mm}
	\algoHead{Online:}
	\begin{myitemize}
		\item[--] If $(P_i,P_j)= (P_1,P_0)$ : $\Sim_{\JSh}$ computes $\bv{\val} = \val +\av{\val}$ on behalf of $P_1$. The steps of $\piJmp$ are simulated similar to $\Sim_{\Jmp}^{P_i}$, where the $\Adv$ acts as one of the senders. 	
		\item[--] If  $(P_i,P_j)= (P_2,P_0)$ : Similar to the case when  $(P_i,P_j)= (P_1,P_0)$. 
		\item[--] If  $(P_i,P_j)= (P_1,P_2)$ : $\Sim_{\JSh}$ sets $\val =0$ by setting $\bv{\val} = \av{\val}$. The steps of $\piJmp$ are simulated similar to $\Sim_{\Jmp}^{P_k}$, where the $\Adv$ acts as the receiver.	
	\end{myitemize}
\end{simulatorbox}

\subsubsection{Dot Product Protocol}
The case for a corrupt $P_0$ is provided in \boxref{fig:DotPSim0}.
\begin{simulatorbox}{$\Sim_{\DotP}^{P_0}$}{Simulator $\Sim_{\DotP}$ for corrupt $P_0$ }{fig:DotPSim0}
	\justify
	\algoHead{Preprocessing:} $\Sim_{\DotP}$ emulates $\FDotPPre$ and derives $\psi$ and respective $\sqr{\cdot}$-shares of $\Chi$  honestly on behalf of $P_1,P_2$. 
	\justify\vspace{-3mm}
	\algoHead{Online:} 
	\begin{myitemize}
		\item[--] $\Sim_{\DotP}^{P_0}$ computes $\sqr{\cdot}$-shares of $\starbeta{\wz}$ on behalf of $P_1,P_2$. The steps of $\piJmp$, required to provide $P_1$ with $\sqr{\starbeta{\wz}}_2$, and $P_2$ with $\sqr{\starbeta{\wz}}_1$, are simulated similar to $\Sim_{\Jmp}^{P_i}$, where $P_0$ acts as one of the sender in both cases.   
		\item[--] $\Sim_{\DotP}^{P_0}$ computes $\starbeta{\wz}$ and $ \bv{\wz}$ on behalf of $P_1,P_2$. The steps of the $\piJmp$, required to provide $P_0$ with $\bv{\wz}+\gv{\wz}$, are simulated similar to $\Sim_{\Jmp}^{P_k}$, where $P_0$ acts as the receiver.   
		
	\end{myitemize}
\end{simulatorbox}

\smallskip
The case for a corrupt $P_1$ is provided in \boxref{fig:DotPSim1}. The case for a corrupt $P_2$ is similar. 
\begin{simulatorbox}{$\Sim_{\DotP}^{P_1}$}{Simulator $\Sim_{\DotP}$ for corrupt $P_1$ }{fig:DotPSim1}
	\justify
	\algoHead{Preprocessing:}$\Sim_{\DotP}^{P_1}$ emulates $\FDotPPre$ and derives $\sqr{\cdot}$-shares of $\psi,\Chi$  honestly on behalf of $P_0,P_2$. 
	\justify\vspace{-3mm}
	\algoHead{Online:} 
	\begin{myitemize}
		\item[--] $\Sim_{\DotP}^{P_1}$ computes $\sqr{\starbeta{\wz}}_1$ on behalf of $P_0$, and $\sqr{\starbeta{\wz}}_2$ on behalf of $P_0$ and $P_2$. The steps of $\piJmp$, required to provide $P_1$ with $\sqr{\starbeta{\wz}}_2$, and $P_2$ with $\sqr{\starbeta{\wz}}_1$, are simulated similar to $\Sim_{\Jmp}^{P_i}$, where $\Adv$ acts as one of the sender in the former case, and as a receiver in the latter case.   
		\item[--] $\Sim_{\DotP}^{P_1}$ computes $\starbeta{\wz}$ and $ \bv{\wz}$ on behalf of $P_2$. The steps of $\piJmp$, required to provide $P_0$ with $\bv{\wz}+\gv{\wz}$, are simulated similar to $\Sim_{\Jmp}^{P_i}$, where $\Adv$ acts as one of the sender.   
	\end{myitemize}
\end{simulatorbox}
\subsubsection{ Truncation Protocol}
The case for a corrupt $P_0$ is provided in \boxref{fig:TruncSim0}.
\begin{simulatorbox}{$\Sim_{\Trunc}^{P_0}$}{Simulator $\Sim_{\Trunc}^{P_0}$ for corrupt $P_0$ }{fig:TruncSim0}
	\justify
		~
	\begin{myitemize}
		\item[--] For $i \in \{0, \ldots, \ell-1\}$, for $j \in \{1,2\}$, $\Sim_{\Trunc}^{P_0}$ samples $\vr_j[i]$  on behalf of $P_j$ along with $\Adv$ using respective keys. 
		\item[--] $\Sim_{\Trunc}^{P_0}$ acting on behalf of $P_1,P_2$ generates $\shrd$-shares of $\arval{(\vr_j[i])}$ for $i \in \{0, \ldots, \ell-1\}, j \in  \EInSet$ non-interactively.
		\item[--] $\Sim_{\Trunc}^{P_0}$ defines $\vecX$, $\vecY$, $\vecP$ and $\vecQ$ as per \boxref{fig:piTr}. The steps for $\piDotP$ are simulated similar to   $\Sim_{\piDotP}^{P_0}$ for generating $\sA, \sB$.
	\end{myitemize}
\end{simulatorbox}

\smallskip
The case for a corrupt $P_1$ is provided in \boxref{fig:TruncSim1}. The case for a corrupt $P_2$ is similar. 
\begin{simulatorsplitbox}{$\Sim_{\Trunc}^{P_1}$}{Simulator $\Sim_{\Trunc}^{P_1}$ for corrupt $P_1$ }{fig:TruncSim1}
	\justify
	~
	\begin{myitemize}
	\item[--] For $i \in \{0, \ldots, \ell-1\}$, $\Sim_{\Trunc}^{P_1}$ samples $\vr_1[i]$  on behalf of $P_0$ along with $\Adv$, using respective keys, and it samples $\vr_2[i]$ randomly on behalf of $P_0,P_2$. 
	\item[--] $\Sim_{\Trunc}^{P_1}$ acting on behalf of $P_0,P_2$ generates $\shrd$-shares of $\arval{(r_j[i])}$ for $i \in \{0, \ldots, \ell-1\}, j \in  \EInSet$ 
	non-interactively.
	\item[--] $\Sim_{\Trunc}^{P_1}$ defines $\vecX$, $\vecY$, $\vecP$ and $\vecQ$ as per \boxref{fig:piTr}. The steps for $\piDotP$ are simulated similar to $\Sim_{\piDotP}^{P_1}$ for generating $\sA, \sB$.
\end{myitemize}
\end{simulatorsplitbox}

%% file: Appendix_4PC_Security.tex
The ideal functionality $\Funcf$ for evaluating a function $f$ to be computed by $\ckt$ in 4PC appears in \boxref{fig:Funcf}.
\begin{systembox}{$\Funcf$}{4PC: Ideal functionality for computing $f$ in 4PC setting}{fig:Funcf}
	\justify
	$\Funcf$ interacts with the servers in $\Partyset$ and the adversary $\Sim$. Let $f$ denote the function to be computed. Let $\wx_s$ be the input corresponding to the server $P_s$, and $\wy_s$ be the corresponding output, i.e $(\wy_0,\wy_1,\wy_2,\wy_3) = f(\wx_0,\wx_1,\wx_2,\wx_3).$ 
	\begin{myitemize}
		\item[\bf Step 1:] $\Funcf$ receives $(\INPUT,\wx_s)$ from $P_s \in \Partyset$, and computes $(\wy_0,\wy_1,\wy_3,\wy_3) = f(\wx_0,\wx_1,\wx_2,\wx_3).$
		\item[\bf Step 2:] $\Funcf$ sends  $(\OUTPUT, \wy_s)$ to $P_s \in \Partyset$.
	\end{myitemize}
\end{systembox}


\subsubsection{Joint Message Passing}
\label{appsec:piJmpF}
Let $\FJmpF$ \boxref{fig:JmpFFunc} denote the ideal functionality and let $\Sim_{\JmpF}^{P_s}$ denote the corresponding simulator for the case of corrupt $P_s \in \Partyset$. 

We begin with the case for a corrupt sender, $P_i$, which is provided in \boxref{fig:JmpFSims}. The case for a corrupt $P_j$ is similar and hence we omit details for the same. 
\begin{simulatorbox}{$\Sim_{\JmpF}^{P_i}$}{Simulator $\Sim_{\JmpF}^{P_i}$ for corrupt sender $P_i$}{fig:JmpFSims}
	\medskip
	\justify
	\begin{myitemize}
		\item[--] $\Sim_{\JmpF}^{P_i}$ receives $\val_i$ from $\Adv$ on behalf of honest $P_k$. If $\val_i = \val_j$ (where $\val_j$ is the value computed by $\Sim_{\JmpF}^{P_i}$ based on the interaction with $\Adv$, and using the knowledge of the shared keys), then $\Sim_{\JmpF}^{P_i}$ sets $b_k = 0$, else it sets $b_k = 1$. If $\Adv$ fails to send a value, $b_k$ is set to be $1$.  $\Sim_{\JmpF}^{P_i}$ sends $b_k$ to $\Adv$ on behalf of $P_k$.  
		\item[--] $\Sim_{\JmpF}^{P_i}$ sends $b_l = b_k$ and $b_j = b_k$ to $\Adv$, and receives $b_i$ from $\Adv$ on behalf of honest $P_l, P_j$, respectively. 
		\item[--] If $b_k = 1$, $\Sim_{\JmpF}^{P_i}$ sets 
		$\ttp = 1$, else it sets 
		$\ttp = 0$.  
		$\Sim_{\JmpF}^{P_i}$ invokes $\FJmpF$ with $(\INPUT, \val_i)$ and $(\SELECT, b_k)$ on behalf of $\Adv$.
	\end{myitemize}
\end{simulatorbox}

\smallskip
The case for a corrupt receiver, $P_k$ is provided in \boxref{fig:JmpFSimr}.

\begin{simulatorsplitbox}{$\Sim_{\JmpF}^{P_k}$}{Simulator $\Sim_{\JmpF}^{P_k}$ for corrupt receiver $P_k$}{fig:JmpFSimr}
	\medskip
	\justify 
	\begin{myitemize}
		\item[--] $\Sim_{\JmpF}^{P_k}$ sends $\val$, $\Hash(\val)$ (where $\val$ is the value computed by $\Sim_{\JmpF}^{P_k}$ based on the interaction with $\Adv$, and using the knowledge of the shared keys) to $\Adv$ on behalf of honest $P_i, P_j$, respectively.
		\item[--] $\Sim_{\JmpF}^{P_k}$ receives $b_{ki}, b_{kj}, b_{kl}$ from $\Adv$ on behalf of $P_i, P_j, P_l$, respectively. If $\Adv$ fails to send a value, it is assumed to be $1$.
		\item[--] $\Sim_{\JmpF}^{P_k}$ sets $b_k$ to be majority value in $b_{ki}, b_{kj}, b_{kl}$. If $b_k = 1$, $\Sim_{\JmpF}^{P_k}$ sets $\ttp = 1$, else it sets $\ttp = 0$.  $\Sim_{\JmpF}^{P_k}$ invokes $\FJmpF$ with $(\INPUT, \bot)$ and $(\SELECT, b_k)$ on behalf of $\Adv$.
	\end{myitemize}
\end{simulatorsplitbox}

\smallskip
The case for a corrupt receiver, $P_l$, which is the server outside the computation involved in $\piJmpF$, is provided in \boxref{fig:JmpFSimt}.

\begin{simulatorbox}{$\Sim_{\JmpF}^{P_l}$}{Simulator $\Sim_{\JmpF}^{P_l}$ for corrupt fourth server $P_l$}{fig:JmpFSimt}
	\medskip
	\justify
	\begin{myitemize}
		\item[--] $\Sim_{\JmpF}^{P_l}$ sends $b_k = 0$ followed by $b_i=0, b_j = 0$ to $\Adv$ on behalf of $P_k$ and $P_i, P_j$, respectively.
		\item[--] $\Sim_{\JmpF}^{P_l}$ invokes $\FJmpF$ with $(\INPUT, \bot)$ and $(\SELECT, b_k)$ on behalf of $\Adv$.
	\end{myitemize}
\end{simulatorbox}

\subsubsection{Sharing Protocol}
\label{appsec:Fpifinput}
The case for corrupt $P_0$ is given in \boxref{fig:ShFSim0}.
\begin{simulatorbox}{$\Sim_{\pifinput}^{P_0}$}{Simulator $\Sim_{\pifinput}^{P_0}$ for corrupt $P_0$ }{fig:ShFSim0}
	\justify 
	\algoHead{Preprocessing:}
	$\Sim_{\pifinput}^{P_0}$ emulates $\FSETUPf$ and gives the keys $(\Key{01}, \Key{02}, \Key{03}, \Key{012}, \Key{013}, \Key{023}$ and $\Key{\Partyset})$ to $\Adv$. 
	The values that are commonly held with $\Adv$ are sampled using the respective keys, while others are sampled randomly.
	\justify\vspace{-3mm}
	\algoHead{Online:}
	\begin{myitemize}
		\item[--] If dealer is $P_0$, $\Sim_{\pifinput}^{P_0}$ receives $\bv{\val}$ from $\Adv$ on behalf of $P_1$. Steps corresponding to $\piJmpF$ are simulated according to $\Sim_{\piJmpF}^{P_i}$ where $P_0$ acts as one of the sender for sending $\bv{\val}$.
		\item[--] If dealer is $P_1$ or $P_2$, $\Sim_{\pifinput}^{P_0}$ sets $\val = 0$ by assigning $\bv{\val} = \av{\val}$. Steps corresponding to $\piJmpF$ for sending $\bv{\val} + \gv{\val}$ to $\Adv$ are simulated according to $\Sim_{\piJmpF}^{P_k}$ where $P_0$ acts as the receiver.  
		\item[--] If dealer is $P_3$, $\Sim_{\pifinput}^{P_0}$ sets $\val = 0$ by assigning $\bv{\val} = \av{\val}$. $\Sim_{\pifinput}^{P_0}$ sends $\bv{\val} + \gv{\val}$ to $\Adv$ on behalf of $P_3$.
		Steps corresponding to $\piJmpF$ are simulated according to $\Sim_{\piJmpF}^{P_j}$ where $P_0$ acts as one of the sender with $P_1$, $P_2$ as the receivers, separately.
	\end{myitemize}
\end{simulatorbox}

The case for corrupt $P_1$ is given in \boxref{fig:ShFSim1}. The case for a corrupt $P_2$ is similar. 
\begin{simulatorsplitbox}{$\Sim_{\pifinput}^{P_1}$}{Simulator $\Sim_{\pifinput}^{P_1}$ for corrupt $P_1$ }{fig:ShFSim1}
	\justify 
	\algoHead{Preprocessing:}
	$\Sim_{\pifinput}^{P_1}$ emulates $\FSETUPf$ and gives the keys $(\Key{01}, \Key{12}, \Key{13}, \Key{012},  \Key{013}, \Key{123}$ and $\Key{\Partyset})$ to $\Adv$. 
	The values that are commonly held with $\Adv$ are sampled using the respective keys, while others are sampled randomly.
	\justify\vspace{-3mm}
	\algoHead{Online:}
	\begin{myitemize}
		\item[--] If dealer is $P_1$, $\Sim_{\pifinput}^{P_1}$ receives $\bv{\val}$ from $\Adv$ on behalf of $P_2$. 
		Steps corresponding to $\piJmpF$ are simulated according to $\Sim_{\piJmpF}^{P_i}$ where $P_1$ acts as one of the sender for sending $\bv{\val}+\gv{\val}$ to $P_0$.
		
		\item[--] If dealer is $P_0$ or $P_2$, $\Sim_{\pifinput}^{P_1}$ sets $\val = 0$ by assigning $\bv{\val} = \av{\val}$. 
		\begin{inneritemize} 
			\item If dealer is $P_0$, $\Sim_{\pifinput}^{P_1}$ sends $\bv{\val}$ to $\Adv$ on behalf of $P_0$. 
			Steps corresponding to $\piJmpF$ are simulated according to $\Sim_{\piJmpF}^{P_j}$ where $P_1$ acts as one of the sender to send $\bv{\val}$.
			\item If dealer is $P_2$, $\Sim_{\pifinput}^{P_1}$ sends $\bv{\val}$ to $\Adv$ on behalf of $P_2$. 
			Steps corresponding to $\piJmpF$ are simulated according to $\Sim_{\piJmpF}^{P_i}$ where $P_1$ acts as one of the sender to send $\bv{\val} + \gv{\val}$.
		\end{inneritemize} 
		\item[--] If dealer is $P_3$, $\Sim_{\pifinput}^{P_1}$ sets $\val = 0$ by assigning $\bv{\val} = \av{\val}$. 
		Steps corresponding to $\piJmpF$ are simulated according to $\Sim_{\piJmpF}^{P_k}$ where $P_1$ acts as the receiver for receiving $\bv{\val} + \gv{\val}$.
	\end{myitemize}
\end{simulatorsplitbox}

\smallskip
The case for corrupt $P_3$ is given in \boxref{fig:ShFSim3}.
\begin{simulatorbox}{$\Sim_{\pifinput}^{P_3}$}{Simulator $\Sim_{\pifinput}^{P_3}$ for corrupt $P_3$ }{fig:ShFSim3}
	\justify 
	\algoHead{Preprocessing:}
	$\Sim_{\pifinput}^{P_3}$ emulates $\FSETUPf$ and gives the keys $(\Key{03}, \Key{13}, \Key{23}, \Key{013}, \Key{023}, \Key{123}$ and $\Key{\Partyset})$ to $\Adv$. 
	The values that are commonly held with $\Adv$ are sampled using the respective keys, while others are sampled randomly.
	\justify\vspace{-3mm}
	\algoHead{Online:}
	\begin{myitemize}
		\item[--] If dealer is $P_3$, $\Sim_{\pifinput}^{P_3}$ receives $\bv{\val}+\gv{\val}$ from $\Adv$ on behalf of $P_0$. 
		Steps corresponding to $\piJmpF$ are simulated according to $\Sim_{\piJmpF}^{P_i}$ where $P_3$ acts as one of the sender with $P_1$,  $P_2$ as the receivers, separately. 
		\item[--] If dealer is $P_0$ or $P_1$ or $P_2$, steps corresponding to $\piJmpF$ are simulated according to $\Sim_{\piJmpF}^{P_l}$ where $P_3$ acts as the server outside the computation. 
	\end{myitemize}
\end{simulatorbox}

\subsubsection{Multiplication Protocol}
\label{appsec:Fpifmul}
The case for corrupt $P_0$ is given in \boxref{fig:MulFSim0}.
\begin{simulatorsplitbox}{$\Sim_{\pifmul}^{P_0}$}{Simulator $\Sim_{\pifmul}^{P_0}$ for corrupt $P_0$ }{fig:MulFSim0}
	\justify 
	\algoHead{Preprocessing:}
	\begin{myitemize}
		\item[--] $\Sim_{\pifmul}^{P_0}$ samples $\sqr{\av{\wz}}_1, \sqr{\av{\wz}}_2, \GammaA{\wx \wy}$ using the respective keys with $\Adv$. $\Sim_{\pifmul}^{P_0}$ samples $\gv{\wz}, \psi, \vr$ randomly on behalf of  the respective honest parties, and computes $\GammaB{\wx \wy}$ honestly.		  
		\item[--] Steps corresponding to $\piJmpF$ are simulated according to $\Sim_{\piJmpF}^{P_i}$ where $P_0$ acts as one of the sender for sending $\GammaB{\wx \wy}$.
		\item[--] $\Sim_{\pifmul}^{P_0}$ computes $\sqr{\Chi}_1, \sqr{\Chi}_2$ honestly. Steps corresponding to $\piJmpF$ are simulated according to $\Sim_{\piJmpF}^{P_k}$ where $P_0$ acts as the receiver for $\sqr{\Chi}_1, \sqr{\Chi}_2$.
	\end{myitemize}
	\justify\vspace{-3mm}
	\algoHead{Online:}
	\begin{myitemize}
		\item[--] $\Sim_{\pifmul}^{P_0}$ computes $\sqr{\starbeta{\wz}}_1, \sqr{\starbeta{\wz}}_2$ honestly.
		Steps corresponding to $\piJmpF$ are simulated according to $\Sim_{\piJmpF}^{P_j}$ where $P_0$ acts as one of the sender for sending $\sqr{\starbeta{\wz}}_1, \sqr{\starbeta{\wz}}_2$.
		\item[--] $\Sim_{\pifmul}^{P_0}$ computes $\bv{\wz} + \gv{\wz}$.
		Steps corresponding to $\piJmpF$ are simulated according to $\Sim_{\piJmpF}^{P_k}$ where $P_0$ acts as the receiver for receiving $\bv{\wz} + \gv{\wz}$.
	\end{myitemize}
\end{simulatorsplitbox}

\smallskip
The case for corrupt $P_1$ is given in \boxref{fig:MulFSim1}. The case for a corrupt $P_2$ is similar.
\begin{simulatorbox}{$\Sim_{\pifmul}^{P_1}$}{Simulator $\Sim_{\pifmul}^{P_1}$ for corrupt $P_1$ }{fig:MulFSim1}
	\justify 
	\algoHead{Preprocessing:}
	\begin{myitemize}
		\item[--] $\Sim_{\pifmul}^{P_1}$ samples $\sqr{\av{\wz}}_1, \gv{\wz}, \psi, \vr, \GammaA{\wx \wy}$ using the respective keys with $\Adv$. $\Sim_{\pifmul}^{P_1}$ samples $\sqr{\av{\wz}}_2$ randomly on behalf of  the respective honest parties.		  
		\item[--] Steps corresponding to $\piJmpF$ are simulated according to $\Sim_{\piJmpF}^{P_l}$ where $P_1$ acts as the server outside the computation while communicating $\GammaB{\wx \wy}$.
		\item[--] $\Sim_{\pifmul}^{P_1}$ computes $\sqr{\Chi}_1$. Steps corresponding to $\piJmpF$ are simulated according to $\Sim_{\piJmpF}^{P_i}$ where $P_1$ acts as one of the sender for $\sqr{\Chi}_1$.
		\item[--] Steps corresponding to $\piJmpF$ are simulated according to $\Sim_{\piJmpF}^{P_l}$ where $P_1$ acts as the server outside the computation while communicating $\sqr{\Chi}_2$.
	\end{myitemize}
	\justify\vspace{-3mm}
	\algoHead{Online:}
	\begin{myitemize}
		\item[--] $\Sim_{\pifmul}^{P_1}$ computes $\sqr{\starbeta{\wz}}_1, \sqr{\starbeta{\wz}}_2$. 
		Steps corresponding to $\piJmpF$ are simulated according to $\Sim_{\piJmpF}^{P_i}$ and $\Sim_{\piJmpF}^{P_k}$, where $P_1$ acts as one of the sender for sending $\sqr{\starbeta{\wz}}_1$, and $P_1$ acts as the receiver for receiving $\sqr{\starbeta{\wz}}_2$, respectively.
		\item[--] $\Sim_{\pifmul}^{P_1}$ computes $\bv{\wz} + \gv{\wz}$. 
		Steps corresponding to $\piJmpF$ are simulated according to $\Sim_{\piJmpF}^{P_i}$ where $P_1$ acts as one of the sender for sending $\bv{\wz} + \gv{\wz}$. 
	\end{myitemize}
\end{simulatorbox}

\smallskip
The case for corrupt $P_3$ is given in \boxref{fig:MulFSim3}.
\begin{simulatorsplitbox}{$\Sim_{\pifmul}^{P_3}$}{Simulator $\Sim_{\pifmul}^{P_3}$ for corrupt $P_3$ }{fig:MulFSim3}
	\justify
	\algoHead{Preprocessing:}
	\begin{myitemize}
		\item[--] $\Sim_{\pifmul}^{P_3}$ samples $\sqr{\av{\wz}}_1, \sqr{\av{\wz}}_2, \gv{\wz}, \psi, \vr, \GammaA{\wx \wy}$ using the respective keys with $\Adv$. $\Sim_{\pifmul}^{P_3}$ computes $\GammaB{\wx \wy}$ honestly.
		\item[--] Steps corresponding to $\piJmpF$ are simulated according to $\Sim_{\piJmpF}^{P_j}$ where $P_3$ acts as one of the sender for sending $\GammaB{\wx \wy}$.
		\item[--] $\Sim_{\pifmul}^{P_3}$ computes $\sqr{\Chi}_1, \sqr{\Chi}_2$. 
		Steps corresponding to $\piJmpF$ are simulated according to $\Sim_{\piJmpF}^{P_j}$ where $P_3$ acts as one of the sender for sending $\sqr{\Chi}_1$ and $\sqr{\Chi}_2$.
	\end{myitemize}
	\justify\vspace{-3mm} 
	\algoHead{Online:}
	\begin{myitemize}
		\item[--] Steps corresponding to $\piJmpF$ are simulated according to $\Sim_{\piJmpF}^{P_l}$ where $P_3$ acts as the server outside the computation involving $\sqr{\starbeta{\wz}}_1, \sqr{\starbeta{\wz}}_2$ and $\bv{\wz} + \gv{\wz}$. 
	\end{myitemize}
\end{simulatorsplitbox}

\subsubsection{Reconstruction Protocol}
\label{appsec:Fpifrec}
The case for corrupt $P_0$ is given in \boxref{fig:ReconFSim0}. The cases for corrupt $P_1, P_2, P_3$ are similar.
\begin{simulatorbox}{$\Sim_{\pifrec}^{P_0}$}{Simulator $\Sim_{\pifrec}^{P_0}$ for corrupt $P_0$ }{fig:ReconFSim0}
	\medskip
	\justify
	\begin{myitemize}
		\item[--] $\Sim_{\pifrec}^{P_0}$ sends $\gv{\val}$ to $\Adv$ on behalf of $P_1, P_2$, and $\Hash(\gv{\val})$ on behalf of $P_3$, respectively.
		\item[--] $\Sim_{\pifrec}^{P_0}$ receives $\Hash(\sqr{\av{\val}}_1), \Hash( \sqr{\av{\val}}_2), \bv{\val}+\gv{\val}$ from $\Adv$ on behalf of $P_2, P_1, P_3$, respectively.
	\end{myitemize}
\end{simulatorbox}

\vspace{-3mm}
\subsubsection{Joint Sharing Protocol}
\label{appsec:Fpifjsh}
The case for corrupt $P_0$ is given in \boxref{fig:JshFSim0}. 
\begin{simulatorbox}{$\Sim_{\pifjsh}^{P_0}$}{Simulator $\Sim_{\pifjsh}^{P_0}$ for corrupt $P_0$ }{fig:JshFSim0}
	\justify
	\algoHead{Preprocessing:}
	\begin{myitemize}
		\item[--] $\Sim_{\pifjsh}^{P_0}$ has knowledge of $\av{\val}$ and $\gv{\val}$, which it obtains while emulating $\FSETUPf$. The common values shared with the $\Adv$ are sampled using the appropriate shared keys, while other values are sampled at random.
	\end{myitemize}
	\justify\vspace{-3mm}
	\algoHead{Online:}	
	\begin{myitemize}
		\item[--] If dealers are $(P_0, P_1)$: $\Sim_{\pifjsh}^{P_0}$ computes $\bv{\val}$ using $\val$. 
		Steps corresponding to $\piJmpF$ are simulated according to $\Sim_{\piJmpF}^{P_j}$ where $P_0$ acts as one of the sender for $\bv{\val}$. 
		\item[--] If dealers are $(P_0, P_2)$ or $(P_0, P_3)$: Analogous to the above case.
		\item[--] If dealers are $(P_1, P_2)$: $\Sim_{\pifjsh}^{P_0}$ sets $\val = 0$ and $\bv{\val} = \sqr{\av{\val}}_1 + \sqr{\av{\val}}_2$. 
		Steps corresponding to $\piJmpF$ are simulated according to $\Sim_{\piJmpF}^{P_k}$ where $P_0$ acts as the receiver for $\bv{\val} + \gv{\val}$. 
		\item[--] If dealers are $(P_3, P_1)$: $\Sim_{\pifjsh}^{P_0}$ sets $\val = 0$ and $\bv{\val} = \sqr{\av{\val}}_1 + \sqr{\av{\val}}_2$. 
		Steps corresponding to $\piJmpF$ are simulated according to $\Sim_{\piJmpF}^{P_l}$ where $P_0$ acts as the server outside the computation for $\bv{\val}$, and according to $\Sim_{\piJmpF}^{P_k}$ where $P_0$ acts as the receiver for $\bv{\val} + \gv{\val}$. 
		\item[--] If dealers are $(P_3, P_2)$: Analogous to the above case.
	\end{myitemize}
\end{simulatorbox}

\smallskip
The case for corrupt $P_1$ is given in \boxref{fig:JshFSim1}. The case for corrupt $P_2$ is similar.
\begin{simulatorsplitbox}{$\Sim_{\pifjsh}^{P_1}$}{Simulator $\Sim_{\pifjsh}^{P_1}$ for corrupt $P_1$ }{fig:JshFSim1}
	\justify
	\algoHead{Preprocessing:}
	\begin{myitemize}
		\item[--] $\Sim_{\pifjsh}^{P_1}$ has knowledge of $\av{}$-values and $\gv{}$ corresponding to $\val$ which it obtains while emulating $\FSETUPf$. The common values shared with the $\Adv$ are sampled using the appropriate shared keys, while other values are sampled at random.	
	\end{myitemize}
	\justify\vspace{-3mm}
	\algoHead{Online:}	
	\begin{myitemize}
		\item[--] If dealers are $(P_0, P_1)$: $\Sim_{\pifjsh}^{P_1}$ computes $\bv{\val}$ using $\val$.
		Steps corresponding to $\piJmpF$ are simulated according to $\Sim_{\piJmpF}^{P_i}$ where $P_1$ acts as one of the sender for $\bv{\val}$.  
		\item[--] If dealers are $(P_1, P_2)$: Analogous to the previous case, except that now $\bv{\val} + \gv{\val}$ is sent instead of $\bv{\val}$.
		\item[--] If dealers are $(P_3, P_1)$: $\Sim_{\pifjsh}^{P_1}$ computes $\bv{\val}$ and $\bv{\val} + \gv{\val}$. 
		Steps corresponding to $\piJmpF$ are simulated according to $\Sim_{\piJmpF}^{P_i}$ where $P_1$ acts as one of the sender for $\bv{\val}$, $\bv{\val} + \gv{\val}$. 
		\item[--] If dealers are $(P_0, P_2)$ or $(P_0, P_3)$ or $(P_3, P_2)$: $\Sim_{\pifjsh}^{P_1}$ sets $\val = 0$ and $\bv{\val} = \sqr{\av{\val}}_1 + \sqr{\av{\val}}_2$. 
		Steps corresponding to $\piJmpF$ are simulated according to $\Sim_{\piJmpF}^{P_k}$ where $P_1$ acts as the receiver for $\bv{\val}$. 
	\end{myitemize}
\end{simulatorsplitbox}

\smallskip
The case for corrupt $P_3$ is given in \boxref{fig:JshFSim3}. 
\begin{simulatorbox}{$\Sim_{\pifjsh}^{P_3}$}{Simulator $\Sim_{\pifjsh}^{P_3}$ for corrupt $P_3$ }{fig:JshFSim3}
	\justify
	\algoHead{Preprocessing:}
	\begin{myitemize}
		\item[--] $\Sim_{\pifjsh}^{P_3}$ has knowledge of $\av{}$-values and $\gv{}$ corresponding to $\val$ which it obtains while emulating $\FSETUPf$. The common values shared with the $\Adv$ are sampled using the appropriate shared keys, while other values are sampled at random.
	\end{myitemize}
	\justify\vspace{-3mm}
	\algoHead{Online:}	
	\begin{myitemize}
		\item[--] If dealers are $(P_1, P_2)$: $\Sim_{\pifjsh}^{P_3}$ sets $\val = 0$. 
		Steps corresponding to $\piJmpF$ are simulated according to $\Sim_{\piJmpF}^{P_l}$ where $P_3$ acts as the server outside the computation for $\bv{\val} + \gv{\val}$. 
		\item[--] If dealers are $(P_0, P_1)$ or $(P_0, P_2)$: Analogous to the above case. 
		\item[--] If dealers are $(P_0, P_3)$: $\Sim_{\pifjsh}^{P_3}$ computes $\bv{\val}$ using $\val$. 
		Steps corresponding to $\piJmpF$ are simulated according to $\Sim_{\piJmpF}^{P_i}$ where $P_3$ acts as one of the sender for sending $\bv{\val}$. 
		\item[--] If dealers are $(P_3, P_1)$: $\Sim_{\pifjsh}^{P_3}$ computes $\bv{\val}$ and $\bv{\val} + \gv{\val}$. 
		Steps corresponding to $\piJmpF$ are simulated according to $\Sim_{\piJmpF}^{P_j}$ where $P_3$ acts as one of the sender for sending $\bv{\val}, \bv{\val} + \gv{\val}$.
		\item[--] If dealers are $(P_3, P_2)$: Analogous to the above case.
	\end{myitemize}
\end{simulatorbox}

\vspace{-3mm}
\subsubsection{Dot Product Protocol}
\label{appsec:Fpidotp}
The case for corrupt $P_0$ is given in \boxref{fig:DotpFSim0}.
\begin{simulatorsplitbox}{$\Sim_{\pifdotp}^{P_0}$}{Simulator $\Sim_{\pifdotp}^{P_0}$ for corrupt $P_0$ }{fig:DotpFSim0}
	\justify
	\algoHead{Preprocessing:}
	\begin{myitemize}
		\item[--] $\Sim_{\pifdotp}^{P_0}$ samples $\sqr{\av{\wz}}_1, \sqr{\av{\wz}}_2, \sqrA{\Gammaxdy}$ using the respective keys with $\Adv$. $\Sim_{\pifdotp}^{P_0}$ samples $\gv{\wz}, \psi, \vr$ randomly on behalf of  the respective honest parties, and computes $\sqr{\Gammaxdy}_2$ honestly.		  
		\item[--] Steps corresponding to $\piJmpF$ are simulated according to $\Sim_{\piJmpF}^{P_i}$ where $P_0$ acts as one of the sender for $\sqr{\Gammaxdy}_2$.
		\item[--] $\Sim_{\pifdotp}^{P_0}$ computes $\Chi_1, \Chi_2$ honestly.
		Steps corresponding to $\piJmpF$ are simulated according to $\Sim_{\piJmpF}^{P_k}$ where $P_0$ acts as the receiver for $\Chi_1$ and $\Chi_2$.
	\end{myitemize}
	\justify\vspace{-3mm} 
	\algoHead{Online:}
	\begin{myitemize}
		\item[--] $\Sim_{\pifdotp}^{P_0}$ computes $\sqr{\starbeta{\wz}}_1, \sqr{\starbeta{\wz}}_2$ honestly. 
		Steps corresponding to $\piJmpF$ are simulated according to $\Sim_{\piJmpF}^{P_j}$ where $P_0$ acts as one of the sender for $\sqr{\starbeta{\wz}}_1, \sqr{\starbeta{\wz}}_2$.
		\item[--] $\Sim_{\pifdotp}^{P_0}$ computes $\bv{\wz} + \gv{\wz}$. 
		Steps corresponding to $\piJmpF$ are simulated according to $\Sim_{\piJmpF}^{P_k}$ where $P_0$ acts as the receiver for $\bv{\wz} + \gv{\wz}$.
	\end{myitemize}
\end{simulatorsplitbox}

\vspace{-3mm}
The case for corrupt $P_1$ is given in \boxref{fig:DotpFSim1}. The case for corrupt $P_2$ is similar.

\begin{simulatorbox}{$\Sim_{\pifdotp}^{P_1}$}{Simulator $\Sim_{\pifdotp}^{P_1}$ for corrupt $P_1$ }{fig:DotpFSim1}
	\justify
	\algoHead{Preprocessing:}
	\begin{myitemize}
		\item[--] $\Sim_{\pifdotp}^{P_1}$ samples $\sqr{\av{\wz}}_1, \gv{\wz}, \psi, \vr, \sqrA{\Gammaxdy}$ using the respective keys with $\Adv$. $\Sim_{\pifdotp}^{P_1}$ samples $\sqr{\av{\wz}}_2$ randomly on behalf of  the respective honest parties.		  
		\item[--] Steps corresponding to $\piJmpF$ are simulated according to $\Sim_{\piJmpF}^{P_l}$ where $P_1$ acts as the server outside the computation for $\sqrB{\Gammaxdy}$.
		\item[--] $\Sim_{\pifdotp}^{P_1}$ computes $\Chi_1$. 
		Steps corresponding to $\piJmpF$ are simulated according to $\Sim_{\piJmpF}^{P_i}$ where $P_1$ acts as one of the sender for $\Chi_1$.
		\item[--] Steps corresponding to $\piJmpF$ are simulated according to $\Sim_{\piJmpF}^{P_l}$ where $P_1$ acts as the server outside the computation for $\Chi_2$.
	\end{myitemize}
	\justify\vspace{-3mm}
	\algoHead{Online:}
	\begin{myitemize}
		\item[--] $\Sim_{\pifdotp}^{P_1}$ computes $\sqr{\starbeta{\wz}}_1, \sqr{\starbeta{\wz}}_2$. 
		Steps corresponding to $\piJmpF$ are simulated according to $\Sim_{\piJmpF}^{P_i}$ and $\Sim_{\piJmpF}^{P_k}$, where $P_1$ acts as one of the sender for $\sqr{\starbeta{\wz}}_1$, and $P_1$ acts as the receiver for $\sqr{\starbeta{\wz}}_2$.
		\item[--] $\Sim_{\pifdotp}^{P_1}$ computes $\bv{\wz} + \gv{\wz}$. 
		Steps corresponding to $\piJmpF$ are simulated according to $\Sim_{\piJmpF}^{P_i}$ where $P_1$ acts as one of the sender for $\bv{\wz} + \gv{\wz}$.
	\end{myitemize}
\end{simulatorbox}

The case for corrupt $P_3$ is given in \boxref{fig:DotpFSim3}.
\begin{simulatorsplitbox}{$\Sim_{\pifdotp}^{P_3}$}{Simulator $\Sim_{\pifdotp}^{P_3}$ for corrupt $P_3$ }{fig:DotpFSim3}
	\justify
	\algoHead{Preprocessing:}
	\begin{myitemize}
		\item[--] $\Sim_{\pifdotp}^{P_3}$ samples $\sqr{\av{\wz}}_1, \sqr{\av{\wz}}_2, \gv{\wz}, \psi, \vr, \sqrA{\Gammaxdy}$ using the respective keys with $\Adv$. $\Sim_{\pifdotp}^{P_3}$ computes $\sqr{\Gammaxdy}$ honestly.				  
		\item[--] Steps corresponding to $\piJmpF$ are simulated according to $\Sim_{\piJmpF}^{P_j}$ where $P_3$ acts as one of the sender for $\sqrB{\Gammaxdy}$.  
		\item[--] $\Sim_{\pifdotp}^{P_3}$ computes $\Chi_1, \Chi_2$. 
		Steps corresponding to $\piJmpF$ are simulated according to $\Sim_{\piJmpF}^{P_j}$ where $P_3$ acts as one of the sender for $\Chi_1, \Chi_2$.
	\end{myitemize}
	\justify\vspace{-3mm}
	\algoHead{Online:}
	\begin{myitemize}
		\item[--] Steps corresponding to $\piJmpF$ are simulated according to $\Sim_{\piJmpF}^{P_l}$ where $P_3$ acts as the server outside the computation for $\sqr{\starbeta{\wz}}_1, \sqr{\starbeta{\wz}}_2$, $\bv{\wz} + \gv{\wz}$. 
	\end{myitemize}
\end{simulatorsplitbox}

\subsubsection{~Truncation Pair Generation}
\label{appsec:Fpiftrgen}
Here we give the simulation steps for $\piftrgen$.
The case for corrupt $P_0$ is given in \boxref{fig:TrgenFSim0}. The case for corrupt $P_3$ is similar. 
\begin{simulatorbox}{$\Sim_{\piftrgen}^{P_0}$}{Simulator $\Sim_{\piftrgen}^{P_0}$ for corrupt $P_0$ }{fig:TrgenFSim0}
	\medskip
	\justify 
	\begin{myitemize}
		\item[--] $\Sim_{\piftrgen}^{P_0}$ samples $R_1, R_2$ using the respective keys with $\Adv$.
		\item[--] Steps corresponding to $\pifjsh$ are simulated according to $\Sim_{\pifjsh}^{P_0}$ (\boxref{fig:JshFSim0}).
	\end{myitemize}
\end{simulatorbox}

The case for corrupt $P_1$ is given in \boxref{fig:TrgenFSim1}. The case for corrupt $P_2$ is similar. 
\begin{simulatorbox}{$\Sim_{\piftrgen}^{P_1}$}{Simulator $\Sim_{\piftrgen}^{P_1}$ for corrupt $P_1$ }{fig:TrgenFSim1}
	\medskip
	\justify 
	\begin{myitemize}
		\item[--] $\Sim_{\piftrgen}^{P_1}$ samples $R_1$ using the respective keys with $\Adv$, and samples $R_2$ randomly.
		\item[--] Steps corresponding to $\pifjsh$ are simulated according to $\Sim_{\pifjsh}^{P_1}$ (\boxref{fig:JshFSim1}).
	\end{myitemize}
\end{simulatorbox}

%% file: main.bbl
\begin{thebibliography}{10}

\bibitem{AbspoelDEN19}
M.~Abspoel, A.~P.~K. Dalskov, D.~Escudero, and A.~Nof.
\newblock An efficient passive-to-active compiler for honest-majority {MPC}
  over rings.
\newblock In {\em {ACNS}}, 2021.

\bibitem{alon20}
B.~Alon, E.~Omri, and A.~Paskin-Cherniavsky.
\newblock {MPC with Friends and Foes}.
\newblock In {\em {CRYPTO}}, pages 677--706, 2020.

\bibitem{ABFLLNOWW17}
T.~Araki, A.~Barak, J.~Furukawa, T.~Lichter, Y.~Lindell, A.~Nof, K.~Ohara,
  A.~Watzman, and O.~Weinstein.
\newblock Optimized honest-majority {MPC} for malicious adversaries - breaking
  the 1 billion-gate per second barrier.
\newblock In {\em {IEEE S{\&}P}}, pages 843--862, 2017.

\bibitem{AFLNO16}
T.~Araki, J.~Furukawa, Y.~Lindell, A.~Nof, and K.~Ohara.
\newblock High-throughput semi-honest secure three-party computation with an
  honest majority.
\newblock In {\em {ACM CCS}}, pages 805--817, 2016.

\bibitem{AsharovL17}
G.~Asharov and Y.~Lindell.
\newblock A full proof of the {BGW} protocol for perfectly secure multiparty
  computation.
\newblock {\em J. Cryptology}, pages 58--151, 2017.

\bibitem{BaumDTZ16}
C.~Baum, I.~Damg{\aa}rd, T.~Toft, and R.~W. Zakarias.
\newblock Better preprocessing for secure multiparty computation.
\newblock In {\em {ACNS}}, pages 327--345, 2016.

\bibitem{BogdanovLW08}
D.~Bogdanov, S.~Laur, and J.~Willemson.
\newblock Sharemind: {A} framework for fast privacy-preserving computations.
\newblock In {\em {ESORICS}}, pages 192--206, 2008.

\bibitem{auction}
P.~Bogetoft, D.~L. Christensen, I.~Damg{\aa}rd, M.~Geisler, T.~Jakobsen,
  M.~Kr{\o}igaard, J.~D. Nielsen, J.~B. Nielsen, K.~Nielsen, J.~Pagter, et~al.
\newblock Secure multiparty computation goes live.
\newblock In {\em {FC}}, pages 325--343, 2009.

\bibitem{BonehBCGI19}
D.~Boneh, E.~Boyle, H.~Corrigan{-}Gibbs, N.~Gilboa, and Y.~Ishai.
\newblock Zero-knowledge proofs on secret-shared data via fully linear pcps.
\newblock In {\em {CRYPTO}}, pages 67--97, 2019.

\bibitem{BGIN19}
E.~Boyle, N.~Gilboa, Y.~Ishai, and A.~Nof.
\newblock Practical fully secure three-party computation via sublinear
  distributed zero-knowledge proofs.
\newblock In {\em {ACM CCS}}, pages 869--886, 2019.

\bibitem{BunnO07}
P.~Bunn and R.~Ostrovsky.
\newblock Secure two-party k-means clustering.
\newblock In {\em {ACM} {CCS}}, pages 486--497, 2007.

\bibitem{FLASH}
M.~Byali, H.~Chaudhari, A.~Patra, and A.~Suresh.
\newblock {FLASH:} fast and robust framework for privacy-preserving machine
  learning.
\newblock {\em {PETS}}, 2020.

\bibitem{ByaliHPS19}
M.~Byali, C.~Hazay, A.~Patra, and S.~Singla.
\newblock Fast actively secure five-party computation with security beyond
  abort.
\newblock In {\em {ACM} {CCS}}, pages 1573--1590, 2019.

\bibitem{ByaliJPR18}
M.~Byali, A.~Joseph, A.~Patra, and D.~Ravi.
\newblock Fast secure computation for small population over the internet.
\newblock In {\em {ACM CCS}}, pages 677--694, 2018.

\bibitem{ASTRA}
H.~Chaudhari, A.~Choudhury, A.~Patra, and A.~Suresh.
\newblock {ASTRA: High Throughput 3PC over Rings with Application to Secure
  Prediction}.
\newblock In {\em {ACM} {CCSW@CCS}}, 2019.

\bibitem{Trident}
H.~Chaudhari, R.~Rachuri, and A.~Suresh.
\newblock {Trident: Efficient 4PC Framework for Privacy Preserving Machine
  Learning}.
\newblock {\em {NDSS}}, 2020.

\bibitem{CGHIKLN18}
K.~Chida, D.~Genkin, K.~Hamada, D.~Ikarashi, R.~Kikuchi, Y.~Lindell, and
  A.~Nof.
\newblock Fast large-scale honest-majority {MPC} for malicious adversaries.
\newblock In {\em {CRYPTO}}, pages 34--64, 2018.

\bibitem{Cleve86}
R.~Cleve.
\newblock Limits on the security of coin flips when half the processors are
  faulty (extended abstract).
\newblock In {\em {ACM} {STOC}}, pages 364--369, 1986.

\bibitem{CHOR18}
R.~Cohen, I.~Haitner, E.~Omri, and L.~Rotem.
\newblock Characterization of secure multiparty computation without broadcast.
\newblock {\em J. Cryptology}, pages 587--609, 2018.

\bibitem{CramerDESX18}
R.~Cramer, I.~Damg{\aa}rd, D.~Escudero, P.~Scholl, and C.~Xing.
\newblock Spd\(\mathbb{Z}\)\({}_{\mbox{2\({}^{\mbox{k}}\)}}\): Efficient {MPC}
  mod 2\({}^{\mbox{k}}\) for dishonest majority.
\newblock In {\em {CRYPTO}}, pages 769--798, 2018.

\bibitem{CramerFIK03}
R.~Cramer, S.~Fehr, Y.~Ishai, and E.~Kushilevitz.
\newblock Efficient multi-party computation over rings.
\newblock In {\em {EUROCRYPT}}, pages 596--613, 2003.

\bibitem{ENCRYPTO}
Cryptography and P.~E.~G. at~TU~Darmstadt.
\newblock {ENCRYPTO Utils}.
\newblock \url{https://github.com/encryptogroup/ENCRYPTO_utils}.

\bibitem{DEK20}
A.~Dalskov, D.~Escudero, and M.~Keller.
\newblock {Fantastic Four: Honest-Majority Four-Party Secure Computation With
  Malicious Security}.
\newblock Cryptology ePrint Archive, 2020.
\newblock \url{https://eprint.iacr.org/2020/1330}.

\bibitem{Damgard0FKSV19}
I.~Damg{\aa}rd, D.~Escudero, T.~K. Frederiksen, M.~Keller, P.~Scholl, and
  N.~Volgushev.
\newblock New primitives for actively-secure {MPC} over rings with applications
  to private machine learning.
\newblock {\em {IEEE S{\&}P}}, 2019.

\bibitem{SPDZ2}
I.~Damg{\aa}rd, M.~Keller, E.~Larraia, V.~Pastro, P.~Scholl, and N.~P. Smart.
\newblock Practical covertly secure {MPC} for dishonest majority - or: Breaking
  the {SPDZ} limits.
\newblock In {\em {ESORICS}}, pages 1--18, 2013.

\bibitem{DamgardOS18}
I.~Damg{\aa}rd, C.~Orlandi, and M.~Simkin.
\newblock Yet another compiler for active security or: Efficient {MPC} over
  arbitrary rings.
\newblock In {\em {CRYPTO}}, pages 799--829, 2018.

\bibitem{DPSZ12}
I.~Damg{\aa}rd, V.~Pastro, N.~P. Smart, and S.~Zakarias.
\newblock Multiparty computation from somewhat homomorphic encryption.
\newblock In {\em {CRYPTO}}, pages 643--662, 2012.

\bibitem{DSZ15}
D.~Demmler, T.~Schneider, and M.~Zohner.
\newblock {ABY} - {A} framework for efficient mixed-protocol secure two-party
  computation.
\newblock In {\em {NDSS}}, 2015.

\bibitem{DuA01}
W.~Du and M.~J. Atallah.
\newblock Privacy-preserving cooperative scientific computations.
\newblock In {\em {IEEE} {CSFW-14}}, pages 273--294, 2001.

\bibitem{EOP19}
H.~Eerikson, M.~Keller, C.~Orlandi, P.~Pullonen, J.~Puura, and M.~Simkin.
\newblock {Use Your Brain! Arithmetic 3PC for Any Modulus with Active
  Security}.
\newblock In {\em {ITC}}, 2020.

\bibitem{FLNW17}
J.~Furukawa, Y.~Lindell, A.~Nof, and O.~Weinstein.
\newblock High-throughput secure three-party computation for malicious
  adversaries and an honest majority.
\newblock In {\em {EUROCRYPT}}, pages 225--255, 2017.

\bibitem{GordonR018}
S.~D. Gordon, S.~Ranellucci, and X.~Wang.
\newblock Secure computation with low communication from cross-checking.
\newblock In {\em {ASIACRYPT}}, pages 59--85, 2018.

\bibitem{JagannathanW05}
G.~Jagannathan and R.~N. Wright.
\newblock Privacy-preserving distributed k-means clustering over arbitrarily
  partitioned data.
\newblock In {\em {ACM} {SIGKDD}}, pages 593--599, 2005.

\bibitem{KOS16}
M.~Keller, E.~Orsini, and P.~Scholl.
\newblock {MASCOT:} faster malicious arithmetic secure computation with
  oblivious transfer.
\newblock In {\em {ACM CCS}}, pages 830--842, 2016.

\bibitem{KellerPR18}
M.~Keller, V.~Pastro, and D.~Rotaru.
\newblock Overdrive: Making {SPDZ} great again.
\newblock In {\em {EUROCRYPT}}, pages 158--189, 2018.

\bibitem{SPDZ3}
M.~Keller, P.~Scholl, and N.~P. Smart.
\newblock An architecture for practical actively secure {MPC} with dishonest
  majority.
\newblock In {\em {ACM CCS}}, pages 549--560, 2013.

\bibitem{CIFAR10}
A.~Krizhevsky, V.~Nair, and G.~Hinton.
\newblock The {CIFAR}-10 dataset.
\newblock 2014.
\newblock \url{https://www.cs.toronto.edu/~kriz/cifar.html}.

\bibitem{lenet}
Y.~LeCun, L.~Bottou, Y.~Bengio, and P.~Haffner.
\newblock Gradient-based learning applied to document recognition.
\newblock {\em Proceedings of the IEEE}, pages 2278--2324, 1998.

\bibitem{MNIST10}
Y.~LeCun and C.~Cortes.
\newblock {MNIST} handwritten digit database.
\newblock 2010.
\newblock \url{http://yann.lecun.com/exdb/mnist/}.

\bibitem{LindellP02}
Y.~Lindell and B.~Pinkas.
\newblock Privacy preserving data mining.
\newblock {\em J. Cryptology}, pages 177--206, 2002.

\bibitem{MakriRSV19}
E.~Makri, D.~Rotaru, N.~P. Smart, and F.~Vercauteren.
\newblock {EPIC:} efficient private image classification (or: Learning from the
  masters).
\newblock In {\em {CT-RSA}}, pages 473--492, 2019.

\bibitem{MazloomLRG20}
S.~Mazloom, P.~H. Le, S.~Ranellucci, and S.~D. Gordon.
\newblock Secure parallel computation on national scale volumes of data.
\newblock In {\em {USENIX}}, pages 2487--2504, 2020.

\bibitem{MR18}
P.~Mohassel and P.~Rindal.
\newblock {ABY}\({}^{\mbox{3}}\): {A} mixed protocol framework for machine
  learning.
\newblock In {\em {ACM} {CCS}}, pages 35--52, 2018.

\bibitem{MRZ15}
P.~Mohassel, M.~Rosulek, and Y.~Zhang.
\newblock Fast and secure three-party computation: The garbled circuit
  approach.
\newblock In {\em {ACM CCS}}, pages 591--602, 2015.

\bibitem{MohasselZ17}
P.~Mohassel and Y.~Zhang.
\newblock Secureml: {A} system for scalable privacy-preserving machine
  learning.
\newblock In {\em {IEEE S{\&}P}}, pages 19--38, 2017.

\bibitem{NV18}
P.~S. Nordholt and M.~Veeningen.
\newblock Minimising communication in honest-majority {MPC} by batchwise
  multiplication verification.
\newblock In {\em {ACNS}}, pages 321--339, 2018.

\bibitem{PatraR18}
A.~Patra and D.~Ravi.
\newblock On the exact round complexity of secure three-party computation.
\newblock In {\em {CRYPTO}}, pages 425--458, 2018.

\bibitem{BLAZE}
A.~Patra and A.~Suresh.
\newblock {BLAZE: Blazing Fast Privacy-Preserving Machine Learning}.
\newblock {\em {NDSS}}, 2020.
\newblock \url{https://eprint.iacr.org/2020/042}.

\bibitem{PeaseSL80}
M.~C. Pease, R.~E. Shostak, and L.~Lamport.
\newblock Reaching agreement in the presence of faults.
\newblock {\em J. {ACM}}, pages 228--234, 1980.

\bibitem{RiaziWTS0K18}
M.~S. Riazi, C.~Weinert, O.~Tkachenko, E.~M. Songhori, T.~Schneider, and
  F.~Koushanfar.
\newblock Chameleon: {A} hybrid secure computation framework for machine
  learning applications.
\newblock In {\em {AsiaCCS}}, pages 707--721, 2018.

\bibitem{SanilKLR04}
A.~P. Sanil, A.~F. Karr, X.~Lin, and J.~P. Reiter.
\newblock Privacy preserving regression modelling via distributed computation.
\newblock In {\em {ACM} {SIGKDD}}, pages 677--682, 2004.

\bibitem{vgg16}
K.~Simonyan and A.~Zisserman.
\newblock Very deep convolutional networks for large-scale image recognition.
\newblock {\em arXiv preprint arXiv:1409.1556}, 2014.

\bibitem{SlavkovicNT07}
A.~B. Slavkovic, Y.~Nardi, and M.~M. Tibbits.
\newblock Secure logistic regression of horizontally and vertically partitioned
  distributed databases.
\newblock In {\em {ICDM}}, pages 723--728, 2007.

\bibitem{ConvStanford}
Stanford.
\newblock {CS231n: Convolutional Neural Networks for Visual Recognition}.

\bibitem{VaidyaYJ08}
J.~Vaidya, H.~Yu, and X.~Jiang.
\newblock Privacy-preserving {SVM} classification.
\newblock {\em Knowl. Inf. Syst.}, pages 161--178, 2008.

\bibitem{WaghGC19}
S.~Wagh, D.~Gupta, and N.~Chandran.
\newblock Securenn: 3-party secure computation for neural network training.
\newblock {\em PoPETs}, pages 26--49, 2019.

\bibitem{Falcon}
S.~Wagh, S.~Tople, F.~Benhamouda, E.~Kushilevitz, P.~Mittal, and T.~Rabin.
\newblock {FALCON: Honest-Majority Maliciously Secure Framework for Private
  Deep Learning}.
\newblock {\em {PoPETS}}, pages 188--208, 2021.
\newblock \url{https://arxiv.org/abs/2004.02229v1}.

\bibitem{YuVJ06}
H.~Yu, J.~Vaidya, and X.~Jiang.
\newblock Privacy-preserving {SVM} classification on vertically partitioned
  data.
\newblock In {\em {PAKDD}}, pages 647--656, 2006.

\end{thebibliography}
